%% file: main.tex
\documentclass[11pt]{article}

\usepackage{fancyhdr}
\renewcommand{\footnotesize}{\fontsize{8}{12}\selectfont}

\usepackage{amsthm}

\newtheorem{@theorem}{Theorem}[section]
\newenvironment{theorem}{\begin{@theorem}}{\end{@theorem}}

\newtheorem{lemma}{Lemma}[section]
\newtheorem{corollary}{Corollary}[section]
\newtheorem{claim}{Claim}[section]
\newtheorem{aproblem}{Problem}[section]
\newtheorem*{Remark}{Remark}
\newtheorem{definition}{Definition}[section]

\usepackage{setspace}
\usepackage[a4paper, margin=1in]{geometry}

\usepackage{hyperref}

\usepackage[
ruled,vlined,boxed,algo2e]{algorithm2e}

\SetAlCapNameFnt{\small}
\SetAlCapFnt{\small}
\SetAlFnt{\small}

\usepackage{comment}
\usepackage{tabularx}
\usepackage{amsmath} 

\usepackage{tikz}
\usetikzlibrary{decorations.pathreplacing}
\usepackage{pgfplots}
\pgfplotsset{compat=1.18}
\usepackage{caption}
\usepackage{enumitem}

\input{macros}

\usepackage{multicol}

\usepackage{float} 
\usepackage{xcolor}

\begin{document}
\title{New algorithms for girth and cycle detection}
\author{Liam Roditty \and Plia Trabelsi}

\author{
Liam Roditty\thanks{Department of Computer Science, Bar Ilan University, Ramat Gan 5290002, Israel. E-mail {\tt liam.roditty@biu.ac.il}. Supported in part by BSF grants 2016365 and 2020356.}
\and 
Plia Trabelsi\thanks{Department of Computer Science, Bar Ilan University, Ramat Gan 5290002, Israel. E-mail {\tt plia.trabelsi@gmail.com}.}
}

\date{}
\maketitle

\usetikzlibrary{patterns}
\maketitle

\begin{abstract}
Let $G=(V,E)$ be an unweighted undirected graph with $n$ vertices and $m$ edges. Let $g$ be the
girth of $G$, that is, the length of a shortest cycle in $G$. 
We present a randomized algorithm with a running time of 
$\tilde{O}\big(\Ynew \cdot n^{1 + \frac{1}{\Ynew - \varepsilon}}\big)$ that returns a cycle of length at most
$
2\Ynew \left\lceil \frac{g}{2} \right\rceil - 2 \left\lfloor \varepsilon \left\lceil \frac{g}{2} \right\rceil \right\rfloor
$,
where $\Ynew \geq 2$ is an integer and $\varepsilon \in [0,1]$, for every graph with $g = \polylog(n)$.

Our algorithm generalizes an algorithm  of  Kadria \etal{} [SODA'22] that computes a cycle of length at most
$4\left\lceil \frac{g}{2} \right\rceil - 2\left\lfloor \varepsilon \left\lceil \frac{g}{2} \right\rceil \right\rfloor
$
in $\tilde{O}\big(n^{1 + \frac{1}{2 - \varepsilon}}\big)$ time.
Kadria \etal{} presented also  an algorithm that finds a cycle of length at most 
$
2\ell \left\lceil \frac{g}{2} \right\rceil
$
in $\tilde{O}\big(n^{1 + \frac{1}{\ell}}\big)$ time, where $\ell$ must be an integer. 
Our algorithm generalizes this algorithm, as well,  by replacing the integer parameter~$\ell$ in the running time exponent with a real-valued parameter~$\Ynew - \varepsilon$, thereby offering greater flexibility in parameter selection and enabling a broader spectrum of combinations between running times and cycle lengths.

We also show that for sparse graphs a better tradeoff is possible, by presenting an $\Ot(\Ynew\cdot m^{1+\frac{1}{\Ynew-\varepsilon}})$ time randomized algorithm that returns a cycle of length at most $2\Ynew(\lfloor \frac{g-1}{2}\rfloor)  - 2(\lfloor \varepsilon \lfloor \frac{g-1}{2}\rfloor \rfloor+1)$, where $\Ynew\geq 3$ is an integer and $\varepsilon\in [0,1)$,  for every graph with $g=\polylog(n)$.

To obtain our algorithms we develop several techniques and introduce a formal definition of \textit{hybrid} cycle detection algorithms.
Both may prove useful in broader contexts, including other cycle detection and approximation problems. 
Among our techniques is a new  cycle searching technique, in which we search for a cycle from a given vertex and possibly all its neighbors in linear time. Using this technique together with more ideas we develop two hybrid algorithms. The first  allows us to obtain an $\Ot(m^{2-\frac{2}{\lceil g/2  \rceil+1}})$-time,
$(+1)$-approximation of $g$. The second  is used to obtain our   $\Ot(\Ynew\cdot n^{1+\frac{1}{\Ynew-\varepsilon}})$-time and $\Ot(\Ynew\cdot m^{1+\frac{1}{\Ynew-\varepsilon}})$-time approximation algorithms.

\end{abstract}

\newpage

\section{Introduction}\label{sec:introduction}

Let $G=(V,E)$ be an unweighted undirected graph with $n$ vertices and $m$ edges. A set of vertices $C_\ell=\{ v_1,v_2,\cdots,v_{\ell+1}\}$ in $G$,  where $\ell \geq 2$,  is a cycle of length $\ell$ if  $v_1=v_{\ell+1}$ and 
$(v_i,v_{i+1})\in E$, for every $1\leq i \leq \ell$.  A  $C_{\leq \ell}$ is a cycle of length at most $\ell$. 
The \textit{girth} $g$ of $G$ is the length of a shortest cycle in $G$. 
The girth of a graph has been studied extensively since the 1970s by researchers from both the graph theory and the algorithms communities.

Itai and Rodeh~\cite{itai1977finding} showed that the girth can be computed in $O(mn)$ time or in $O(n^\omega)$ time, where $\omega< 2.371552$~\cite{williams2024new}, if Fast Matrix Multiplication (FMM) algorithms are used. 
They also proved that the problem of computing the girth is equivalent to the problem of deciding whether there is a $C_3$ (triangle) in a graph or not. 

In practice, algorithms for FMM
have very large constant factors in their running time. Combinatorial algorithms,
informally, are algorithms which do not use algebraic methods that are being used by FMM algorithms, and
consequently are often more practical.
Vassilevska W. and Williams~\cite{WilliamsW18} showed that if there exists a truly subcubic time\footnote{$O(n^{3-c})$ time for a constant $c > 0$.} combinatorial algorithm
which detects if a graph has a triangle (and therefore also a subcubic time algorithm that computes the exact girth), then there exists a truly subcubic time combinatorial algorithm for Boolean Matrix Multiplication (BMM)
(and therefore also for unweighted All Pairs Shortest Path (APSP), see \cite{galil1997all}, \cite{seidel1995all}, \cite{shoshan1999all}). 
 Such an algorithm
 would be a major breakthrough. 
As a result, to get a faster running time for computing the girth, it is natural to settle for an \textit{approximation} algorithm  for  the girth instead of an exact computation.
An $(\alpha,\beta)$-approximation $\hat{g}$ of $g$ (where $\alpha \geq 1$ and $\beta \geq 0$), satisfies $g\leq \hat{g} \leq \alpha \cdot g+\beta$. We denote an approximation as an $\alpha$-approximation 
if $\beta=0$ and as a $(+\beta)$-approximation
if $\alpha=1$. 

Itai and Rodeh~\cite{itai1977finding} presented a  $(+1)$-approximation algorithm that runs in $O(n^2)$ time.
Notice that in contrast to the BMM or APSP problems, where a running time of $\Omega(n^2)$ is inevitable since  the output size is $\Omega(n^2)$, in the girth problem the output is a single number, thus, there is no natural barrier for subquadratic time algorithms. 
Indeed, Lingas and Lundell~\cite{lingas2009efficient} presented a $\frac{8}{3}$-approximation algorithm that runs in $\Ot(n^{3/2})$ time, and 
Roditty and V. Williams~\cite{roditty2012subquadratic} presented a $2$-approximation algorithm that runs in $\Ot(n^{5/3})$ time. 
Dahlgaard, Knudsen and St\"{o}ckel~\cite{dahlgaard2017new} presented  two  tradeoffs  
between  running time and approximation. One generalizes the algorithms of \cite{lingas2009efficient,roditty2012subquadratic} and
computes a cycle of length at most $2\lceil \frac{g}{2}\rceil+2\lceil \frac{g}{2(\ell-1)}\rceil$
in $\Ot(n^{2-1/\ell})$  time.
The other computes, whp, a $C_{\leq 2^\ell g}$, for any integer $\ell\geq 2$, in $\Ot(n^{1+1/\ell})$  time.

Kadria \etal{}~\cite{kadria2022algorithmic} significantly improved upon the second algorithm of~\cite{dahlgaard2017new} and  presented an algorithm, that for every integer $\ell \geq 1$, 
computes a $C_{\leq 2 \ell \lceil g / 2\rceil}$ in $\Ot(n^{1+1 / \ell})$ time. 
They also presented an algorithm, that  for every $\varepsilon\in (0,1)$, computes a cycle of length at most $4\lceil \frac{g}{2}\rceil  - 2\lfloor \varepsilon \lceil \frac{g}{2}\rceil \rfloor \leq (2-\varepsilon)g+4$, in $\tO(n^{1+1/(2-\varepsilon)})$ time, for every graph with $g=\polylog(n)$. 

These two algorithms of Kadria \etal{}, as well as few other approximation algorithms (see for example, \cite{lingas2009efficient}, \cite{ChechikLRS20}, \cite{RodittyT13}), were obtained using a general framework for girth approximation in which a search is performed over the range of possible values of $g$, using some algorithm $\A$ that gets as an input an integer $\tilde{g}$ which is a guess for the value of $g$. In each step of the  search $\A$ either returns a cycle  $C_{\leq f(\tilde{g})}$, where $f$ is a non decreasing function,  or  determines that    $g>\tilde{g}$. 
The goal of the search is to find the smallest $\tilde{g}$, for which  $\A$ returns  a cycle,
because  for this value  we have $g> \tilde{g}-1$ (and thus $g\geq \tilde{g}$), and algorithm $\A$  returns a  $C_{\leq f(\tilde{g})}$. 
This cycle is of length at most $f(g)$ 
since  $g\geq \tilde{g}$ and $f$ is a non decreasing function ($f$ can represent the approximation, for example $f(\tilde{g})=2\tilde{g}$ yields a $2$-approximation). 
The two possible outcomes of $\A$ and its usage in the general girth approximation framework inspired us to formally define 
the notion of a $(\gamma,\delta)$-\textit{hybrid} algorithm as follows:

\newpage
\begin{definition}
A $(\gamma,\delta)$-hybrid algorithm is an algorithm that either outputs a $C_{\leq \gamma}$ or determines that $g>\delta$.
\end{definition}

When $\gamma = \delta$, the algorithm is referred to as a $\gamma$-hybrid algorithm.
The girth approximation framework described above suggests that a possible approach for developing efficient  girth approximation algorithms is  by developing efficient  $(\gamma,\delta)$-hybrid algorithms. 

Kadria \etal{}~\cite{kadria2022algorithmic} 
designed several algorithms that satisfy the definition of $(\gamma, \delta)$-hybrid algorithms. 
Their girth approximation algorithms mentioned above were obtained using two different $(f(\tilde{g}), \tilde{g})$-hybrid algorithms. 
Additionally, for  every integer $k\geq 2$, they presented a $(2k,3)$-hybrid and a $(2k,4)$-hybrid algorithms that run in $O(\min\{m^{1+1/(k+1)},n^{1+2/k}\})$ time, and 
a $(\max \{2k,g\},5)$-hybrid algorithm that runs in $O(\min\{m^{1+2/(k+1)},n^{1+3/k}\})$ time. 
Therefore, for $k\geq 2$, $\Tilde{g}\in \{3,4,5\}$ and $\alpha=\lceil\frac{\tilde{g}}{2}\rceil$,   there is a
$( \max \{2k,g\}, \tilde{g})$-hybrid algorithm that runs in $O(\min\{ 
    m^{1+\frac{\alpha-1}{k+1}}
    ,
    n^{1+\frac{\alpha}{k}}
\})$ time. 
A natural question is whether these 
three
algorithms are only part of a general 
tradeoff between the running time, $\gamma$ and $\delta$.

\begin{aproblem}
\label{P-open problem 1}
    Let  $\tilde{g}\geq 6$ and $k\geq 2$ be two integers and let $\alpha=\lceil\frac{\tilde{g}}{2}\rceil$. 
    Is it possible to obtain  a $( \max \{2k,g\}, \tilde{g})$-hybrid algorithm that runs  in $O(\min\{ m^{1+\frac{\alpha-1}{k+1}} , n^{1+\frac{\alpha}{k}}\})$ time?
\end{aproblem}

In this paper we present a  $( \max \{2k,g\}, \tilde{g})$-hybrid algorithm that runs, whp, in $\Ot((\frac{k+1}{\alpha-1}+\alpha) \cdot \min\{ m^{1+\frac{\alpha-1}{k+1}}, n^{1+\frac{\alpha}{k}} \})$ time, for every $\tilde{g}\geq 3$. 
This algorithm provides an affirmative answer to Problem~\ref{P-open problem 1}, albeit the $(\frac{k+1}{\alpha-1}+\alpha)$ factor in the running time. 

Kadria \etal{} \cite{kadria2022algorithmic} presented also a $(2k, 6)$-hybrid algorithm
whose running time is $\tilde{O}(\min\{n m^{\frac{1}{2}+\frac{1}{2(k+1)}}, n^{\frac{3}{2}+\frac{1}{k}}\})$,
for every integer $k\geq 4$. The running time of our algorithm when $\tilde{g}=6$ is $\Ot(k \cdot \min\{ m^{1+\frac{2}{k+1}}, n^{1+\frac{3}{k}} \})$. We improve the running time for every constant $k\geq5$, in the case of $m < O(n^{1+1/k})$ (when $m \geq O(n^{1+1/k})$ the two algorithms run a similar procedure and thus have the same asymptotic running time).

Using our
$( \max \{2k,g\}, \tilde{g})$-hybrid algorithm we obtain a generalization of the  $\tO(n^{1+1/(2-\varepsilon)})$ time algorithm of Kadria \etal{}~\cite{kadria2022algorithmic} that computes a cycle of length at most $4\lceil \frac{g}{2}\rceil  - 2\lfloor \varepsilon \lceil \frac{g}{2}\rceil \rfloor \leq (2-\varepsilon)g+4$, for every graph with $g=\polylog(n)$. 
Our generalized algorithm  runs in  $\Ot(\Ynew\cdot n^{1+1/(\Ynew-\varepsilon)})$ time, whp, and returns a cycle of length at most $2\Ynew\lceil \frac{g}{2}\rceil  - 2\lfloor \varepsilon \lceil \frac{g}{2}\rceil \rfloor \leq (\Ynew-\varepsilon)g+\Ynew+2$, where $\Ynew\geq 2$ is an integer and  $\varepsilon\in [0,1]$,  for every graph with $g=\polylog(n)$. 
We also show that if the graph is sparse then the approximation can be improved. More specifically, we present an algorithm that runs in  $\Ot(\Ynew\cdot m^{1+1/(\Ynew-\varepsilon)})$ time, whp, and returns a cycle of length at most $(\Ynew-\varepsilon)g-\Ynew+2\varepsilon$, where $\Ynew\geq 3$ is an integer and $\varepsilon\in [0,1)$,  for every graph with $g=\polylog(n)$.

Our $\Ot(\Ynew \cdot n^{1+1/(\Ynew - \varepsilon)})$-time algorithm also generalizes the $\Ot(n^{1+1/\ell})$-time algorithm of Kadria \etal{}~\cite{kadria2022algorithmic}, which computes a $C_{\leq 2\ell \lceil g/2\rceil}$ for every \textbf{integer} $\ell \geq 1$. In our algorithm, the integer parameter $\ell$ that appears in the exponent of the running time is replaced by a real-valued parameter $\Ynew - \varepsilon$.
Thus, we introduce many new points on the tradeoff curve between running time and approximation ratio. Specifically, for every integer $\ell \leq \polylog(n)$,
up to $\lceil g / 2 \rceil - 1$ additional tradeoff points are added.\footnote{Up to $\lceil g / 2 \rceil - 1$ tradeoff points are added since for every such $\ell$, when $\varepsilon$ is a multiple of $ \frac{1}{\lceil g/2\rceil}$, we get an $\Ot(n^{1+1/(\ell-\varepsilon)})$-time algorithm which computes a $C_{\leq 2(\ell -\varepsilon) \lceil g/2\rceil}$.}
For example, consider $\ell = 3$ and a graph with girth $g = 5$ or $g = 6$. Our algorithm yields two additional points on the tradeoff curve, corresponding to $\varepsilon = \frac{1}{3}$ and $\varepsilon = \frac{2}{3}$. For $\varepsilon = \frac{1}{3}$, we compute a $C_{\leq 16}$ in $\Ot(n^{1 + \frac{3}{8}})$ time, and for $\varepsilon = \frac{2}{3}$, we compute a $C_{\leq 14}$ in $\Ot(n^{1 + \frac{3}{7}})$ time.
These points lie between the two points on the tradeoff curve given by the algorithm of Kadria \etal{}~\cite{kadria2022algorithmic}, which computes either a $C_{\leq 12}$ in $\Ot(n^{1 + \frac{3}{6}})$ time or a $C_{\leq 18}$ in $\Ot(n^{1 + \frac{3}{9}})$ time.
See Figure~\ref{fig:epsilon} for a comparison.

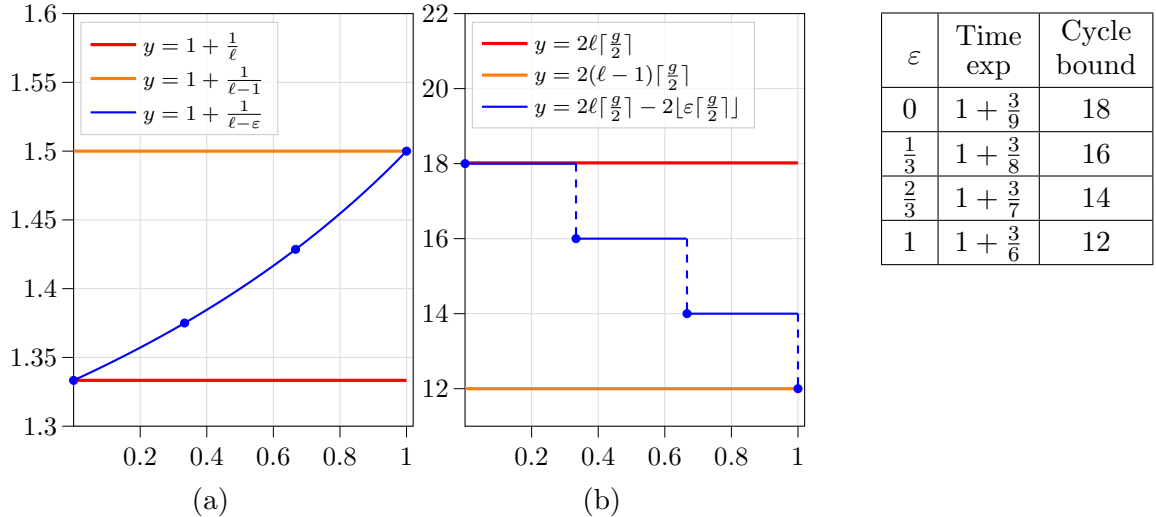
\begin{figure}[t]
\begin{center}
   \begin{multicols}{3}
\hfill

    \begin{tikzpicture}
    
    \definecolor{color0}{rgb}{0.12156862745098,0.66666666666667,0.705882352941177}
    \definecolor{color1}{rgb}{1,0.498039215686275,0.0549019607843137}
    
    \begin{axis}[
    width=0.38\textwidth,
    height= 0.443\textwidth, 
    legend cell align={left},
    legend style={at={(0.02,0.98)}, anchor = north west,fill opacity=0.8, draw opacity=1, text opacity=1, draw=white!80!black,
    font=\scriptsize,
    },
    tick align=outside,
    tick pos=left,
    x grid style={white!86.2745098039216!black},
    xmajorgrids,
    xmin=0, xmax=1.02,
    xtick style={color=black},
    xtick={0.2,0.4,0.6,0.8,1},
    xticklabel style={font=\small},
    y grid style={white!86.2745098039216!black},
    ylabel style={rotate=-90},
    ymajorgrids,
    ymin=1.3, ymax=1.6,
    ytick style={color=black},
    ytick={1.3,1.35,1.4,1.45,1.5,1.55, 1.6},
    yticklabel style={font=\small, xshift=2pt},
    ]

    \addplot [
    very thick, 
            domain=0:1, 
            samples=100, 
            red]{1+1/3};
    \addlegendentry{$y =1 + \frac{1}{\ell}$};

    \addplot [
    very thick, 
            domain=0:1, 
            samples=100, 
            color1]{1+1/2};
    \addlegendentry{$y =1 + \frac{1}{\ell-1}$};

   \addplot [
    thick, 
            domain=0:1, 
            samples=100, 
            blue]{1 + (1)/(3-x)};
    \addlegendentry{$y =1 + \frac{1}{\ell-\varepsilon}$};

    \addplot[
        only marks, 
        mark=*, 
        mark size=1.5pt, 
        draw=blue, 
        fill=blue
    ]
   table {
    x y
    0 1.3333
    0.3333 1.375
    0.6667 1.4286
    1 1.5
    };
    coordinates {(0.3333, 1.375)};

    \end{axis}
    \node at (1.8, -1) {(a)};
    \end{tikzpicture}

\columnbreak

\hfill

    \begin{tikzpicture}
    
    \definecolor{color0}{rgb}{0.12156862745098,0.66666666666667,0.705882352941177}
    \definecolor{color1}{rgb}{1,0.498039215686275,0.0549019607843137}
    
    \begin{axis}[
    width=0.38\textwidth,
    height= 0.443\textwidth, 
    legend cell align={left},
    legend style={at={(0.02,0.98)}, anchor = north west,fill opacity=0.8, draw opacity=1, text opacity=1, draw=white!80!black,
    font=\scriptsize,
    },
    tick align=outside,
    tick pos=left,
    x grid style={white!86.2745098039216!black},
    xmajorgrids,
    xmin=0, xmax=1.02,
    xtick style={color=black},
    xtick={0.2,0.4,0.6,0.8,1},
    xticklabel style={font=\small},
    y grid style={white!86.2745098039216!black},
    ylabel style={rotate=-90},
    ymajorgrids,
    ymin=11, ymax=22,
    ytick style={color=black},
    ytick={12,14,16,18,20,22},
    yticklabel style={font=\small, xshift=2pt},
    ]

    \addplot [very thick, 
        domain=0:1, 
        samples=100, 
        red]
        {18.02};
    \addlegendentry{$y =2\ell\lceil\frac{g}{2}\rceil$};

    \addplot [very thick, 
        domain=0:1, 
        samples=100, 
        color1]
        {12};
    \addlegendentry{$y =2(\ell-1)\lceil\frac{g}{2}\rceil$};

    \addplot [thick, 
            domain=0:1, 
            samples=104, 
            jump mark mid,
            blue]{2*3*3 -2*floor(x*3)};
    \addlegendentry{$y =2\ell\lceil\frac{g}{2}\rceil- 2\lfloor\varepsilon\lceil\frac{g}{2}\rceil\rfloor$};
    \draw[thick, dashed,blue] (0.33333,18) -- (0.33333,16);
    \draw[thick, dashed,blue] (0.66667,16) -- (0.66667,14);
    \draw[thick, dashed,blue] (1,14) -- (1,12);

   \addplot[
        only marks, 
        mark=*, 
        mark size=1.5pt, 
        draw=blue, 
        fill=blue
    ]
   table {
        x y
        0 18
        0.3333 16
        0.6667 14
        1 12
    };

    \end{axis}

     \node at (1.8, -1) {(b)};
    \end{tikzpicture}

    \hfill
    \columnbreak
    \hfill

    \renewcommand{\arraystretch}{1.2}

     \vphantom{\rule{0pt}{2.42cm}}
    \begin{tabular}{|c|c|c|}
    \hline \rule{0pt}{0.8cm}
    \raisebox{0.8ex}{$\varepsilon$} & \shortstack{Time \\[3pt] exp} & \shortstack{Cycle\\ bound} \\
    \hline 
     $0$     & $1 + \frac{3}{9}$ & $18$ \\ \hline
     $\frac{1}{3}$   & $1 + \frac{3}{8}$ & $16$ \\ \hline
     $\frac{2}{3}$   & $1 + \frac{3}{7}$ & $14$ \\ \hline
     $1$     & $1 + \frac{3}{6}$   & $12$ \\ 
    \hline
    \end{tabular}
\hfill
\end{multicols}
\vspace{-1cm}
  \caption{
Our $\Ot(\ell \cdot n^{1 + \frac{1}{\ell-\varepsilon}})$-time girth approximation algorithm
compared to the $\Ot(n^{1 + \frac{1}{\ell}})$-time algorithm of ~\cite{kadria2022algorithmic}, for every $\varepsilon \in [0,1]$, choosing $\ell=3$ and $g = 5$ or $6$. 
(a) The $y$-axis is the exponent of $n$ in the running time, and the $x$-axis is $\varepsilon$.  
(b) The $y$-axis is the upper bound on the length of the returned cycle, and the $x$-axis is $\varepsilon$.
The blue points correspond to our algorithm at four specific choices of $\varepsilon$: $\varepsilon = 0, \frac{1}{3}, \frac{2}{3}, 1$ (see the table on the right).
}
\label{fig:epsilon}
\end{center}

\end{figure}

Together, the tradeoffs in the two algorithms of Kadria \etal{}~\cite{kadria2022algorithmic} that we generalize encompass several other known results, including those of Itai and Rodeh~\cite{itai1977finding}, Lingas and Lundell~\cite{lingas2009efficient}, Roditty and V. Williams~\cite{roditty2012subquadratic} and the first tradeoff of Dahlgaard \etal{}~\cite{dahlgaard2017new} (when $g=\polylog(n)$). 
Notably, some of these algorithms have resisted improvement for many years.
Therefore, the unification of these algorithms within a single tradeoff curve in our result, together with the addition of new valid points on this curve, reinforces the possibility that it captures a fundamental relationship between running time and approximation quality.
This, in turn, motivates further investigation into whether a matching lower bound exists for this tradeoff.

The rest of this  paper is organized as follows.
In Section~\ref{sec:overview} we provide an overview. 
Preliminaries are in Section~\ref{sec:prelim}.
In Section~\ref{sec:neighborhood_technique}, we present a new cycle searching technique that is used by our algorithms.
In Section~\ref{sec:2k-hybrid and +1 aprox} we present a $2k$-hybrid algorithm and then use it to obtain a $(+1)$-approximation algorithm for the girth.
In Section~\ref{sec:one-alg} we generalize the $2k$-hybrid algorithm and present a $( \max \{2k,g\}, \tilde{g})$-hybrid algorithm.
In Section~\ref{sec:girth-approx} we use the hybrid algorithm from Section~\ref{sec:one-alg} to obtain
two more approximation algorithms for the girth.

\section{Overview}
\label{sec:overview}

Among the techniques that we develop to obtain our new algorithms, is a new cycle searching technique that might be of independent interest. 
Our new  technique exploits the property that if  $s\in V$ is not on a $C_{\leq 2k}$,
then for any two neighbors $x$ and $y$ of $s$, 
the set of vertices at distance exactly $k-1$ from $x$ and $y$ that are also at distance $k$ from $s$ are disjoint (see Figure~\ref{fig:nbrballorcycleIllustration}). This  allows us to  check  efficiently  for all the neighbors of $s$ if they are on a $C_{\leq 2k}$.
Using this technique, together with more 
tools that we develop, we obtain \textit{two} hybrid algorithms.

\input{tikzdisjoint}

The first is a relatively simple 
$O(m^{1+\frac{k-1}{k+1}})$-time, $2k$-hybrid algorithm.
We use this
hybrid algorithm in
the girth approximation framework described earlier, 
to obtain an $\Ot(m^{1+\frac{\ell -1}{\ell+1}})$-time, $(+1)$-approximation of the girth, where $g=2\ell$ or $g=2\ell -1$ (the running time can also be written as $\Ot(m^{2-\frac{2}{\lceil g/2  \rceil+1}})$).
We remark that using an algorithm of 
\cite{dahlgaard2017cappedkwalks} for $C_{2k}$ detection,
it is possible to obtain a $(+1)$-approximation in $\Ot(\ell^{O(\ell)}\cdot m^{1+\frac{\ell -1}{\ell+1}})$ time (see  Section~\ref{subsec:+1approx} for more details).
However,  the additional $\ell^{O(\ell)}$ factor might be significant even for small values of $\ell$.

The second is the $( \max \{2k,g\}, \tilde{g})$-hybrid algorithm that solves Problem~\ref{P-open problem 1}. 
Its main component is an 
$(2k,2\alpha)$-hybrid algorithm that runs in $\Ot((\frac{k+1}{\alpha-1}+\alpha) \cdot m^{1+\frac{\alpha-1}{k+1}} )$-time, whp, and generalizes the first $O(m^{1+\frac{k-1}{k+1}})$-time $2k$-hybrid algorithm, by introducing an additional parameter $\alpha\leq k$. 
Using $\alpha$ we can tradeoff between the running time and the lower bound on $g$ and obtain a faster running time at the price of a worse lower bound.

We compare our $(2k,2\alpha)$-hybrid  algorithm to algorithm $\Cycle$ of 
Kadria \etal{}~\cite{kadria2022algorithmic}, 
an $O(m+n^{1+\frac{1}{\ell}})$-time,\footnote{ $\Cycle$ runs in $O(n^{1+\frac{1}{\ell}}+m)$ time,  which can be reduced to $O(n^{1+\frac{1}{\ell}})$ time, as shown in \cite{kadria2022algorithmic}.} 
$(2\ell\alpha, \tilde{g})$-hybrid algorithm, where $ \alpha = \lceil \tilde{g} / 2\rceil$, 
that they used to obtain the $\Ot(n^{1+\frac{1}{\ell}})$-time, $2 \ell \lceil g / 2\rceil$-approximation  algorithm. 
As we show later,  the running time of our $(2k,2\alpha)$-hybrid algorithm  can be bounded by $\Ot((\frac{k+1}{\alpha-1}+\alpha) \cdot n^{1+\frac{\alpha}{k}})$.
Since in our algorithm $k$ is not necessarily a multiple of $\alpha$ (compared to the $\ell\alpha$ of $\Cycle$), our algorithm allows more flexibility, and  we achieve many more possible tradeoffs between the running time and the output cycle length.
For example, if we consider a multiplicative approximation better than $3$, when the value of $g$ is a constant known in advance, our algorithm 
can return longer cycles that are still shorter than $3g$, in a faster running time.
See Figure~\ref{fig:graph_3g_approx_length} for a comparison.\footnote{
Kadria \etal{}~\cite{kadria2022algorithmic} presented also an $O((\alpha-c)\cdot n^{1+\frac{\alpha}{2\alpha - c}})$-time, $(4\alpha-2c, 2\alpha)$-hybrid algorithm, where $0 < c \leq \alpha$ are integers. 
For $c=2\alpha-k$, this is an $O((k-\alpha)\cdot n^{1+\frac{\alpha}{k}})$-time, $(2k,2\alpha)$-hybrid algorithm, similar to our $(2k,2\alpha)$-hybrid algorithm.
However, since $0 < c\leq \alpha$, the possible values of $k$ are restricted and must satisfy $\alpha \leq k < 2\alpha$. 
By choosing $\alpha = \lceil \frac{g}{2} \rceil$ and an appropriate $k$ 
the two algorithms have 
a similar flexibility for a $2$-approximation, but since in our algorithm also larger
values of $k$ are allowed,
we can achieve a faster running time for a $t$-approximation where $t>2$. \label{footnoteTC}}

\begin{figure}[t]
\begin{center}
   \begin{multicols}{2}
\hfill
    \begin{tikzpicture}
    
    \definecolor{color0}{rgb}{0.12156862745098,0.66666666666667,0.705882352941177}
    \definecolor{color1}{rgb}{1,0.498039215686275,0.0549019607843137}
    
    \begin{axis}[
    width=0.5\textwidth,
    height= 0.46\textwidth, 
    legend cell align={left},
    legend style={fill opacity=0.8, draw opacity=1, text opacity=1, draw=white!80!black,
    font=\scriptsize,
    },
    tick align=outside,
    tick pos=left,
    x grid style={white!86.2745098039216!black},
    xmajorgrids,
    xmin=2, xmax=19,
    xtick style={color=black},
    xtick={2,4,6,8,10,12,14,16,18},
    y grid style={white!86.2745098039216!black},
    ylabel style={rotate=-90},
    ymajorgrids,
    ymin=1.29, ymax=1.645,
    ytick style={color=black},
    ytick={1.3,1.35,1.4,1.45,1.5,1.55, 1.6},
    ]

    \addplot [blue,mark=*,mark size=1.7pt,draw=blue] table{
    x  y
    3	1.500
    4	1.400
    5	1.429
    6	1.375
    7	1.400
    8	1.364
    9	1.385
    10	1.357
    11	1.375
    12	1.353
    13	1.368
    14	1.350
    15	1.364
    16	1.348
    17	1.360
    18	1.346
    };
    \addlegendentry{$(2k,2\alpha)$-hybrid}

        \addplot [
      draw=red, fill=red,
      mark=*, mark size=1.7pt
    ]
    table{%
   3	1.500
    4	1.500
    5	1.500
    6	1.500
    7	1.500
    8	1.500
    9	1.500
    10	1.500
    11	1.500
    12	1.500
    13	1.500
    14	1.500
    15	1.500
    16	1.500
    17	1.500
    18	1.500
    };
    \addlegendentry{$\Cycle$ \cite{kadria2022algorithmic}
    }

   \addplot [
    dashed,
    dash pattern=on 1pt off 0.8pt, 
            domain=1.5:21, 
            samples=100, 
            color0]{1 + (x)/(3*x - 2)};
    \addlegendentry{$y =1 + \frac{g}{3g-2}$};

    \addplot [dashed, 
    dash pattern=on 1pt off 0.8pt, 
        domain=1.5:21, 
        samples=100, 
        color0]{1 + (x+1)/(3*x - 1)};
    \addlegendentry{$y =1 + \frac{g+1}{3g-1}$};

    \addplot[
        only marks, 
        mark=*, 
        mark size=2.2pt, 
        draw=blue, 
        fill=blue
    ]
    coordinates {(3, 1.500)}; 
    \addplot[
        only marks, 
        mark=*, 
        mark size=1.3pt, 
        draw=red, 
        fill=red
    ]
    coordinates {(3, 1.500)};

    \end{axis}
    \node at (2.8, -1) {(a)};
    \end{tikzpicture}

\columnbreak

    \begin{tikzpicture}
    
    \definecolor{color0}{rgb}{0.12156862745098,0.66666666666667,0.705882352941177}
    \definecolor{color1}{rgb}{1,0.498039215686275,0.0549019607843137}
    
    \begin{axis}[
    width=0.5\textwidth,
    height= 0.46\textwidth, 
    legend cell align={left},
    legend style={at={(0.02,0.98)}, anchor = north west,fill opacity=0.8, draw opacity=1, text opacity=1, draw=white!80!black,
    font=\scriptsize
    },
    tick align=outside,
    tick pos=left,
    x grid style={white!86.2745098039216!black},
    xmajorgrids,
    xmin=2, xmax=19,
    xtick style={color=black},
    xtick={2,4,6,8,10,12,14,16,18},
    y grid style={white!86.2745098039216!black},
    ylabel style={rotate=-90},
    ymajorgrids,
    ymin=6, ymax=57,
    ytick style={color=black},
    ytick={10,20,30,40,50},
    ]
  
    \addplot [blue,mark=*,mark size=1.7pt,draw=blue] table{
    x  y
    3	8
    4	10
    5	14
    6	16
    7	20
    8	22
    9	26
    10	28
    11	32
    12	34
    13	38
    14	40
    15	44
    16	46
    17	50
    18	52
    };
    \addlegendentry{$(2k,2\alpha)$-hybrid} 
    \addplot [
     red,mark=*,mark size=1.7pt,draw=red] table{
     3	8
    4	8
    5	12
    6	12
    7	16
    8	16
    9	20
    10	20
    11	24
    12	24
    13	28
    14	28
    15	32
    16	32
    17	36
    18	36
    };
    \addlegendentry{$\Cycle$ \cite{kadria2022algorithmic}
    }

    \addplot [dashed, 
        dash pattern=on 1pt off 0.8pt, 
        domain=1.5:21, 
        samples=100, 
        color0]{3*x};
    \addlegendentry{$y =3g$};

    \addplot[
        only marks, 
        mark=*, 
        mark size=2.2pt, 
        draw=blue, 
        fill=blue
    ]
    coordinates {(3, 8)}; 
    \addplot[
        only marks, 
        mark=*, 
        mark size=1.3pt, 
        draw=red, 
        fill=red
    ]
    coordinates {(3, 8)}; 
    
    \end{axis}

     \node at (2.8, -1) {(b)};
    \end{tikzpicture}
    \hfill
\end{multicols}
\vspace{-\baselineskip} 
  \caption{
$(2k,2\alpha)$-hybrid algorithm vs.   $\Cycle$ algorithm of  \cite{kadria2022algorithmic}. 
Both produce a multiplicative approximation strictly better than $3$, given a constant $g$ known in advance.
(a) The $y$-axis is the exponent of $n$ in the fastest running time that achieves such an approximation, and the $x$-axis is $g$.  
(b) The $y$-axis is the upper bound on the length of the cycle returned within this time, and the $x$-axis is $g$.}
\label{fig:graph_3g_approx_length}
\end{center}

\end{figure}

The flexibility of our algorithm is also demonstrated in Figure~\ref{fig:graph_g=6}. For a given constant value of $k$, if our $(2k,2\alpha)$-hybrid algorithm returns a cycle then its length is at most $2k$.
If we want algorithm  $\Cycle$  to output a $C_{\leq 2k}$, then $\lfloor \frac{k}{\alpha} \rfloor$ is the largest $\ell$ that we can choose, since $\ell$ must be an integer. The running time is $O(n^{1+1/\ell}) = O(n^{1+1/\lfloor \frac{k}{\alpha} \rfloor})$. Our algorithm achieves a better running time if $k$ is not divisible by $\alpha$. (In Figure~\ref{fig:graph_g=6} we choose $\alpha=3$.)

\begin{figure}[t] 
  \centering

    \begin{tikzpicture}
    
    \definecolor{color0}{rgb}{0.12156862745098,0.466666666666667,0.705882352941177}
    \definecolor{color1}{rgb}{1,0.498,0.0549}
    
    \begin{axis}[
    width=0.65\textwidth,    
    height= 0.42\textwidth, 
    legend cell align={left},
    legend style={fill opacity=0.8, draw opacity=1, text opacity=1, draw=white!80!black,
    font=\scriptsize
    },
    tick align=outside,
    tick pos=left,
    x grid style={white!86.2745098039216!black},
    xmajorgrids,
    xmin=2, xmax=15,
    xtick style={color=black},
    xtick={2,3,4,5,6,7,8,9,10,11,12,13,14},
    y grid style={white!86.2745098039216!black},
    ymajorgrids,
    ymin=1.1, ymax=2.1,
    ytick style={color=black},
    ytick={1.1,1.2,1.3,1.4,1.5,1.6,1.7,1.8,1.9,2.0},
    ]

    \addplot [draw=blue, fill=blue, mark=*, mark size=2.2pt, only marks]
    table{%
    x  y
    3 2
    4 1.75
    5 1.6
    6 1.5
    7 1.42857
    8 1.375
    9 1.33333
    10 1.3
    11 1.27273
    12 1.25
    13 1.23077
    14 1.214286
    };
    \addlegendentry{
    $(2k,2\alpha)$-hybrid
    } 
    
    \addplot [
    draw=red, fill=red, mark=*, mark size=1.7pt, only marks
    ]
    table{%
    x y
    3 2
    4 2
    5 2
    6 1.5
    7 1.5
    8 1.5
    9 1.33333
    10 1.33333
    11 1.33333
    12 1.25
    13 1.25
    14 1.25
    };
    \addlegendentry{
    $\Cycle$ \cite{kadria2022algorithmic}}

    \addplot [thick, 
            domain=2.75:14.95, 
            samples=100, 
            color0]{1 + 1/(x/3)};
    \addlegendentry{$y = 1+\frac{3}{x}$};
    
    \addplot [thick, 
            domain=1.5:14.95, 
            samples=104, 
            jump mark mid,
            color1]{1 + 1/floor(x/3)};
    \addlegendentry{$y =1 + \frac{1}{\lfloor x/3 \rfloor}$};
    \draw[thick, dashed,color1] (6,1.5) -- (6,2);
    \draw[thick, dashed,color1] (9,1.33333) -- (9,1.5);
    \draw[thick, dashed,color1] (12,1.25) -- (12,1.33333);

    \end{axis}
    
    \end{tikzpicture}

  \captionsetup{justification=raggedright}
  \caption{
  $\alpha=3$.  Red points are the $O(n^{1+1/\lfloor \frac{k}{\alpha} \rfloor})$-time $\Cycle$ algorithm \cite{kadria2022algorithmic}, and blue points are our $\tO(n^{1+\frac{\alpha}{k}})$-time $(2k, 2\alpha)$-hybrid algorithm. Both algorithms either return a $C_{\leq 2k}$ or determine that $g>2\alpha$, for a constant $k$.
The $y$-axis is the exponent of $n$ in the running time. The $x$-axis is $k$.}
  \label{fig:graph_g=6}
\end{figure}

Next, we overview our $(2k,2\alpha)$-hybrid algorithm that either finds a $C_{\leq 2k}$ or determines that $g> 2\alpha$ in $\Ot((\frac{k+1}{\alpha-1}+\alpha) \cdot m^{1+\frac{\alpha-1}{k+1}} )$ time. 
In order to determine that  $g\geq 2\alpha$, we can check for every $v\in V$ if $v$ is on a $C_{\leq 2\alpha}$ (and either find a $C_{\leq 2k}$ since  $\alpha\leq k$, or return that $g>2\alpha$).
If $v$ is on a $C_{\leq 2\alpha}$, then all the vertices and edges of this $C_{\leq 2\alpha}$ are at distance at most $\alpha$ from $v$.
If, for every $v\in V$, 
the number of edges at distance at most $\alpha$ is  $O(deg(v)\cdot m^{\frac{\alpha-1}{k+1}})$ 
then 
using standard techniques we can check for every $v\in V$ if $v$ is on a $C_{\leq 2\alpha}$ in $O(m^{1+\frac{\alpha-1}{k+1}})$ time.
However, this is not necessarily the case, and the region at  distance at most $\alpha$ from some vertices might be dense. 
To deal with dense regions within the promised running time  we develop  an iterative sampling procedure (see $\BFSSampleM$ in Section~\ref{sec:one-alg}), whose goal is to sparsify the graph, or to return a $C_{\leq 2k}$.
One component of the iterative sampling procedure is a generalization of our new cycle searching technique mentioned above.
In the generalization instead of checking whether a vertex $s$ and its neighbors are on a $C_{\leq 2k}$, we check whether all the vertices up to a  possibly further distance from $s$ are on a $C_{\leq 2\alpha}$, for $\alpha\leq k$, and if not we mark them so that they can be removed later. 

If the  iterative sampling procedure ends without finding a   $C_{\leq 2k}$ then there are two possibilities. 
Let $r=(k+1) \bmod (\alpha -1)$.
If $r=0$ then it holds that 
the number of edges at distance at most $\alpha$ from every $v\in V$ is  $O(deg(v)\cdot m^{\frac{\alpha-1}{k+1}})$, whp, as required.
If $r>0$ then 
it holds that 
the number of edges at distance at most $r$ from every $v\in V$ is  $O(m^{\frac{r}{k+1}})$, whp.
This does not necessarily imply that the graph is sparse enough for checking whether $g>2\alpha$.
In this case, we run another algorithm (see $\HandleRemainderM$ in Section~\ref{sec:one-alg}) that continues to sparsify the graph until 
the number of edges at distance at most $\alpha$ from every $v\in V$ is  $O(deg(v)\cdot m^{\frac{\alpha-1}{k+1}})$
and 
checking whether $g>2\alpha$ is possible within the required running time of $O(m^{1+\frac{\alpha-1}{k+1}})$.

\section{Preliminaries}
\label{sec:prelim}

Let $G=(V,E)$ be an unweighted undirected graph  with $n$ vertices and $m$ edges. 
Let $U\subseteq V$ be a set of vertices and let $G\setminus U$  be the graph obtained from $G$ by deleting all the vertices of $U$ together with their incident edges. 
For two graphs $G=(V,E)$ and $G'=(V',E')$, let $G \setminus G'$ be $G\setminus V'$. We say that $G \subseteq G'$ if $V\subseteq V'$ and $E \subseteq E'$. For convenience, we use both $u \in V$ and $u\in G$ to say that $u \in V$.
For every $u,v \in V$, let $d_G(u,v)$
be the length of a shortest path between $u$ and $v$ in $G$. 
The \emph{girth} $g$ of $G$ is the length of a shortest cycle in  $G$. 
Let $\wt(C)$ be the length of a cycle $C$. 
For an integer $\ell$, we denote a cycle of length (at most) $\ell$ by ($C_{\leq \ell}$) $C_\ell$.\footnote{Both $C_{\ell}$ and $C_{\leq \ell}$ might  not be simple cycles. However, the cycles that our algorithms return are simple.}
Let $E(v)$ be the edges incident to $v$ and $E(v,i)$ the $i$th edge in $E(v)$. 
Let $deg_{G}(v)$ be the degree of $v$ in $G$.
Let $N(v)$ be the set of neighbors of $v$, namely $N(v) = \{w \mid (v,w)\in E\}$.
For an edge set  $S$, let 
$V(S)$ be the endpoints of $S$'s edges, that is, 
$V(S)=\{u\in V \mid \exists (u,v)\in S\}$.
Let $e=(u,v)\in E$ and $w\in V$. The distance $d_G(w,e)$ between $w$ and $e$  is $\min\{ d_G(w, u), d_G(w,v) \} + 1$.
For every $u\in V$ and a real number $k$ 
let $B(G,u,k) =(V_u^k(G),E_u^k(G))$ be the 
\emph{ball graph} of $u$, where $V_u^k(G)=\{ v\in V\mid d_G(u,v)\leq k\}$ and $E_u^k(G)=\{e\in E \mid d_G(u,e)\leq k\}$~\cite{kadria2022algorithmic}.
For an integer $\ell \geq 0$ and a vertex $u\in V$, Let $L_u^\ell(G)=
\{ w \mid d_{G}(u,w)=\ell\}$.\footnote{When the graph is clear from the context, we sometimes omit $G$ from the notations.}

We now turn to  present several  essential tools  that are required in order to obtain our new algorithms.
We first restate an important property of the ball graph $B(v, R)$.

\begin{lemma}[\cite{kadria2022algorithmic}] \label{L-not-part-of}
Let  $0\leq t \leq R$ be two integers and let $v\in V$. If $B(v, R)$ is a tree then  no vertex in $V_v^{R-t}$  is part of a cycle of length at most $2t$ in $G$.

\end{lemma}

We  use procedure  $\BallOrCycle(G,v,R)$ \cite{kadria2022algorithmic,lingas2009efficient} (for completeness, we include a pseudo-code in Algorithm~\ref{A-BallOrCycle-modified}) that  searches for a $C_{\leq 2R}$  in the ball graph $B(v,R)$.
We summarize the properties of $\BallOrCycle$ in the next lemma.

\begin{algorithm2e}[t]
    \caption{$\BallOrCycle(G,v,R)$ \cite{kadria2022algorithmic,lingas2009efficient}}
    \label{A-BallOrCycle-modified}
    $Q \gets$ a queue that contains $v$ with $d(v)=0$\;
      \While{$Q \neq \emptyset$}{
      $u \gets  dequeue(Q)$\; 
      
      $V_v^R  \gets V_v^R \cup \{u\}$\; 
      
      $i \gets 1$\;
      \While{$(i\leq |E(u)|) \codeAnd d(u)+1\leq R$}{
            $(u,w)\gets E(u,i)$\;
            \If{$w\in Q$}{
                \Return $\codeNull$, $P(LCA(u,w),u)\cup \{ (u,w) \} \cup P(LCA(u,w),w)$\footnotemark\;
            }
            $Q \gets Q \cup \{ w \}$ with $d(w) = d(u)+1$\;

            $i \gets i +1$\;
          }
   
      }
      \Return 
      $V_v^R$, $\codeNull$\;
\end{algorithm2e}

\footnotetext{$LCA(u,w)$ is the least common ancestor of $u$ and $w$ in the tree rooted at $v$ before the edge $(u,w)$ was discovered. $P(x,y)$ is the path in this tree between $x$ and $y$.}

\begin{lemma}[\cite{kadria2022algorithmic}] \label{L-BallOrCycle}
Let $v\in V$. If the ball graph $B(v,R)$ is not a tree then 
\mbox{\rm $\BallOrCycle(G,v,R)$} returns a  $C_{\leq 2R}$ from $B(v,R)$. 
If $B(v,R)$ is a tree then \mbox{\rm $\BallOrCycle(G,v,R)$} returns $V_v^R$.\footnote{If $V_v^R$ is returned then we assume that $V_v^R$ is ordered by the distance from $v$, and for every $u \in V_v^R$ we  store $d(u,v)$ with $u$. Thus, given the set $V_v^R$, we can find $V_v^{R'}$ for every $R'<R$ in $O(|V_v^{R'}|)$ time.}
The running time of  $\BallOrCycle(G,v,R)$ is $O(|V_v^R|)$.
\end{lemma}

Next, we obtain a simple $t$-hybrid algorithm, called $\dHybrid$, using 
$\BallOrCycle$. $\dHybrid$ (see Algorithm~\ref{A-dHybrid}) gets a graph $G$ and an integer $t\geq 2$, and runs $\BallOrCycle(v,t)$ from every  $v\in G$ as long as no cycle is found by $\BallOrCycle$. If $\BallOrCycle$ finds a cycle then $\dHybrid$ stops and returns that cycle. If no cycle is found then $\dHybrid$ returns $\codeNull$. We prove the next lemma.

 \begin{algorithm2e}[t]
    \caption{$\dHybrid(G,t)$} \label{A-dHybrid}
	\ForEach{$v\in V$}{
	    $(V_v^{t},C) \gets \BallOrCycle(v, t)$\;
	        \lIf {$C\neq \codeNull$}{\Return $C$}
	}
	\Return $\codeNull$\;
\end{algorithm2e}

\begin{lemma}
\label{L-dHybrid-properties} 
    $\dHybrid(G,t)$ either finds a $C_{\leq 2t}$ or determines that $g>2t$, in $O(\sum_{v\in V}|V_v^t|)$ time.
\end{lemma}

\begin{proof}
By Lemma~\ref{L-BallOrCycle}, if $\BallOrCycle(u,t)$ returns a cycle $C$ then  $\wt(C)\leq 2t$. 
Also by Lemma~\ref{L-BallOrCycle}, if $\BallOrCycle(u,t)$ does not return a cycle then the ball graph $B(u,t)$ is a tree, and specifically $u$ is not part of a $C_{\leq 2t}$ in $G$.
Therefore, if no cycle was found during any of the calls then all the vertices in $G$ are not part of a $C_{\leq 2t}$ in $G$. Hence, $G$ does not contain a $C_{\leq 2t}$, and we get that $g > 2t$.
By Lemma~\ref{L-BallOrCycle}, the running time of $\BallOrCycle(v,t)$ is $O(|V_v^t|)$, which is $O(\sum_{v\in V}|V_v^t|)$ in total.
\end{proof}

We now show that if the input graph satisfies a certain sparsity property then the running time of $\dHybrid(G,t)$ can be bounded as follows.

\begin{corollary}
\label{C-dHybrid-sparse-time}
If $|E_u^{t-1}| < D^{t-1}
$ for every $u\in V$ then $\dHybrid(G,t)$ runs in $O(mD^{t-1})$ time. 
\end{corollary}

\begin{proof}
For every $v\in V$, we have $O(|V_v^{t}|) \leq 
    O(|E_v^{t}|) \leq 
    O(\sum_{u \in N(v)} |E_{u}^{t-1}|) 
    $, which is at most 
    $ 
    O(\sum_{u \in N(v)} D^{t-1}) =
    O(deg(v)\cdot  D^{t-1})$.
   Thus, by Lemma~\ref{L-dHybrid-properties}, the running time of $\dHybrid(G,t)$ is 
   $O(\sum_{v\in V}|V_v^t|)
   \leq 
   O(\sum_{v\in V}deg(v)\cdot  D^{t-1}) = O(m  D^{t-1})$.    
\end{proof}

Next, we present procedure $\IsDense(G,w,T,r)$ from \cite{kadria2022algorithmic} (for completeness, we include a pseudo-code in Algorithm~\ref{A-Dense}). 
$\IsDense$ gets 
a graph $G$, a vertex $w$, a budget $T \geq 1$ (real) and a distance $r\geq 0$ (integer). 
 In the procedure a BFS is executed from $w$. The BFS   counts the edges that are scanned as long as their total number is less than $T$ and the
farthest vertex from $w$ is at distance at most $r$.

\begin{algorithm2e}[t]
    \caption{$\IsDense(G,w,T,r)$ \cite{kadria2022algorithmic}} 
    \label{A-Dense}
    $T'\gets 0$, $\ell_c \gets 1$, $\ell_n \gets 0$, $j\gets 0$\;
    $Q \gets \{w \}$\; 
    \While{$\big(Q\text{ is not empty }\big) \codeAnd \big( j < r\big) \codeAnd \big( T'< T\big)$}{
	    $u \gets dequeue(Q)$\;
	    $\ell_c \gets \ell_c - 1$, $i \gets 1$\;
     	\While{\big($i\leq |E(u)|\big)  \codeAnd \big( T'< T$\big)}{
     	    $(u,v)\gets E(u,i)$\;
     	    remove $(u,v)$ from $E(v)$\;
     	    $T'\gets T'+1$\;
      		\If{$v$ is not marked}{
      		    $Q \gets Q \cup \{ v \}$, mark $v$, $\ell_n \gets \ell_n +1$\;
		    }
		$i \gets i +1$\;
      	}
	\If{$\ell_c  = 0$}{
	    $\ell_c \gets \ell_n$, $\ell_n\gets 0$\;
	    $j \gets j+1$\;
	}
}

\lIf{$\big(T' \geq T \big)
$
}{\Return $\codeYes$}
\lElse{\Return $\codeNo$}
\end{algorithm2e}

\begin{lemma}[\cite{kadria2022algorithmic}] \label{L-Dense}

Procedure \mbox{\rm $\IsDense(G,w,T,r)$} runs in $O(
\lceil T \rceil) = O(T)$ 
time. If $\IsDense(G,w,T,r)=\codeNo$ then $|E_w^r|<T$.
If  $\IsDense(G,w,T,r) = \codeYes$ 
then $|E_w^r| \geq T$.

\end{lemma}

Given a vertex $v$ and a distance $R$ we sometimes want to bound 
$|E_{v}^{R}|$. Therefore, we adapt a lemma and a corollary of \cite{kadria2022algorithmic} from vertices to  edges.

\begin{lemma}\label{L-SumOfDenseM}
Let $x,y$ be positive integers, let $D\geq 1$ be a real number, and let $w\in V$.
If $|E_w^{x}| < D^{x}$, and $|E_u^{y}| < D^{y}$ for every $u\in V$, then $|E_w^{x+y}| < D^{x+y}$. 
\end{lemma}

\begin{proof}
Let $w\in V$ and assume that  $|E_w^{x}| < D^{x}$. 
We also know that $|E_u^{y}| < D^{y}$, for every $u\in V$. 
We denote $L = V_w^{x}\setminus \{w\}$. 
If $L = \emptyset$ then $|E_w^{x}| = |E_w^{x+y}| = 0 \leq D^{x+y}$ as required. Now assume that $L \neq \emptyset$. Since $y\geq 1$, $E_w^{x+y} = \bigcup_{u\in L} E_u^{y}$. Therefore,  $|E_w^{x+y}| = |\bigcup_{u\in L } E_u^{y}|$.
As the ball graph $B(w,x) = (V_w^{x}, E_w^{x})$ is connected, we know that $|V_w^{x}| \leq |E_w^{x}| +1 < D^x +1$ and since $x>0$, $|L| = |V_w^{x}| - 1 < (D^x + 1) - 1 =  D^x$.
Thus, we get that 
$   |E_w^{x+y}| = 
    |\bigcup_{u\in L} E_u^{y}| \leq
    \sum_{u\in L}|E_u^{y}| < 
    \sum_{u\in L}D^{y} = 
    |L| \cdot D^{y}  <
    D^{x}\cdot D^{y} =
    D^{x+y}
$  so  $|E_w^{x+y}| < D^{x+y}$.
\end{proof}

Using Lemma~\ref{L-SumOfDenseM}, we prove following corollary.

\begin{corollary}\label{C-SumOfDenseM}
Let $x$ be a positive integer and let $D\geq 1$ be a real number. 
If $|E_w^{x}| < D^{x}$ for every $w\in V$, then $|E_w^{ix}| < D^{ix}$, for every $w\in V$ and $i\geq 1$.
\end{corollary}

\begin{proof}
The proof is by induction on $i$. For $i=1$ it follows from our assumption that $|E_w^{x}| < D^{x}$. Assume now that the claim holds for  $j = i-1$. This implies that $|E_w^{(i-1)x}| < D^{(i-1)x}$, for every $w\in V$. Combining this with the fact that $|E_w^{x}| < D^{x}$, for every $w\in V$, by Lemma~\ref{L-SumOfDenseM} we get that $|E_w^{ix}| < D^{ix}$, for every $w\in V$ and $i\geq1$.
\end{proof}

We also
adapt
procedure $\SparseOrCycleM(G,D,x,y)$ 
of \cite{kadria2022algorithmic} to our needs.
$\SparseOrCycleM$ (see Algorithm~\ref{A-SparseOrCycleM}) gets a graph $G$, a parameter $D \geq 1$,  and two integers $x, y >0$,
and
iterates over vertices using a for-each loop.
Let $w$ be the vertex currently considered
and $G(w)$ the current graph. 
If $\IsDense(G(w),w,D^{x},x)=\codeYes$ then 
$\BallOrCycle(G(w),w,x-1+y)$ is called.
If $\BallOrCycle$ returns a cycle $C$ then $C$ is returned by $\SparseOrCycleM$. Otherwise, the vertex set $V_w^{x-1}$ is removed from $G(w)$ along with the edge set $E_w^x$. 
After the loop ends, if no cycle was found, we return $\codeNull$.
Let $W\subseteq  V$ be the set of vertices for which $\BallOrCycle$ was called and no cycle was found, and $\hat{G}=(\hat{V},\hat{E})$
the graph
after $\SparseOrCycleM$ ends.

\begin{algorithm2e}[t]
    \caption{$\SparseOrCycleM(G,D,x,y)$} \label{A-SparseOrCycleM}
	\ForEach{$w\in V$}{
	    \If{$\IsDense(G,w,D^{x}, x)=\codeYes$}{
	        $(V_w^{x-1+y},C) \gets \BallOrCycle(w, x-1+y)$\;
	        \lIf {$C\neq \codeNull$}{\Return $C$}
	       $G\gets G \setminus V_w^{x-1}$\;
        }
	}
	\Return $null$\;
\end{algorithm2e}

The following lemma is similar to the corresponding lemmas from~\cite{kadria2022algorithmic}. We provide here the proof for completeness.

\begin{lemma}\label{L-SparseOrCycleM} 
$\SparseOrCycleM(G, D ,x,y)$ satisfies the following: 
\begin{enumerate} [label=(\roman*)]

\item
If a cycle $C$ is returned then $\wt(C)\leq 2(x-1+y)$
\item
If a cycle is not returned then $|E_u^{x}| 
<
D^x$, for every $u\in \hat{G}$
\item
If $u \in G \setminus \hat{G}$ then $u$ is not part of a $C_{\leq 2y}$ in $G$
\item 
$\SparseOrCycleM(G, D ,x,y)$ runs in $O(nD^x + \sum_{w\in W}(|V_w^{x-1+y}|))$ time.
\end{enumerate}

\end{lemma}

\begin{proof}
\begin{enumerate}[label=(\roman*)]
  \item Since $\SparseOrCycleM$ returns a cycle $C$ only if 
    a call to $\BallOrCycle(w, x-1+y)$ returns a cycle $C$,
    it follows from Lemma~\ref{L-BallOrCycle} that $\wt(C)\leq 2(x-1+y)$.
  \item
  Let $u\in\hat{G}$. This implies that $u$ was not removed, and therefore
  $u$ was considered in the for-each loop at some stage during the execution of $\SparseOrCycleM$. At this stage  $\IsDense(G(u),u,D^{x}, x)$ was $No$, as otherwise, since no cycle was returned $u$ would have been removed while removing $V_{u}^{x-1}$.
  It follows from Lemma~\ref{L-Dense} that $|E_u^{x}| < D^x$ in $G(u)$. As edges can only be removed during the run of $\SparseOrCycleM$, we have $|E_u^{x}| < D^x$ also in $\hat{G}$.
    
  \item  Since $u \in G \setminus \hat{G}$ ($u$ was removed) it follows that there was a    vertex $w$ such that $u\in V_{w}^{x-1}$ after a call to      $\BallOrCycle(w,x-1+y)$ did not return a cycle. 
    By Lemma~\ref{L-BallOrCycle}, the ball graph $B(w,x-1+y)$ is a tree, and it follows from Lemma~\ref{L-not-part-of} that no vertex in $V_{w}^{x-1}$ is part of a $C_{\leq 2y}$ in the current graph $G(w)$. Therefore, $u$ is not part of  a $C_{\leq 2y}$ in $G(w)$. 
    Since during the run of $\SparseOrCycleM$ we remove only vertices that are not part of a $ C_{\leq2y}$, $u$ is not part of a $C_{\leq 2y}$ also in $G$.
  
  \item  
  The cost of each call to $\IsDense$ is $O(D^{x})$. In the worst case we call $O(|V|)$ times to $\IsDense$. The total cost of this step is $O(|V|D^x)=O(nD^x)$.
  At most one call to $\BallOrCycle(w, x-1+y)$ returns a cycle. We can bound this call with $O(n)$.
  Each call to $\BallOrCycle(w, x-1+y)$ that does not return a cycle costs $O(|V_w^{x-1+y}|)$.  For each such call, the cost of removing the set $V_w^{x-1}$ from $G$ is bounded by $O(|E_w^{x}|)$. However, since we are in the case that $\BallOrCycle(w, x-1+y)$ does not return a cycle it follows that $B(w,x-1+y)$ is a tree and $|E_w^{x}|=O(|V_w^{x}|)\leq O(|V_w^{x-1+y}|)$ (as $y\geq 1$).
\end{enumerate}
\end{proof}

Similarly to $\dHybrid$, we show for $\SparseOrCycleM$ that if $G$ satisfies a certain sparsity property, the running time can be bounded as follows. 
 
\begin{corollary}
\label{C-SparseOrCycle-sparse-time}
If $|E_u^{x-1+y}| < D^{x-1+y}$ for every vertex $u\in V$ then $\SparseOrCycleM(G, D ,x,y)$ runs in $O(nD^x+mD^{y-1})$ time. 
\end{corollary}

\begin{proof}
By Lemma \ref{L-SparseOrCycleM}, $\SparseOrCycleM(G, D ,x,y)$ runs in $O(|V|D^x + \sum_{w\in W}(|V_w^{x-1+y}|))$ time.
For every $w\in W$, the call to $\IsDense(G,w,D^x,x)$ returned $\codeYes$, so it follows from Lemma~\ref{L-Dense} that $|E_w^x|\geq D^x$.
The edge set $E_w^x$ is removed while removing $V_w^{x-1}$. Therefore, for every $w\in W$ we remove at least $D^x$ edges. Since at most $m$ edges can be removed, the size of $W$ is at most $\frac{m}{D^x}$. By our assumption, we have $O(|V_w^{x-1+y}|)\leq O(|E_w^{x-1+y}|)\leq O(D^{x-1+y})$. Therefore, we get that $O(\sum_{w\in W}(|V_w^{x-1+y}|))\leq O(|W|\cdot D^{x-1+y}) \leq O(\frac{m}{D^x} \cdot D^{x-1+y}) = O(mD^{y-1})$. Thus, the running time of $\SparseOrCycleM$ is $O(nD^x+mD^{y-1})$.
\end{proof}

Finally, we include a standard lemma about sampling a hitting set (see, e.g., \cite{aingworth1999fast},  \cite{lingas2009efficient}, \cite{roditty2013fast}).

\begin{lemma} \label{L-whpM}
    It is possible to obtain in $O(m)$ time, using sampling, a set of edges $S$ of size $\tilde{\Theta} (\frac{m}{s})$, that hits, whp, the $s$ closest edges of every $v\in V$.
\end{lemma}

We remark that some of our algorithms get a graph $G$ that is being updated during their run. 
Within their scope,
$G$ denotes the current graph that includes all updates done so far.

\section{A new cycle searching technique}
\label{sec:neighborhood_technique}

In this section we introduce a new cycle searching technique implemented in algorithm $\DoubleBfsCycle$.
 This technique exploits the property that 
 a vertex $s \in V$ is not on a $C_{\leq 2k}$, 
 to check efficiently whether any neighbor of $s$ lies on a $C_{\leq 2k}$.

Consider a vertex $s\in V$. It is straightforward to check whether $s$ is on a $C_{\leq 2k}$, for every integer 
$k$,
using  $\BallOrCycle(G,s,k)$ in $O(n)$ time. If $\BallOrCycle(G,s,k)$ does not return a $C_{\leq 2k}$ then 
for every $x,y\in N(s)$ it holds that  $V_x^{k-1}(G')\cap V_y^{k-1}(G')= \emptyset$, where $G'= G\setminus \{s\}$, as otherwise there was a $C_{\leq 2k}$ passing through $s$ and $\BallOrCycle(G,s,k)$ would have returned a $C_{\leq 2k}$ (see Figure~\ref{fig:nbrballorcycleIllustration}). We show that it is possible to exploit this property to check  for every  $v\in N(s)$  whether $v$ is on a $C_{\leq 2k}$, using $\BallOrCycle(G', v,k)$,   in  $O(n+m)$ time instead of $O(deg(s) \cdot n)$. 
More specifically, we present algorithm 
$\DoubleBfsCycle(G, s, k)$ (see Algorithm~\ref{A-DoubleBfsCycle1}) that gets 
a graph $G$, a vertex $s$, and an  integer $k\geq 2$.  We first initialize $\hat{U}$ to $\emptyset$. Then, we run $\BallOrCycle(G,s,k)$. If a cycle $C$ is found by $\BallOrCycle(G, s,k)$ then $C$ is returned by $\DoubleBfsCycle$.
Otherwise, we add the vertex $s$ to $\hat{U}$, keep the neighbors of $s$ in $N_s$, and then remove $s$ from $G$. Recall that  $G'$ equals $G\setminus \{s\}$. Next, for every $v\in N_s$ we run $\BallOrCycle(G',v,k)$,  as long as a cycle is not found. If a cycle $C$ is returned by $\BallOrCycle(G',v,k)$ then $C$ is returned by $\DoubleBfsCycle$.\footnote{For our needs it suffices to stop and return a cycle passing through an $s$ neighbor once we find one, though $\BallOrCycle$ can be run from all the neighbors of $s$ in the same running time bound of $O(n+m)$.}
Otherwise,
$v$ is added to $\hat{U}$. After the loop ends, the vertex $s$ and its adjacent edges are added back to the graph, and the set $\hat{U}$ is returned by $\DoubleBfsCycle$.
We prove the following lemma.

\begin{algorithm2e}[t]
\caption{$\DoubleBfsCycle(G, s, k)$}
\label{A-DoubleBfsCycle1}

    $\hat{U} \gets \emptyset$\;
  
    $(V_s^k,C) \gets \BallOrCycle(s, k)$\;
     
    \lIf {$C\neq \codeNull$}{\Return $(\emptyset, C)$}
       
        $\hat{U}\gets \hat{U} \cup \{s\}$\;

        $N_s \gets N(s)$\;
        
        $G \gets G \setminus \{s\}$\;
        
         \ForEach{$v \in N_s$}{
             $(V_v^k,C) \gets \BallOrCycle(v, k)$\;
             
             \lIf {$C\neq \codeNull$}{\Return $(\emptyset, C)$}
             
             $\hat{U}\gets \hat{U} \cup \{v\}$ \;           
             
         }        
    
    add $s$ and the edge set $\{(s,v) \mid v\in N_s\}$ back to $G$;
    
    \Return $(\hat{U}, \codeNull);$
 
\end{algorithm2e}

\begin{lemma}  \label{L-DoubleBfsCycle1} 
    If algorithm $\DoubleBfsCycle(G,s,k)$ finds a cycle $C$ then $\wt(C)\leq 2k$. Otherwise, no vertex in $V_s^{1}$ is part of a $C_{\leq 2k}$ in $G$, and the set $\hat{U}= V_s^{1}$ is returned. 
\end{lemma}
\begin{proof}
    If a cycle was found, it happened during the call to $\BallOrCycle(G,s,k)$ or one of the calls to $\BallOrCycle(G',v,k)$. Therefore, by Lemma~\ref{L-BallOrCycle}, the cycle length is at most $2k$.
   
    If no cycle was found during the run of $\BallOrCycle(G,s,k)$ then $s$ is not on a $C_{\leq 2k}$ in $G$. 
    In addition, for every $v\in N_s$ $\BallOrCycle(G',v,k)$ did not find a cycle. Hence, $v$ is not on a $C_{\leq 2k}$ in $G'$ and therefore also in $G$, since $G'=G \setminus \{s\}$ and  $s$ is not on a $C_{\leq 2k}$.
    Since no cycle was found, $s$ and the vertices
    $v$ for every $v\in N_s$
    are added to $\hat{U}$, so 
    by the definition of $V_s^{1}$ 
    we have $\hat{U} = V_s^{1}$.
\end{proof}

To bound the running time of $\DoubleBfsCycle$, 
we show how to use the fact that no $C_{\leq 2k}$ was found by
$\BallOrCycle(G,s,k)$,
to efficiently run $\BallOrCycle(G',v,k)$ for every $v\in N_s$.

\begin{lemma}\label{L-Batch-RT}
Let $s\in V$. If the ball graph $B(G,s,k)$ is a tree then the total cost of running $\BallOrCycle(G',v,k)$ for every $v\in N_s$ is $O(|V_s^k(G)|+\sum_{w\in L_s^k(G)} deg(w) )$. 
\end{lemma}

\begin{proof}
By the definitions of $V_s^k$, $L_s^k$ and $G'$,
we know that $V_s^k (G) = \cup_{v\in N_s}V_v^{k-1}(G') \cup\{s\}$ and $L_s^k(G
) = \cup_{v\in N_s} L_v^{k-1}(G')$.
Since $B(G,s,k)$ 
contains no cycles
it follows that 
$V_x^{k-1}(G') \cap V_{y}^{k-1}(G') =\emptyset$ and 
$L_x^{k-1}(G') \cap L_{y}^{k-1}(G')=\emptyset$,
for any two distinct vertices 
$x,y\in N_s$.
Therefore, \textbf{(i)} $|V_s^k(G)| = \sum_{v\in N_s} |V_v^{k-1} (G')|+1$, and \textbf{(ii)} $L_s^k(G
) = \mathbin{\mathaccent\cdot\cup}_{v\in N_s} L_v^{k-1}(G')$.
From Lemma~\ref{L-BallOrCycle} it follows that the total cost 
of the calls to $\BallOrCycle(G',v,k)$ for every $v\in N_s$ is $O(\sum_{v\in N_s}|V_v^k(G')|)$.
It holds that $|V_v^k (G')| \leq |V_v^{k-1}(G')| + \sum_{w\in L_v^{k-1}(G')} deg(w)$, for every $v\in N_s$. 
Thus,
we get that the total cost is 
$O(\sum_{v\in N_s}|V_v^k(G')|) = O(\sum_{v\in N_s}(|V_v^{k-1}(G')|+\sum_{w\in L_v^{k-1}(G')} deg(w)))$. This equals to 
$ O(\sum_{v\in N_s}|V_v^{k-1}(G')|+\sum_{v\in N_s} \sum_{w\in L_v^{k-1}(G')} deg(w)))$, and it follows from \textbf{(i)} and \textbf{(ii)} that this is at most $
O(|V_s^k(G)|+\sum_{w\in L_s^k(G)} deg(w) )$. 
\end{proof}

We use Lemma~\ref{L-Batch-RT} to bound the running time of $\DoubleBfsCycle$.

\begin{lemma} \label{L-DoubleBfsCycle-time}
    Algorithm $\DoubleBfsCycle$ runs in $O(n+m)$ time. 
\end{lemma}
\begin{proof}
    Running $\BallOrCycle(G,s,k)$ costs $O(n)$. If a cycle is found, it is returned and the running time is $O(n) = O(n+m)$. If no cycle is found, then removing (and later adding back) $s$ and its  edges costs $O(deg(s))$.   
    By Lemma~\ref{L-Batch-RT}, the cost of running     $\BallOrCycle(G',v,k)$ for every $v\in N_s$ is $O(|V_s^k(G)|+\sum_{w\in L_s^{k}(G)}deg(w)) = O(n+m)$.
    Adding $s$ and  $N_s$ to $\hat{U}$ takes $O(V_s^{1})=O(n)$ time. Thus,  the total running time of $\DoubleBfsCycle$ is $O(n+m)$.
\end{proof}

 \section{\texorpdfstring{A $2k$-hybrid algorithm and a $(+1)$-approximation of the girth}{A 2k-hybrid algorithm and a (+1)-approximation of the girth}}
 
\label{sec:2k-hybrid and +1 aprox} 

In this section we first show how to use algorithm $\DoubleBfsCycle$  from the previous section to obtain a $2k$-hybrid algorithm that in $O(m^{1+\frac{k-1}{k+1}})$ time, either returns a $C_{\leq 2k}$ or determines that $g>2k$. 
Then, we use the $2k$-hybrid algorithm to compute a $(+1)$-approximation of $g$.

\subsection{\texorpdfstring{A $2k$-hybrid algorithm}{A 2k-hybrid algorithm}}

We first present  algorithm $\twoSparseOrCycle(G,k)$ that gets a graph $G$ and an integer $k\geq 2$. 
Let $G_0$ ($G_1$) be  $G$ before (after) running $\twoSparseOrCycle$. 
$\twoSparseOrCycle(G,k)$ either finds a $C_{\leq 2k}$ or removes vertices that are not on a $C_{\leq 2k}$, such that 
for every $u\in G_1$, 
the ball graph $B(G_1,u,2)$  is relatively sparse, that is, $|E_u^2(G_1)|< m^{\frac{2}{k+1}}$. 
$\twoSparseOrCycle$ (see Algorithm~\ref{A-2sparseorcycle}) iterates over vertices using a for-each loop.
Let $s$ be the vertex currently considered.
If $ \sum_{v \in N(s)} deg(v) \geq m^{\frac{2}{k+1}}$ then $\DoubleBfsCycle(G,s,k)$ is called. 
If $\DoubleBfsCycle$ returns a cycle $C$ then 
$\twoSparseOrCycle$ returns $C$.
If $\DoubleBfsCycle$ returns  a vertex set $\hat{U}$ then  $\hat{U}$ is removed from $G$. After the loop ends, if no cycle was found, we return $\codeNull$.

   \begin{algorithm2e}[t]
    \caption{$\twoSparseOrCycle(G,k)$}
    \label{A-2sparseorcycle}
        \ForEach{$s\in V$}  
        {
            \If{ 
            $ \sum_{v \in N(s)} deg(v) \geq m^{\frac{2}{k+1}}$
            }
            {
            
                $(\hat{U},C) \gets \DoubleBfsCycle(G,s, k)$\;
                \lIf{$C\neq null$}{\Return $C$}
                $G\gets G\setminus \hat{U}$\;
         
            }
        }
    \Return \codeNull\;
    \end{algorithm2e}

\begin{Remark}
   Notice that 
   $\SparseOrCycleM(G, m^{\frac{2}{k+1}},2,k)$ 
    either finds a $C_{\leq 2k+2}$ or removes vertices that are not on a $C_{\leq 2k}$, 
    such that for every $u\in G_1$  it holds that  $|E_u^2(G_1)|< m^{\frac{2}{k+1}}$. Using  $\DoubleBfsCycle$ instead of $\BallOrCycle$ in $\twoSparseOrCycle$ enables us in the case that a cycle is found 
    to bound the cycle length with $2k$ rather than $2k+2$, while still 
    maintaining the property that  $|E_u^2(G_1)|< m^{\frac{2}{k+1}}$, for every $u\in G_1$,  in the case that no cycle is found.
    
\end{Remark}

We prove the following lemma.

\begin{lemma}\label{L-twoSparseOrCycle} 
$\twoSparseOrCycle(G,k)$ satisfies the following: 
\begin{enumerate} [label=(\roman*)]
\item
If a cycle $C$ is returned then $\wt(C)\leq 2k$
\item
If a cycle is not returned then $|E_u^{2}| 
<
m^{\frac{2}{k+1}}$, for every $u\in G_1$
\item
If $u \in G_0 \setminus G_1$ then $u$ is not part of a $C_{\leq 2k}$ in $G_0$
\item 
$\twoSparseOrCycle(G,k)$ runs in $O(m^{1+\frac{k-1}{k+1}})$ time.
\end{enumerate}
\end{lemma}

\begin{proof}

    \begin{enumerate} [label=(\roman*)]
    \item 
    Since $\twoSparseOrCycle$ returns a cycle $C$ only if 
    a call to $\DoubleBfsCycle(G, s, k)$ returns a cycle $C$,
    it follows from Lemma~\ref{L-DoubleBfsCycle1} that $\wt(C)\leq 2k$.

    \item 
    Let $u\in G_1$. Since $u$ was not removed,  $u$ was considered in the for-each loop at some stage during the execution of $\twoSparseOrCycle$. At this stage
    $ \sum_{v \in N(s)} deg(v) < m^{\frac{2}{k+1}}$, as otherwise, since no cycle was returned, by Lemma~\ref{L-DoubleBfsCycle1} the call to $\DoubleBfsCycle$ with $u$ would have returned $\hat{U} = V_u^1$, so $u$ would have been removed while removing $\hat{U}$.
    Since $|E_s^2| \leq \sum_{v \in N(s)} deg(v)$ we  have $|E_u^{2}| < m^{\frac{2}{k+1}}$. As edges can only be removed during the run of $\twoSparseOrCycle$, we have $|E_u^{2}| < m^{\frac{2}{k+1}}$ also in $G_1$.

   \item 
    Since $u \in G_0 \setminus G_1$ it follows that there was a vertex $s$ such that $u\in \hat{U}$ after a call to $\DoubleBfsCycle(G,s,k)$ did not return a cycle. By Lemma~\ref{L-DoubleBfsCycle1}, no vertex in $\hat{U}$ is part of a $C_{\leq 2k}$ in $G$. Therefore, $u$ is not part of  a $C_{\leq 2k}$ in $G$.
    Since during the run of $\twoSparseOrCycle$ we remove only vertices that are not part of a $ C_{\leq2k}$, $u$ is not part of a $C_{\leq 2k}$ also in $G_0$.

    \item 
    Computing $ \sum_{v \in N(s)} deg(v)$ takes $O(|N(s)|) = O(deg(v))$ time, as all the degrees can be computed in advance in $O(m)$ time. 
    We compute this value for at most $n$ distinct vertices so the running time of this part is at most  $O(\sum_{v \in V} deg(v)) = O(m)$ in total.
    By Lemma~\ref{L-DoubleBfsCycle1}, running $\DoubleBfsCycle$ takes $O(n+m) = O(m)$ time.
    Each edge in $E_s^2$ contributes at most $2$ to the sum $ \sum_{v \in N(s)} deg(v)$, so 
    $ \sum_{v \in N(s)} deg(v) \leq 2|E_s^2|$ and 
    $\frac{1}{2} \sum_{v \in N(s)} deg(v) \leq |E_s^2|$.
    If a call to $\DoubleBfsCycle(G,s,k)$  did not return a cycle, then by Lemma~\ref{L-DoubleBfsCycle1}, the set $\hat{U}=V_s^{1}$ is returned. $\twoSparseOrCycle$ removes the set $\hat{U}=V_s^{1}$ and by doing so, the edge set $E_s^2$ is also removed. We charge each edge of $E_s^2$ with $O(m^{\frac{k-1}{k+1}})$. Thus the total cost that we charge for $s$ is $ O(|E_s^2|\cdot m^{\frac{k-1}{k+1}})
    \geq 
    O( \sum_{v \in N(s)} deg(v)
    \cdot m^{\frac{k-1}{k+1}})
    \geq
    O(m^{\frac{2}{k+1}}\cdot m^{\frac{k-1}{k+1}} )= O(m)
    $,
    which covers the $O(m)$ cost of $\DoubleBfsCycle(G,s,k)$.
    
    Since each edge can be charged and removed from $G$ at most once during the execution of $\twoSparseOrCycle$, the running time of $\twoSparseOrCycle$ is at most $O(m \cdot m^{\frac{k-1}{k+1}} )  = O(m^{1+\frac{k-1}{k+1}})$. \qedhere
\end{enumerate}
\end{proof}

Next, we use  $\twoSparseOrCycle$ to design a $2k$-hybrid algorithm called $\kHybrid$.
Notice first that if  $|E_u^{k-1}|< m^{\frac{k-1}{k+1}}$ for every $u\in V$, then 
it is straightforward to obtain an $O(m^{1+\frac{k-1}{k+1}})$-time  $2k$-hybrid algorithm, by  running $\dHybrid(G,k)$. 
Thus, in  $\kHybrid$ we  ensure that if we call  $\dHybrid$ then it holds for  every $u\in V$ that $|E_u^{k-1}|< m^{\frac{k-1}{k+1}}$. 
To do so, we  run $\twoSparseOrCycle$ and possibly $\SparseOrCycleM$. 
If no cycle was returned
then it holds that 
$|E_u^{k-1}|< m^{\frac{k-1}{k+1}}$ for every $u\in V$, and we can safely run  $\dHybrid$. 

$\kHybrid$ (see Algorithm~\ref{A-kHybrid}) gets a graph $G$ and an integer $k\geq 2$. 
$\kHybrid$ is composed of  three stages. 
In the first stage we call
$\twoSparseOrCycle(G,k)$. 
If $\twoSparseOrCycle$ returns a cycle $C$ then $\kHybrid$ stops and returns $C$, otherwise we proceed to the second stage. 
In the second stage, if $(k-1) \bmod 2 \neq 0$, we call 
$\SparseOrCycleM(G, m^{\frac{1}{k+1}}, 1, k)$.
If $\SparseOrCycleM$ returns a cycle $C$ then $\kHybrid$ stops and returns $C$, otherwise we proceed to the last stage. 
In the last  stage, we call $\dHybrid(G,k)$.
If $\dHybrid$ returns a cycle $C$ then $\kHybrid$ stops and returns $C$, otherwise $\kHybrid$ returns $\codeNull$.
In the next lemma we prove the correctness and analyze the running time of $\kHybrid(G,k)$.

 \begin{algorithm2e}[t]
    \caption{$\kHybrid(G,k)$}
    \label{A-kHybrid}
        $C \gets \twoSparseOrCycle(G,k)$\;
        
        \lIf{$C\neq \codeNull$}{\Return $C$}
    
        \BlankLine
    
       \If{$ (k-1) \bmod 2 \neq 0$}{
        
          $C\gets \SparseOrCycleM(G, m^{\frac{1}{k+1}} , 1, k)$

          \lIf {$C\neq \codeNull$}{\Return $C$}
        
        }
         \BlankLine
    
         $C\gets \dHybrid(G,k)$\;
          
          \lIf {$C\neq \codeNull$}{\Return $C$}

        \BlankLine
        \Return $\codeNull$\;
     
\end{algorithm2e}

\begin{lemma} \label{L-kHybrid} 
    $\kHybrid(G,k)$ either returns a $C_{\leq 2k}$ or determines that $g>2k$, in $O(m^{1+\frac{k-1}{k+1}})$ time.
\end{lemma}

\begin{proof}

First, $\kHybrid(G,k)$ returns a cycle $C$ only if $\twoSparseOrCycle(G,k)$, $\SparseOrCycleM(G, m^{\frac{1}{k+1}} , 1, k)$ or $\dHybrid(G,k)$ returns a cycle $C$. Therefore, if a cycle is returned then it follows from Lemmas~\ref{L-twoSparseOrCycle}, \ref{L-SparseOrCycleM}, or \ref{L-dHybrid-properties}, respectively,  that $\wt(C)\leq 2k$.

Second, we show that if no cycle was found then $g>2k$.
If no cycle was found then it might be that some vertices were removed from the graph.
A vertex $u$ can be removed either by  $\SparseOrCycleM$ or by $\twoSparseOrCycle$.  It follows from Lemma~\ref{L-SparseOrCycleM} and Lemma~\ref{L-twoSparseOrCycle}, that $u$ is not part of a $C_{\leq 2k}$ when $u$ is removed. Since only vertices that are not on a $C_{\leq 2k}$ are removed, every $C_{\leq 2k}$ that was in the input graph also belongs to the updated graph. 
We then call $\dHybrid(G,k)$ with the updated graph. Since we are in the case that no cycle was found, $\dHybrid$ did not return a cycle. It follows from Lemma~\ref{L-dHybrid-properties} that $g>2k$ in the updated graph, and therefore $g>2k$ also in the input graph.

Now we turn to analyze the running time of $\kHybrid$.
At the beginning, $\kHybrid$ calls $\twoSparseOrCycle$.
By Lemma~\ref{L-twoSparseOrCycle}, $\twoSparseOrCycle$ runs in $O(m^{1+\frac{k-1}{k+1}})$ time.
Let $G_1$ be the graph after the call to $\twoSparseOrCycle$. By Lemma~\ref{L-twoSparseOrCycle},
for every $u\in G_1$ we have $|E_u^{2}| < m^{\frac{2}{k+1}}$. 

Next, $\kHybrid$ checks if $(k-1) \bmod 2 \neq 0$. 
We divide the rest of the proof  to the case that $(k-1) \bmod 2 = 0$ and to the case that $(k-1) \bmod 2 \neq  0$. 
If $(k-1) \bmod 2 = 0$ then $\kHybrid$ calls $\dHybrid$. 
Since $(k-1) \bmod 2 = 0$ and since for every $u\in G_1$ we have $|E_u^{2}| < m^{\frac{2}{k+1}}$, it follows from Corollary~\ref{C-SumOfDenseM} that $|E_u^{k-1}| <  m^{\frac{k-1}{k+1}}$ for every $u\in G_1$. 
By Corollary~\ref{C-dHybrid-sparse-time}, the running time of $\dHybrid$ is $O(m\cdot m^{\frac{k-1}{k+1}}) = O(m^{1+\frac{k-1}{k+1}})$.

We now turn to the case that $(k-1) \bmod 2 \neq  0$. 
In this case it might be that $|E_u^{k-1}| \geq  m^{\frac{k-1}{k+1}}$ for some vertices $u\in G_1$. 
Therefore, we first call $\SparseOrCycleM(G_1, m^{\frac{1}{k+1}}, 1, k)$.
Notice that since we are in the case that $(k-1) \bmod 2 \neq 0$ it holds that $k \bmod 2 = 0$. 
Moreover, $|E_u^{2}| < m^{\frac{2}{k+1}}$ for every $u\in G_1$. Thus, it follows from Corollary~\ref{C-SumOfDenseM} that $|E_u^{k}| <  m^{\frac{k}{k+1}}$ for every $u\in G_1$. 
By Corollary~\ref{C-SparseOrCycle-sparse-time} if $|E_u^{k}| <  m^{\frac{k}{k+1}}$ for every $u\in G_1$ then the running time of $\SparseOrCycleM(G_1, m^{\frac{1}{k+1}}, 1, k)$ is $O(nm^{\frac{1}{k+1}}+mm^{\frac{k-1}{k+1}})=O(m^{1+\frac{k-1}{k+1}})$.

Let $G_2$ be the graph after $\SparseOrCycleM$ ends.  By Lemma~\ref{L-SparseOrCycleM}, for every $u\in G_2$ we have $|E_u^{1}| <  m^{\frac{1}{k+1}}$.
It follows from Corollary~\ref{C-SumOfDenseM} that $|E_u^{k-1}| <  m^{\frac{k-1}{k+1}}$.
Now $\kHybrid$ calls $\dHybrid$, and using  Corollary~\ref{C-dHybrid-sparse-time}
again, we get that the running time of $\dHybrid$ is $O(m\cdot m^{\frac{k-1}{k+1}}) = O(m^{1+\frac{k-1}{k+1}})$.

It follows from the above discussion that $\SrtCycSpr$ either returns a $C_{\leq 2k}$ or determines that $g>2k$, and
the running time is $O( m^{1+\frac{k-1}{k+1}})$.
\end{proof}

\subsection{\texorpdfstring{A $(+1)$-approximation of the girth}{A (+1)-approximation of the girth}}
\label{subsec:+1approx}

Next, we describe algorithm $\AdditiveGirthApprox$, which uses $\kHybrid$ and the framework described in Section~\ref{sec:introduction}, to obtain a $(+1)$-approximation of $g$, when $g\leq \log n$.
$\AdditiveGirthApprox$ (see Algorithm~\ref{A-AdditiveGirthApprox}) gets a graph $G$.
In $\AdditiveGirthApprox$, we set $k$ to $2$ and start a while loop.
In each iteration, we create a copy $G'$ of $G$ and call $\kHybrid(G',k)$. If  $\kHybrid$ finds a cycle $C$ then $\AdditiveGirthApprox$ stops and returns $C$, otherwise we increment $k$ by $1$ and continue to the next iteration. We prove the following theorem.

 \begin{algorithm2e}[t]
        \caption{$\AdditiveGirthApprox(G)$}
        \label{A-AdditiveGirthApprox}
        \CommentSty{\color{blue}}
        $k \gets 2$\;
        
        \While{$\textsf{true}$}
        {
            $G' \gets$ a copy of $G$\;
            $C\gets \kHybrid(G',k)$\;
            \lIf{$C \neq \codeNull$}{return $C$}
             
            $k \gets k +1$\;    
        }

\end{algorithm2e}

\begin{theorem}
\label{T-Main-Approx} 
Algorithm $\AdditiveGirthApprox(G)$ returns either a $C_{g}$ or a $C_{g+1}$,  
and runs in $\Ot(m^{1+\frac{\ell -1}{\ell+1}})$ time, where $g=2\ell$ or $g=2\ell -1$ and $2\leq  \ell < \log n$.\footnote{
We note that when $\ell \geq \log n$,  the $O(n^2)$-time, $(+1)$-approximation of Itai and Rodeh~\cite{itai1977finding} can be used.}
\end{theorem}

\begin{proof}
We first prove the bound on the approximation. 
$\AdditiveGirthApprox$ always returns a cycle in $G$ so it cannot be that a $C_{<g}$ will be returned. 
$\AdditiveGirthApprox$ starts with $k = 2$. 
It follows from Lemma~\ref{L-kHybrid} that as long as $k<\ell$ the calls to $\kHybrid$ will not return a cycle since the graph does not contain a $C_{\leq 2k}$. 
Consider now the iteration in which $k=\ell$. 
$\AdditiveGirthApprox$ calls to  $\kHybrid(G',\ell)$ where $G'=G$. It follows from Lemma~\ref{L-kHybrid} that $\kHybrid$ will either return a $C_{\leq 2\ell}$, or determine that $g>2\ell$. Since  we assume that $g\leq 2\ell$, $\kHybrid$ will return a $C_{\leq 2\ell}$
which is either a $C_g$ or a $C_{g+1}$ since $g=2\ell$ or $g=2\ell -1$.

We now turn to analyze the running time. 
Creating a copy of $G$ takes $O(m)$ time, and by Lemma~\ref{L-kHybrid} the running time of $\kHybrid(G',k)$ is $O(m^{1+\frac{k -1}{k+1}})$. Therefore, for every $k \leq \ell$, the running time of the iteration of the while loop with this value of $k$ is $O(m^{1+\frac{k -1}{k+1}})$.
From the previous part of this proof
it follows that the last iteration of the while loop  is when $k = \ell$, thus, the running time of $\AdditiveGirthApprox$ is 
$O(m^{1+\frac{1}{3}}+m^{1+\frac{2}{4}}+\cdots+m^{1+\frac{\ell-2}{\ell}}+m^{1+\frac{\ell-1}{\ell+1}})\leq O(\ell\cdot m^{1+\frac{\ell-1}{\ell+1}})
$.
Therefore, when $\ell < \log n$, the running time is $\Ot(m^{2-\frac{2}{\ell+1}})$.
\end{proof}

\begin{Remark}

Dahlgaard, Knudsen and St\"{o}ckel~\cite{dahlgaard2017cappedkwalks} showed that a $C_{2k}$, if exists, can be found in $O(k^{O(k)}\cdot m^{\frac{2k}{k+1}})$ time. 
It is possible to use their $C_{2k}$ detection algorithm to obtain a 
 $(+1)$-approximation for the girth in $\Ot(\ell^{O(\ell)}\cdot m^{1+\frac{\ell -1}{\ell+1}})$ time, where $g=2\ell$ or $g=2\ell -1$,
as mentioned in Section~\ref{sec:overview}. This $(+1)$-approximation is obtained as follows. 
We run the detection algorithm with increasing values of $k$, starting with $k=2$. 
If a $C_{2k}$ is detected we stop and return the detected cycle. In such a case $2k$ is a $(+1)$-approximation for the girth. 
If there is no $C_{2k}$ then it follows from the analysis of~\cite{dahlgaard2017cappedkwalks} that  a $C_{2k-1}$, if exists, can be detected within the same running time using a slight modification of their algorithm.
If a $C_{2k-1}$ is detected we stop and return the detected cycle. In such a case $2k-1$ is the girth.
If neither a $C_{2k}$ nor a $C_{2k-1}$ is found we increment $k$ by $1$.
Since a cycle must be found when $k=\ell$,  we obtain a $(+1)$-approximation for the girth in $\Ot(\ell^{O(\ell)}\cdot m^{1+\frac{\ell -1}{\ell+1}})$ time, where $g=2\ell$ or $g=2\ell -1$. 

A disadvantage of the algorithm described above is the $\ell^{O(\ell)}$ factor which in our new algorithm does not exist. 
To better appreciate our contribution we compare 
our $\kHybrid$ algorithm (for $C_{\leq 2k}$ detection) and the algorithm of \cite{dahlgaard2017cappedkwalks} (for $C_{2k}$ detection), that are used for obtaining the approximation.

Both algorithms first sparsify the graph (or find a cycle).
The algorithm of \cite{dahlgaard2017cappedkwalks} first handles  high degree vertices, by running an $O(m)$-time $C_{2k}$ detection algorithm from each high degree vertex (vertex of degree $>m^{2/(k+1)}$). If a $C_{2k}$ is found then a $C_{2k}$ is returned, otherwise the high degree vertices (and their adjacent edges) are removed.
This is possible within the running time since there cannot be too many high degree vertices. 
In our algorithm, we handle first vertices with  at least $m^{2/(k+1)}$ edges up to distance $2$ rather than $1$. This is possible using our $\DoubleBfsCycle$ algorithm, since we can search for a $C_{\leq 2k}$ from a vertex and also its neighbors in $O(m)$ time.
If a $C_{2k}$ is found then a $C_{2k}$ is returned, otherwise the vertex and its neighbors are removed (therefore also the edges up to distance $2$ from the vertex).
Then, if required, we handle also high degree vertices and remove them.

After the sparsification stage, in
\cite{dahlgaard2017cappedkwalks} they prove that when the maximum degree is bounded, 
if there is no $C_{2k}$ then the number of \textit{capped $k$-walks} (walks of length $k$ that visit only nodes according to a fixed ordering) in the graph is bounded by $O(k^{O(k)}m^{1+(k-1)/(k+1)})$, where the $k^{O(k)}$ factor follows from their analysis.
They use this property to build a series of subgraphs in which the total number of edges is bounded, and therefore it is possible to search for a $C_{2k}$ in these graphs efficiently.
 In our algorithm, after the sparsification stage it holds that the total number of edges up to distance $k$ from all the vertices is bounded by $O(m^{1+(k-1)/(k+1)})$ (avoiding the $k^{O(k)}$ factor), and we can search efficiently for a  $C_{\leq 2k}$.
    
\end{Remark}

\section{A general hybrid algorithm}
\label{sec:one-alg}

Algorithm $\kHybrid(G,k)$, presented in the previous section, either returns a $C_{\leq 2k}$ or determines that $g>2k$, in $O(m^{1+\frac{k-1}{k+1}})$ time.
In this section we introduce an additional parameter $2\leq \alpha\leq k$ and present a $(2k,2\alpha)$-hybrid algorithm  that  either returns a $C_{\leq 2k}$ or determines that $g>2\alpha$, in $O((\frac{k+1}{\alpha-1}+\alpha) \cdot m^{1+\frac{\alpha-1}{k+1}})$ time. 
In Section~\ref{sec:girth-approx} we use the $(2k,2\alpha)$-hybrid algorithm to present two tradeoffs for girth approximation.

In order to obtain the $(2k,2\alpha)$-hybrid algorithm we first  extend algorithm $\DoubleBfsCycle$.
Then, we use the extended $\DoubleBfsCycle$ together with  additional tools that we develop to either return a $C_{\leq 2k}$ or sparsify dense regions of the graph, so that we  can check whether $g>2\alpha$ (or return a $C_{\leq 2k}$) in $O(m^{1+\frac{\alpha-1}{k+1}})$ time, by running $\dHybrid(G,\alpha)$.

\subsection{\texorpdfstring{Extending $\DoubleBfsCycle$}{Extending NeighborhoodBallOrCycle}}

In algorithm $\DoubleBfsCycle$ we mark
vertices that can be removed from the graph, by using 
the property that if $\BallOrCycle(v,k)$ did not return a $C_{\leq 2k}$ then $v$ is not on a $C_{\leq 2k}$.
In \cite{kadria2022algorithmic}, they introduce an additional parameter $\alpha$ and used the following extended version of this property:
If $\BallOrCycle(v,k)$ did not return a $C_{\leq 2k}$ then  no vertex of $V_v^{k-\alpha}$ is  on a $C_{\leq 2\alpha}$. 
We use the same approach and modify $\DoubleBfsCycle$ to get an additional integer parameter $\alpha$ such that $2\leq\alpha\leq k$.
After each call to $\BallOrCycle(G',v,k)$, where $v \in N_s$, if no cycle was found we add $V_v^{k-\alpha}$, instead of $v$,  to $\hat{U}$. 
The modified pseudo-code appears in Algorithm~\ref{A-DoubleBfsCycle}.
We rephrase Lemma~\ref{L-DoubleBfsCycle1} to suit this modification.

\begin{algorithm2e}[t]
\caption{$\DoubleBfsCycle(G, s, k,\alpha)$}
\label{A-DoubleBfsCycle}

    $\hat{U} \gets \emptyset$\;
  
    $(V_s^k,C) \gets \BallOrCycle(s, k)$\;
     
    \lIf {$C\neq \codeNull$}{\Return $(\emptyset, C)$}
    \BlankLine
       
    $\hat{U}\gets \hat{U} \cup \{s\}$\;

    $N_s \gets N(s)$\;
    
    $G \gets G \setminus \{s\}$\;
    
     \ForEach{$(s,v) \in N_s$}{
         $(V_v^k,C) \gets \BallOrCycle(v, k)$\;
         
         \lIf {$C\neq \codeNull$}{\Return $(\emptyset, C)$}
         
        $ \hat{U}\gets \hat{U} \cup V_{v}^{k-\alpha}$ \;           
         
     }

    add $s$ and the edge set $\{(s,v) \mid v\in N_s\}$ back to $G$;
        
    \Return $(\hat{U}, \codeNull);$
 
\end{algorithm2e}

\begin{lemma}  \label{L-DoubleBfsCycle}
    Let $2\leq \alpha \leq k$. If  $\DoubleBfsCycle(G,s,k,\alpha)$ finds a cycle $C$ then $\wt(C)\leq 2k$. Otherwise, no vertex in $V_s^{k-\alpha+1}$ is part of a $C_{\leq 2\alpha}$ in $G$, and the set $\hat{U}= V_s^{k-\alpha+1}$ is returned. 
\end{lemma}

\begin{proof}
    The proof that if a cycle $C$ is returned then $\wt(C)\leq 2k$ is as in Lemma~\ref{L-DoubleBfsCycle1}.
    It is left to show that if no cycle is found then no vertex in $V_s^{k-\alpha+1}$ is part of a $C_{\leq 2\alpha}$ in $G$, and the set $\hat{U}= V_s^{k-\alpha+1}$ is returned.

    From Lemma~\ref{L-not-part-of} it follows that if $\BallOrCycle(G',v,k)$ did not  return a cycle then no vertex in $V_v^{k-\alpha}$ is part of a $C_{\leq 2\alpha}$ in $G'$, and therefore also in $G$, since $G'=G \setminus \{s\}$ and  since $s$ itself is not on a $C_{\leq 2\alpha}$ as otherwise $\BallOrCycle(G',v,k)$ would not have been called. 
    Thus, if no cycle was found, we have $\hat{U}=V_s^{k-\alpha+1}(G)$, as in this case $\hat{U}$ contains $s$ and $V_v^{k-\alpha}(G')$ for each $v\in N(s)$, which equals to $V_s^{k-\alpha+1}(G)$.
\end{proof}

For the running time, we note that the sets $V_{v}^{k}$ are computed during the execution of $\BallOrCycle(G',v,k)$, for every $(s,v)\in N_s$ for which no cycle was found. Their total size is also $O(n+m)$, and we can obtain from them the sets $V_{v}^{k-\alpha}$ and add these sets to $\hat{U}$ in $O(n+m)$ time. Therefore, the modified $\DoubleBfsCycle$ also runs in $O(n+m)$ time.

\subsection{\texorpdfstring{A $(\max\{2k,g\},\gt)$-hybrid algorithm}{A (max{2k,g},tilde(g))-hybrid algorithm}}

In this section we present a $(\max\{2k,g\},2\alpha)$-hybrid algorithm called $\OneAlg$, where $\alpha = \lceil \frac{\gt}{2}\rceil $. $\OneAlg$ (see Algorithm~\ref{A:OneAlg}) gets 
a graph $G$ and two integers $\alpha, k\geq 2$.
If $m \geq 1+\lceil n\cdot (1+n^{1/k})\rceil$ then we run algorithm $\SrtCycDns(G,k)$, which is based on algorithm $\DegenerateOrCycle$ of \cite{kadria2022algorithmic} (for completeness, we include a pseudo-code in Algorithm~\ref{A:srtCycDns}).
If $m < 1+\lceil n\cdot (1+n^{1/k})\rceil$ then the main challenge is when
 $\alpha \leq k< \frac{n}{2}$. In this case we run algorithm $\SrtCycSpr(G,k,\alpha)$ (described later). 
The 
cases that $k\geq \frac{n}{2}$ or  $k<\alpha$ are relatively simple and 
are 
treated by algorithm $\SpCase(G,k,\alpha)$ (see Section~\ref{subsec:special}). 
We summarize 
the properties of 
$\OneAlg(G,k,\alpha)$
in the next theorem.

  \begin{algorithm2e}[t]
            \caption{$\OneAlg(G,k,\alpha)$}
            \label{A:OneAlg}
    
            \If{$m \geq 1+\lceil n\cdot (1+n^{1/k})\rceil$}
            {
                       \Return $\SrtCycDns(G,k)$\;
            }
            \BlankLine
        
            \If{$\alpha \leq k< \frac{n}{2} $}
            {
            \Return $\SrtCycSpr(G,k,\alpha)$\;
            }
        
            \Return $\SpCase(G,k,\alpha)$
        
    \end{algorithm2e}

\begin{algorithm2e}[t]
    \caption{$\SrtCycDns(G,k)$}
    \label{A:srtCycDns}
    
        $G'\gets (V',E')$ is the edge induced subgraph formed by an arbitrary subset of $1+\lceil n\cdot (1+n^{1/k})\rceil$ edges; 

        \BlankLine
        
        $S \gets \{a \in V' \mid deg_{G'}(a)\leq 1+n^{1/k}\}$\;

        \While{$S \neq \emptyset$}
        {
            pick $a$ from $S$\;     
            \ForEach{$(a,b)\in E'$}{\lIf {$deg_{G'}(b)=2+n^{1/k}$}{$S \gets S \cup \{b\}$}}
            remove $a$ from $G'$ and from $S$\;
        }
        
        $(*, C) \gets \BallOrCycle(G',w,k)$ with some $w \in V'$\;
        \Return $C$\tcp*{$C$ is a cycle of length $\leq 2k$}   
    \end{algorithm2e}
    
\begin{theorem}
\label{Thm-One-Alg}
Let $\alpha,k\geq2$  be   integers.  $\OneAlg(G,k,\alpha)$ runs whp in $\Ot((\frac{k+1}{\alpha-1}+\alpha) \cdot \min\{m^{1+\frac{\alpha-1}{k+1}},n^{1+\frac{\alpha}{k}}\})$ time 
and either returns a $C_{\leq \max\{2k,g\}}$, or determines that $g > 2\alpha$.  
\end{theorem}

The next corollary follows from  Theorem~\ref{Thm-One-Alg}, when $\alpha \leq k$.

\begin{corollary}
\label{C-shortCycleBigK}
    Let $2\leq \alpha \leq k$.
    Algorithm $\OneAlg(G,k,\alpha)$ runs whp in $\tilde{O}((\frac{k+1}{\alpha-1}+\alpha) \cdot \min \{m^{1+\frac{\alpha-1}{k+1}},n^{1+\frac{\alpha}{k}}\})$
    time and either returns a $C_{\leq 2k}$, or determines that $g > 2\alpha$.
\end{corollary}

In the rest of this section, we present  the proof of Theorem~\ref{Thm-One-Alg}.
As follows from~\cite{kadria2022algorithmic},  if $m \geq 1+\lceil n\cdot (1+n^{1/k})\rceil$ then  $\SrtCycDns(G,k)$ returns in  $O(\min\{m,n^{1+1/k}\})$ time  a $C_{\leq 2k}$. We now consider the case in which $m < 1+\lceil n\cdot (1+n^{1/k})\rceil$.
We prove in Section~\ref{subsec:special} that  $\SpCase(G,k,\alpha)$ satisfies the claim of Theorem~\ref{Thm-One-Alg}, when $k\geq\frac{n}{2}$ or $k<\alpha$.

Our main technical contribution is algorithm $\SrtCycSpr$ that handles the case of
$\alpha \leq k < \frac{n}{2}$.
Notice  that if  $|E_u^{\alpha-1}|< m^{\frac{\alpha-1}{k+1}}$ for every $u\in V$, then 
$\dHybrid(G,\alpha)$ is a $(2k,2\alpha)$-hybrid algorithm that either finds a $C_{\leq 2\alpha}$ (which is also a $C_{\leq 2k}$ as $\alpha \leq k$) or determines that $g>2\alpha$, in $O(m^{1+\frac{\alpha-1}{k+1}})$ time. 
Thus, in  $\SrtCycSpr$ we  ensure that if we call  $\dHybrid$,
the property that $|E_u^{\alpha-1}|< m^{\frac{\alpha-1}{k+1}}$, for  every $u\in V$ (whp), holds.
To do so, we  run $\BFSSampleM$ and possibly $\HandleRemainderM$. 
If no cycle was returned, the property holds, 
and we can safely run  $\dHybrid$.

$\SrtCycSpr$ (see Algorithm~\ref{A-BigKM}) gets a graph $G$ and 
two integers $\alpha, k\geq 2$ such that $\alpha \leq k < \frac{n}{2}$,
and is composed of three stages.
In the first stage we call
$\BFSSampleM(G,k,\alpha)$ (described later). 
If $\BFSSampleM$ returns a cycle $C$ then $\SrtCycSpr$ stops and returns $C$,
otherwise we proceed to the second stage. 
In the second stage, if $(k+1) \bmod (\alpha-1) \neq 0$, we call 
$\HandleRemainderM(G,k,\alpha)$ (also described later). 
If $\HandleRemainderM$ returns a cycle $C$ then $\SrtCycSpr$ stops and returns $C$, otherwise we proceed to the last stage. 
In the last  stage, we call $\dHybrid(G,\alpha)$.
If $\dHybrid$ returns a cycle $C$ then $\kHybrid$ stops and returns $C$, otherwise $\SrtCycSpr$ returns $\codeNull$.

 \begin{algorithm2e}[t]

    \caption{$\SrtCycSpr(G,k,\alpha)$}
    \label{A-BigKM}
       
        $C \gets \BFSSampleM(G,k,\alpha)$\;
        
        \lIf{$C\neq \codeNull$}{\Return $C$}
    
        \BlankLine
    
        \If{$(k+1) \bmod (\alpha-1)  \neq 0$}{
        $C \gets \HandleRemainderM(G,k,\alpha)$\;
        
        \lIf{$C\neq \codeNull$}{\Return $C$}
        }
     
         \BlankLine

          $C \gets \dHybrid(G,\alpha)$\;
          
          \lIf {$C\neq \codeNull$}{\Return $C$}

        \BlankLine
        \Return $\codeNull$\;

    \end{algorithm2e}

Next, we give a high level description of  $\BFSSampleM$.
The goal of $\BFSSampleM$ is to either sparsify the graph 
without removing any $C_{\leq 2\alpha}$, 
or to report a $C_{\leq 2k}$. 
For simplicity  assume that $(k+1) \bmod (\alpha-1) =0$.
In such a case, if $\BFSSampleM$ does not report a $C_{\leq 2k}$, then the graph after $\BFSSampleM$ ends contains every $C_{\leq 2\alpha}$  that  
was in the original graph, and satisfies, whp, the following sparsity property: 
For every $u\in V$ it holds that $|E_u^{\alpha-1}| < m^{\frac{\alpha-1}{k+1}}$. 

This implies that in $\BFSSampleM$ we need to find every $u\in V$ 
which is in a dense region with
$|E_u^{\alpha-1}| \geq m^{\frac{\alpha-1}{k+1}}$, and to check if $u$ is in a $C_{\leq 2\alpha}$, so that if not we can remove $u$.
Finding every such $u$ is possible within the time limit by running $\IsDense(G,u,m^{\frac{\alpha-1}{k+1}}, \alpha - 1)$ for every $u \in V$.
The problem  is that checking whether $u$  is on a $C_{\leq 2\alpha}$ for every such $u$ is too costly since there might be $n$ such vertices, and this check costs $O(n)$ using $\BallOrCycle(G,u,\alpha)$. 

One  way to overcome this problem is to  sample  an edge set $S$ of size $\tilde{\Theta}(m^{1-\frac{\alpha-1}{k+1}})$  that hits the $m^{\frac{\alpha-1}{k+1}}$ closest edges of each vertex,
and then use $S$ to detect the vertices in the dense regions that are not on a $C_{\leq 2\alpha}$.
In $\BFSSampleM$ we use a \textit{detection process} in which we call $\BallOrCycle$ or $\DoubleBfsCycle$ from the endpoints of $S$'s edges, and then, if no cycle was found, we use the information obtained from this call to identify vertices that are not on a  $C_{\leq 2\alpha}$. 
The detection process either detects vertices that are not on a $C_{\leq 2\alpha}$ and  can be removed, or reports a $C_{\leq 2k}$. 
However, it is not clear how to 
implement this detection process efficiently, since just running $\BallOrCycle$ from the endpoints of $S$'s edges takes $O(n m^{1-\frac{\alpha-1}{k+1}})$ time which might be too much.
Our solution is an \textit{iterative} sampling procedure that starts with a smaller hitting set of edges, of size $\tilde{\Theta}(m^{\frac{\alpha-1}{k+1}})$.
For such a hitting set we 
can run our detection process.
If a $C_{\leq 2k}$ was not reported, then we remove the appropriate vertices and sparsify the graph without removing any $C_{\leq 2\alpha}$.
When the graph is sparser, the running time of our detection process 
becomes faster. Thus, in the following iteration we can sample a larger hitting set for which we run this process, and either return a  $C_{\leq 2k}$ or sparsify the graph further for the next iteration.
We continue the iterative sampling procedure until we get to the required sparsity property in which $|E_u^{\alpha-1}| < m^{\frac{\alpha-1}{k+1}}$ for every $u\in V$ (whp).

We remark that in the first iteration of $\BFSSampleM$, the detection process calls $\DoubleBfsCycle$, while in the rest of the iterations $\BallOrCycle$ is called. 
The use of $\DoubleBfsCycle$ allows us, in the case that no $C_{\leq 2k}$ is reported, 
to bound $|E_v^d|$ with $m^{\frac{d}{k+1}}$ for every $v\in V$, rather than $|E_v^{d-1}|$.
This is used to achieve the required sparsity property.
Since 
$\DoubleBfsCycle$ runs in $O(m)$ time we can only use it in the first iteration when the sampled set is small enough. In the rest of the iterations we use $\BallOrCycle$ instead.

We now formally describe 
$\BFSSampleM$.
$\BFSSampleM$ (see Algorithm~\ref{BFSSample-M}) gets 
a graph $G$ and two integers
$\alpha , k \geq 2$
such that $\alpha \leq k < \frac{n}{2}$. We 
first set $y$ to $\lceil \frac{k+1}{\alpha-1}\rceil-1$.
Then, we start the main for loop that has at most $y$ iterations.
In the $i$th iteration,
we initialize $\hat{V}$ to $\emptyset$ and sample a set $S_i\subseteq E$ of size $\Tilde{\Theta}(m^\frac{i\cdot(\alpha-1)}{k+1})$.
Next, 
we scan the 
endpoints in $V(S_i)$ using an inner for-each loop.

If $i=1$, we call $\DoubleBfsCycle(G,s,k,\alpha)$ from every endpoint $s\in V(S_1)$.
$\DoubleBfsCycle$ returns either a cycle or a set of vertices $\hat{U}$. If $\DoubleBfsCycle$ returns a cycle then the cycle is returned by $\BFSSampleM$. 
Otherwise, we add $\hat{U}$ to $\hat{V}$.

If $i>1$ then we call $\BallOrCycle(s, k_i)$, where $k_i=(k+1)-(i-1)\cdot(\alpha-1)$,  from every endpoint $s\in V(S_i)$.
If a cycle is found by $\BallOrCycle$ then the cycle is returned by $\BFSSampleM$. If $\BallOrCycle$ does not return a cycle then we add $V_{u_j}^{k_i -\alpha}$ to $\hat{V}$.

Right after the inner for-each loop ends, we remove $\hat{V}$ from  $G$, and continue to the next iteration of the main for loop.
If no cycle was found after $y$ iterations, we return $\codeNull$.
Let $\ell \leq y$ be the last iteration in which vertices were removed. Let $\hat{V_i}$ be the set of vertices that were removed during the $i$th iteration, where $\hat{V_i} = \emptyset$ for $i>\ell$.
Let $G_0$ ($G_1$) be  $G$  before (after) running $\BFSSampleM$.
Figure~\ref{fig:BFSSampleMFirstIteration} and Figure~\ref{fig:BFSSampleMIteration} illustrate the key steps of the first and the following iterations of $\BFSSampleM$, respectively. 
We summarize the properties $\BFSSampleM$ in the next lemma. 

\begin{algorithm2e}[t]
\caption{$\BFSSampleM(G,k,\alpha)$}
\label{BFSSample-M}
    $y \gets \lceil \frac{k+1}{\alpha-1}\rceil-1$\; 
    
    \For{$i \gets 1$ to $y$}{
        sample a set $S_i$ of $\Tilde{\Theta}(m^\frac{i\cdot(\alpha-1)}{k+1})$ edges\;
        
        $\hat{V} \gets \emptyset$\;
        
        \ForEach{$s \in V(S_i)$}{

                \If{$i=1$}{ 
                    $(\hat{U},C) \gets \DoubleBfsCycle(G,s, k, \alpha)$\;
                    
                    \lIf {$C\neq \codeNull$}{\Return $C$}
        	      $\hat{V}\gets \hat{V} \cup    \hat{U} $\;
                }%
                \Else{ %
                    $k_i \gets (k+1)-(i-1)\cdot(\alpha-1)$\;
                    $(V_{s}^{k_i},C) \gets \BallOrCycle(s, k_i)$\;
                    
                    \lIf {$C\neq \codeNull$}{\Return $C$}
        	      $\hat{V}\gets \hat{V} \cup    V_{s}^{k_i-\alpha} $;
                }%
        }
       $G\gets G \setminus \hat{V}$\;
    }
    \Return $\codeNull$;
 
\end{algorithm2e}

\input{tikzpictureBFSSAMPLE}

\begin{lemma}\label{L-BFSSample-M} 
$\BFSSampleM(G,k,\alpha)$ satisfies the following: 
\begin{enumerate} [label=(\roman*)]
\item
If a cycle $C$ is returned then $\wt(C)\leq 2k$
\item
If a cycle is not returned then $|E_u^{(k+1)-y \cdot (\alpha-1)}| 
<
m^{1-\frac{y \cdot (\alpha-1)}{k+1}}$, for every $u\in G_1$, whp
\item
If $u \in G_0 \setminus G_1$ then $u$ is not part of a $C_{\leq 2\alpha}$ in $G_0$
\item $\BFSSampleM(G,k,\alpha)$ runs in $\Ot(\lfloor \frac{k+1}{\alpha-1}\rfloor \cdot
m^{1+\frac{\alpha-1}{k+1}})
$ time, whp.  
\end{enumerate}
\end{lemma}

\begin{proof} 
\begin{enumerate} [label=(\roman*)]
    \item 
    We first show that if a cycle $C$ is returned then $\wt(C)\leq 2k$. 
    $\BFSSampleM$ returns a cycle only if 
    a call to $\DoubleBfsCycle(G, s, k, \alpha)$ returns a cycle or a call to $\BallOrCycle(s, k_i)$, where $2\leq i \leq y$ and $k_i=(k+1)-(i-1)\cdot(\alpha-1)$, returns a cycle. 
    By Lemma~\ref{L-DoubleBfsCycle}, if $\DoubleBfsCycle(G, s, k, \alpha)$ returns a cycle $C$ then $\wt(C)\leq 2k$.
    By Lemma~\ref{L-BallOrCycle}, if $\BallOrCycle(s,k_i)$ returns a cycle $C$ then $\wt(C) \leq 2 k_i$.
    Since $i>1$ is an integer, and since $\alpha\geq 2$, we have
    $k_i = (k+1)-(i-1)\cdot(\alpha-1) \leq (k+1)-(\alpha-1)\leq k$. Therefore,  $\wt(C) \leq 2k$.

    \item 
      Next, we show that if a cycle is not returned then $|E_u^{(k+1)-y \cdot (\alpha-1)}| < m^{1-\frac{y \cdot (\alpha-1)}{k+1}}$, for every $u\in G_1$, whp.
      Since $k\geq \alpha$, we have $y\geq 1$ and there is at least one iteration.

     Consider the $i$th iteration of the main loop. 
     We show that if no cycle was found during the $i$th iteration
     then after the $i$th iteration every 
     $u\in G_0\setminus (\cup_{j=1}^{i}\hat{V}_{j})$
     satisfies the following property: $|E_{u}^{(k+1)-i\cdot(\alpha-1)}| < m^{1-\frac{i\cdot(\alpha -1)}{k+1}}$,  whp.

    Let $u\in G_0\setminus (\cup_{j=1}^{i}\hat{V}_{j})$. 
    By Lemma~\ref{L-whpM}, the set $S_i$ hits the $m^{1-\frac{i\cdot(\alpha-1)}{k+1}}$ closest edges of every vertex of $G_0\setminus (\cup_{j=1}^{i-1}\hat{V}_{j})$, whp. 
    Assume that $S_i$ is indeed such a hitting set, and assume, towards a contradiction, that after the $i$th iteration $|E_{u}^{(k+1)-i\cdot(\alpha-1)}| \geq m^{1-\frac{i\cdot(\alpha -1)}{k+1}}$.
    Since $G_0\setminus (\cup_{j=1}^{i}\hat{V}_{j}) \subseteq G_0\setminus (\cup_{j=1}^{i-1}\hat{V}_{j})$ we have $u\in G_0 \setminus (\cup_{j=1}^{i-1}\hat{V}_{j})$. 
    Now, since  $u\in G_0 \setminus (\cup_{j=1}^{i-1}\hat{V}_{j})$ and since the graph is not updated in the inner for-all loop, it follows that there is an edge $(u_1,u_2) \in S_i$ such that
    $(u_1,u_2) \in E_{u}^{(k+1)-i\cdot(\alpha-1)}$.
     By the definition of $E_{u}^{(k+1)-i\cdot(\alpha-1)}$, either $u_1$ or $u_2$, denoted with $s$, satisfies that $d(u,s)\leq (k+1)-i\cdot(\alpha-1) -1$. By the definition of $V(S_i)$, we know that $s\in V(S_i)$.

     When $i=1$, we have $d(u,s)\leq  (k+1) - i
     \cdot(\alpha-1) -1 = k-\alpha+1$. Therefore, $u\in V_{s}^{k-\alpha+1}$. At the first iteration, when $i=1$, the input graph has not changed yet so $G=G_0$. Since no cycle was found by $\DoubleBfsCycle(G_0,s,k,\alpha)$, it follows from Lemma~\ref{L-DoubleBfsCycle} that $\hat{U} = V_{s}^{k-\alpha+1}$. Therefore,  $u\in \hat{U}$, and $u$ is added to $\hat{V_i}$ after the call to $\DoubleBfsCycle(G_0,s,k,\alpha)$. Hence, $u\notin G_0\setminus (\cup_{j=1}^{i}\hat{V}_i)$, a contradiction.  
     
     We now handle the case that $i>1$. It holds that $d(u,s)\leq  (k+1)-i\cdot(\alpha-1)-1 = (k+1)-(i-1)\cdot(\alpha-1)- \alpha =k_i-\alpha$.
    Hence, $u\in V_{s}^{k_i-\alpha}$.
    Since no cycle was found during the $i$th iteration, $u$ is added to $\hat{V_i}$ after the call to 
    $\BallOrCycle(s, k_i)$,
    and therefore, $u\notin G_0\setminus (\cup_{j=1}^{i}\hat{V}_i)$, a contradiction.  
    
    Now, if $\BFSSampleM$ does not return a cycle we get for $i=y$ that if $u\in G_0\setminus (\cup_{j=1}^{y}\hat{V}_{j})=G_1$
    then 
    $|E_u^{(k+1)-y \cdot (\alpha-1)}| < m^{1-\frac{y \cdot (\alpha-1)}{k+1}}$, whp. 

   \item 
    Next, we prove that if $u \in G_0 \setminus G_1$ then $u$ is not part of a $C_{\leq 2\alpha}$ in $G_0$.
    Since $u \in G_0 \setminus G_1$ and since $G_1 = G_0 \setminus (\cup_{j=1}^y\hat{V}_j$) it holds that $u\in \hat{V}_i$, where $1\leq i \leq y$.
    
    If $i=1$ then since $u\in \hat{V}_i$ it follows that there was in the first iteration a vertex $s$ such that $u\in \hat{U}$ after a call to $\DoubleBfsCycle(G_0,s,k,\alpha)$ did not return a cycle. By Lemma~\ref{L-DoubleBfsCycle}, $\hat{U} = V_{s}^{k-\alpha+1}$ and no vertex in $V_s^{k-\alpha+1}$ is part of a $C_{\leq 2\alpha}$ in $G_0$. Therefore, $u$ is not part of  a $C_{\leq 2\alpha}$ in $G_0$.

    If $i>1$ then since $u\in \hat{V}_i$ it follows that there was in the $i$th iteration a vertex $s$ such that 
    $u\in V_{s}^{k_i-\alpha}$
    after a call to $\BallOrCycle(s, k_i)$ did not return a cycle.
    As $\BallOrCycle(s, k_i)$ did not return a cycle, by Lemma~\ref{L-BallOrCycle} we know that $B(s,k_i)$ is a tree.
    It follows from Lemma~\ref{L-not-part-of} that no vertex in $V_{s}^{k_i-\alpha}$, and in particular $u$, is part of a $ C_{\leq2\alpha}$ in $G_0 \setminus (\cup_{j=1}^{i-1}\hat{V}_{j})$.
    Since during the run of $\BFSSampleM$ we remove only vertices that are not part of a $ C_{\leq2\alpha}$, $u$ is not part of a $C_{\leq 2\alpha}$ also in $G_0$.

    \item
    Finally, we show that $\BFSSampleM$ runs in  $\Ot(\lfloor \frac{k+1}{\alpha-1}\rfloor \cdot
    m^{1+\frac{\alpha-1}{k+1}})
    $ time, whp.
    To do so, we show that the running time of the $i$th iteration of the main for loop is whp $\Ot(m^{1+\frac{\alpha-1}{k+1}})$.

    We start with the first iteration, in which $i=1$. The size of $S_1$ is $\Tilde{\Theta}(m^{\frac{\alpha-1}{k+1}})$. 
    The size of $V(S_1)$ is at most $2\cdot |S_1|$. 
    For every $s\in V(S_1)$ we run $\DoubleBfsCycle(G_0,s, k, \alpha)$.
    By Lemma \ref{L-DoubleBfsCycle-time}, running $\DoubleBfsCycle$ from $s$ costs $O(n+m)=O(m)$.
    Adding $\hat{U}$ to $\hat{V}$ costs $O(n) = O(m)$. 
    Therefore, the total running time for all $s\in V(S_1)$ is at most $ 2\cdot |S_1|\cdot O(m) = 
       \Tilde{O}(m^{\frac{\alpha-1}{k+1}}\cdot m) = 
       \Tilde{O}(m^{1+\frac{\alpha-1}{k+1}})$.

    Now we assume that $i>1$.
    We proved in (ii) that if $i > 1$ and $u\in G_0\setminus (\cup_{j=1}^{i-1}\hat{V}_{j})$ then after the $(i-1)$th iteration, if no cycle was found, we have $|E_{u}^{(k+1)-(i-1)\cdot(\alpha-1)}| < m^{1-\frac{(i-1)\cdot(\alpha -1)}{k+1}}$,  whp.
    By Lemma~\ref{L-BallOrCycle}, for every $s\in V(S_i)$, the cost of running 
     $\BallOrCycle(s, k_i)$ is $O(|V_{s}^{k_i}|) = 
     O(|E_{s}^{k_i}|)=
     O(|E_{s}^{(k+1)-(i-1)\cdot(\alpha-1)}|)$. In our case this is at most $O(m^{1-\frac{(i-1)\cdot(\alpha -1)}{k+1}})$.
     As the size of $S_i$ is $\Tilde{\Theta}(m^\frac{i\cdot(\alpha-1)}{k+1})$ and the size of $V(S_i)$ is at most $2\cdot |S_i|$, the total running time of the calls to $\BallOrCycle$ for every $s\in V(S_i)$ is, whp,
         $
         2\cdot |S_i| \cdot O(m^{1-\frac{(i-1)\cdot(\alpha -1)}{k+1}}) =
        \Tilde{\Theta}(m^\frac{i\cdot(\alpha-1)}{k+1}) \cdot O(m^{1-\frac{(i-1)\cdot(\alpha -1)}{k+1}}) =\tilde{O}(m^{1+\frac{\alpha-1}{k+1}}).
        $

       The cost of adding $V_{s}^{k_i-\alpha}$ to $\hat{V}$ is $O(|V_{s}^{k_i-\alpha}|) $. This is at most $O(|V_{s}^{k_i}|) = O(m^{1-\frac{(i-1)\cdot(\alpha -1)}{k+1}})$ (whp), which is $\tilde{O}(m^{1+\frac{\alpha-1}{k+1}})$ for all $s\in V(S_i)$ (similarly to the previous calculation).
       
       The cost of removing a vertex $v$ is $O(deg(v))$. Thus, for every $i\geq 1$, the total cost of removing all the vertices in $\hat{V_i}$ is at most $O(m)$,    
        so the total running time of the $i$th iteration is $\tilde{O}(m^{1+\frac{\alpha-1}{k+1}})$, whp.
    
        If we are in the scenario that a cycle is returned, then the $i$th iteration  stops at an earlier stage, and therefore the running time is also $\tilde{O}(m^{1+\frac{\alpha-1}{k+1}})$.
    
        Now, since there are at most $y$ iterations of the main for loop, the running time of $\BFSSampleM$ is, whp\footnote{This is the running time in the case that $S_i$ was a hitting sets as described for every $1\leq i \leq y$. This happens whp since 
        we assume that
        $k < \frac{n}{2} < n$ and 
        therefore 
        $y < n$. For every $1\leq i \leq y$ the probability that $S_i$ is not such a hitting set is at most $\frac{1}{n^c}$. Therefore, using the standard union bound argument, the probability that there exists $1\leq i \leq y$ such that $S_i$ is not a hitting set is at most $y \cdot \frac{1}{n^c} \leq 
        \frac{1}{n^{c-1}}$. For large enough $c$, we get that $S_i$ is a hitting set for every $1\leq i \leq y$, whp.}, $\tilde{O}(y\cdot m^{1+\frac{\alpha-1}{k+1}}) = \Ot(\lfloor \frac{k+1}{\alpha-1}\rfloor \cdot
    m^{1+\frac{\alpha-1}{k+1}})$. \qedhere
\end{enumerate}
\end{proof}

Recall that our goal is to obtain the sparsity property that
$|E_u^{\alpha-1}| < m^{\frac{\alpha-1}{k+1}}$, for every $u\in V$, so that we can run $\dHybrid(G,\alpha)$.   
However, after running $\BFSSampleM$ the required sparsity property is guaranteed to hold (whp) only if $(k+1) \bmod (\alpha-1)  = 0$. 
In the case that $(k+1) \bmod (\alpha-1)  \neq 0$  we need an additional  step which is implemented in $\HandleRemainderM$,  to guarantee  that the required sparsity property holds.

Next, we formally describe $\HandleRemainderM$. 
$\HandleRemainderM$ (see Algorithm~\ref{A-handleRM}) gets a graph $G$ and two integers $\alpha,k \geq 2$ such that $k+1=q(\alpha-1)+r$ where $q\geq 1$ and $r>0$.
We  set
$D$ to $m^{\frac{1}{k+1}}$ and $r$ to $(k+1) \bmod (\alpha-1)$. 
Then,  a while loop runs as long as $(\alpha-1) \bmod r \neq 0$. 
Let $r_i$ be
the value of $r$ when the $i$th iteration
begins, so that $r_1$ is $ (k+1) \bmod (\alpha-1)$. 
Let $\ell$ be
the total number of iterations
and
$r_{\ell +1}$
the value of $r$ after the $\ell$th iteration.
During the $i$th iteration, we set $r_{i+1}$ to 
$(\lceil \frac{\alpha-1}{r_i} \rceil \cdot r_i) - (\alpha-1)$, where $\lceil \frac{\alpha-1}{r_i} \rceil \cdot r_i$ is the smallest multiple of $r_i$ that is at least $\alpha-1$ (see Figure~\ref{fig:r-jumps-M}).
Then, we call $\SparseOrCycleM(G, D, r_{i+1}, \alpha)$.
If $\SparseOrCycleM$ returns a cycle $C$ then $\HandleRemainderM$ returns $C$.
 If $\SparseOrCycleM$ does not return a cycle then it might be that some vertices were removed from $G$, and we continue to the next iteration.
If the while loop ends without returning a cycle then we return $\codeNull$.
Let $G_0$ ($G_1$) be  $G$  before (after) running $\HandleRemainderM$.
Next, we prove two properties on the value of  $r$ during the run of  $\HandleRemainderM$.

\begin{algorithm2e}[t]
\caption{$\HandleRemainderM(G,k,\alpha)$}
\label{A-handleRM}

    $D\gets m^{\frac{1}{k+1}}$\;
    
    $r \gets (k+1) \bmod (\alpha-1)$\;

    \While{ $(\alpha-1) \bmod r \neq 0$ } 
    {

        $r\gets (\lceil \frac{\alpha-1}{r} \rceil \cdot r) - (\alpha-1)$\; 
        
        $C\gets \SparseOrCycleM(G, D, r, \alpha)$\;
        
        \lIf {$C\neq \codeNull$}{\Return $C$}
        
    } 
    
    \Return $\codeNull$\;
     
\end{algorithm2e}

\begin{claim} \label{C-remindersM} 
Let $k+1 = q(\alpha-1)+r_1$ and assume  $r_1>0$. \textbf{(i)} $0 < r_{\ell+1} < r_{\ell} < \cdots < r_1 <\alpha-1$. \textbf{(ii)} ($r_{i+1} + \alpha-1) \bmod  r_i = 0$, for every $1\leq i\leq \ell$.

\end{claim}

\begin{proof}
\begin{enumerate} [label=(\roman*)]
\item 
 First, since $r_1 = (k+1) \bmod (\alpha-1)$ and we assume that $r_1>0$, we have $0 < r_1 < \alpha-1$.
Now we show that $0 < r_{i+1}<r_i$, for every $1 \leq i \leq \ell$.  
This implies that  $0 < r_{\ell+1} < r_{\ell} < \cdots < r_1 <\alpha$, as required.

We first show by induction that $r_i>0$, for every $1 \leq i \leq \ell$. The base of the induction follows from the assumption that $r_1 > 0$. We assume that $r_{i}>0$ and prove that $r_{i+1}>0$. 
During the $i$th iteration of the while loop, we set $r_{i+1}$ to 
$(\lceil \frac{\alpha-1}{r_i} \rceil \cdot r_i) - (\alpha-1)$. 
Since this occurs during  the $i$th iteration it must be that $(\alpha-1) \bmod r_i \neq 0$, as otherwise the $i$th iteration would not have started.
Since $(\alpha-1) \bmod r_i \neq 0$ we have $\lceil\frac{\alpha-1}{r_i}\rceil>\frac{\alpha-1}{r_i}$ and therefore, 
$\lceil \frac{\alpha-1}{r_i} \rceil \cdot r_i > \alpha - 1$.
We get that 
$r_{i+1} = (\lceil \frac{\alpha-1}{r_i} \rceil \cdot r_i) - (\alpha-1) >0$, as required.

We now turn to prove that $r_{i+1}<r_i$.
As $\lceil \frac{\alpha-1}{r_i} \rceil \cdot r_i$ is the smallest multiple of $r_i$ that is at least $\alpha-1$, we know that 
 $(\lceil\frac{\alpha-1}{r_i}\rceil -1)\cdot r_i = \lceil\frac{\alpha-1}{r_i}\rceil\cdot r_i -r_i < \alpha-1$.
 Therefore, $\lceil\frac{\alpha-1}{r_i}\rceil\cdot r_i - (\alpha-1) < r_i$. Since $r_{i+1} = \lceil\frac{\alpha-1}{r_i}\rceil\cdot r_i -(\alpha-1)$, we get that $r_{i+1}<r_i$.

\item
Let $1\leq i\leq \ell$. Since $ r_{i+1} = \lceil\frac{\alpha-1}{r_i}\rceil\cdot r_i -(\alpha-1)$, it follows that $ r_{i+1} +(\alpha-1)$ is a multiple of $r_i$, so   ($r_{i+1} + \alpha-1) \bmod  r_i = 0$.  \qedhere
 \end{enumerate}
\end{proof}

\begin{figure}[t]
    \centering

    \begin{tikzpicture}[scale=1.3]
        \draw[thick] (0,0) -- (4,0);
        \draw[thick] (0,-0.07) -- (0,0.07);
        \draw[thick] (3.3,-0.07) -- (3.3,0.07);
        \draw[thick] (4,-0.07) -- (4,0.07);
        
        \node[font=\small] at (0,-0.05) [below] {$0$};
        \node[font=\small] at (3.25,-0.05) [below] {$\alpha-1$};
        \node[font=\small] at (4.25,-0.05) [below] {
        $\lceil \frac{\alpha-1}{r_i}\rceil \cdot r_i$
        };
        
        \foreach \x in {0, 1, 2, 3} {
            \draw[->] (\x, 0.25) arc[start angle=160, end angle=0, x radius=0.5, y radius=0.3
            ];
            \node[font=\small] at (\x + 0.5, 0.57) {$r_i$};
        }

        \draw[semithick, decorate, decoration={brace,  amplitude=4pt}] (3.3,0.1) -- (3.97,0.1);
        \node[font=\tiny] at (3.6, 0.3) {$r_{i+1}$};

    \end{tikzpicture}

    \caption{The relation between $r_i$, $r_{i+1}$ and $\alpha-1$}
    \label{fig:r-jumps-M}
\end{figure}

We now prove the main lemma regarding $\HandleRemainderM$.

\begin{lemma}\label{L-HandleRemainderM}
Let $k+1 = q(\alpha-1)+r_1$ and assume that $r_1>0$, $q\geq 1$,  and  that $|E_u^{r_{1}}| < D^{r_{1}}$, for every $u\in V$.
$\HandleRemainderM(G,k,\alpha)$ 
satisfies the following: 
\begin{enumerate} [label=(\roman*)]
\item If a cycle $C$ is returned then $\wt(C)\leq 2k$
\item If a cycle is not returned then
$|E_u^{\alpha-1}| < D^{\alpha-1}$, for every vertex $u\in G_2$
\item 
If $u \in G_1 \setminus G_2$ 
then $u$ is not on a $C_{\leq 2\alpha}$ in  $G_1$
\item 
$\HandleRemainderM(G,k,\alpha)$ runs in $O(r_1\cdot m^{1+\frac{\alpha-1}{k+1}})$ time.
\end{enumerate}
\end{lemma}

\begin{proof}
\begin{enumerate} [label=(\roman*)]
    \item
    First, we show that if a cycle $C$ is returned then $\wt(C) \leq 2k$.     
   $\HandleRemainderM$  returns  a cycle only if a call  to $\SparseOrCycleM(G, D, r_{i+1}, \alpha)$, where $1\leq i \leq \ell$, returns  a cycle. 
   By Lemma~\ref{L-SparseOrCycleM}, if $\SparseOrCycleM(G, D, r_{i+1}, \alpha)$ returns a cycle $C$ then $\wt(C)\leq 2(r_{i+1} -1 +\alpha)$. 
    By Claim~\ref{C-remindersM}(i), $r_{i+1} < r_1$.
    In addition, since $k+1 =q(\alpha-1) +r_1$ with $q\geq 1$, 
     we have $r_{i+1} -1 +\alpha < r_{1} -1 +\alpha = (\alpha-1) + r_{1} \leq k+1$.
    As $r_{i+1} -1 +\alpha$ and $k$ are integers, we have $r_{i+1} -1 +\alpha \leq k$.
    Therefore, $\wt(C)\leq 2(r_{i+1} -1 +\alpha) \leq 2k$.

    \item
    Next, we show that if a cycle is not returned then
    $|E_u^{\alpha-1}| < D^{\alpha-1}$, for every vertex $u\in G_2$.
    To do so, we show that if $\HandleRemainderM$ does not return a cycle then when the algorithm ends, for every vertex $u \in G_2$ we have $|E_u^{r_{\ell+1}}| < D^{r_{\ell+1}}$.

     If we do not enter the while loop and $\ell=0$ then by our assumption $|E_u^{r_{1}}| < D^{r_{1}}$, for every $u\in V$, as required.
    Now we assume that we enter the loop so $\ell>0$.
    Consider the $i$th iteration of the while loop. During the $i$th iteration $\SparseOrCycleM(G,D,r_{i+1},\alpha)$ is called. 
    If $\SparseOrCycleM$ does not return a cycle, then it follows from  Lemma~\ref{L-SparseOrCycleM} that if $u\in \hat{V}$ ($u$ was not removed) then  $|E_u^{r_{i+1}}| < D^{r_{i+1}}$. 
    Therefore, if no cycle was returned by $\HandleRemainderM$ we get for $i=\ell$ that after the $\ell$th iteration, we have $|E_u^{r_{\ell+1}}| < D^{r_{\ell+1}}$, for every $u \in G_2$. 
    
    The while loop ends when $(\alpha-1) \bmod r = 0$. Therefore, we know that after the while loop, $(\alpha-1) \bmod r_{\ell +1} = 0$. 
    Additionally,  By Claim~\ref{C-remindersM}(i) $ r_{\ell +1} < \alpha - 1$, so there is $z>0$ such that $\alpha - 1  = zr_{\ell +1}$. By Corollary~\ref{C-SumOfDenseM}, since $|E_u^{r_{\ell+1}}| < D^{r_{\ell+1}}$, we know that $|E_u^{\alpha-1}| = |E_u^{zr_{\ell +1}}| < D^{{zr_{\ell +1}}} = 
    D^{\alpha-1}$, as required.

    \item
    Next, we prove that if $u \in G_1 \setminus G_2$ then $u$ is not on a $C_{\leq 2\alpha}$ in  $G_1$.
    If $u \in G_1 \setminus G_2$ then by the definition of $G_1$ and $G_2$, $u$ was removed while executing  $\HandleRemainderM$.
    During the run of $\HandleRemainderM$, a vertex can be removed only by $\SparseOrCycleM(G, D, r_{i+1}, \alpha)$ for some $1\leq i \leq \ell$.
    Therefore, by Lemma~\ref{L-SparseOrCycleM}, $u$ is not part of a $ C_{\leq2\alpha}$ in $G$.
    Since during the run of $\HandleRemainderM$  only vertices that are not part of a $C_{\leq2\alpha}$ are removed, $u$ is not part of a $C_{\leq 2\alpha}$ also in $G_1$.

    \item
    Finally, we show that $\HandleRemainderM$ runs in $O(r_1\cdot m^{1+\frac{\alpha-1}{k+1}})$ time.
    To do so, we show that the running time of the $i$th iteration of the while loop is $O(m^{1+\frac{\alpha-1}{k+1}})$. When $\ell=0$ there are no iterations and the running time is $O(1)$. Now we assume that $\ell>0$.
   During the $i$th iteration, we call $\SparseOrCycleM(G, D, r_{i+1}, \alpha)$. 
   We proved in (ii) that for $i > 1$ if $\SparseOrCycleM$ did not return a cycle during the $(i-1)$th iteration and $u\in \hat{V}$ ($u$ was not removed) then after the $(i-1)$th iteration $|E_u^{r_{i}}| < D^{r_{i}}$. For $i=1$ by our assumption $|E_u^{r_{1}}|  < D^{r_{1}}$.
    Therefore, before the $i$th iteration starts, by Corollary~\ref{C-SumOfDenseM}, $|E_u^{zr_{i}}| < D^{zr_{i}}$ for every integer $z>0$.  By Claim~\ref{C-remindersM}(ii), $r_{i+1}-1+\alpha$ is divisible by $r_i$, so $|E_u^{r_{i+1}-1+\alpha}|  < 
    D^{r_{i+1}-1+\alpha}$ for every vertex $u$.
    It then follows from Corollary~\ref{C-SparseOrCycle-sparse-time} that $\SparseOrCycleM(G, D, r_{i+1}, \alpha)$ runs in $O(nD^{r_{i+1}}+mD^{\alpha-1})$ time.
    By Claim~\ref{C-remindersM}(i), $r_{i+1} < \alpha-1$ so $O(nD^{r_{i+1}}+mD^{\alpha-1})\leq 
    O(mD^{\alpha-1})=
    O(m^{1+\frac{\alpha-1}{k+1}})$, and the running time of the $i$th iteration is $O(m^{1+\frac{\alpha-1}{k+1}})$.   
    Now, it follows from Claim~\ref{C-remindersM}(i) that after at most $r_1 < \alpha-1$ iterations, the value of $r$ cannot decrease anymore (since it cannot become less than $1$, and $(\alpha-1) \bmod 1 = 0$) so the while loop ends. As we saw, the running time of each iteration is $O(m^{1+\frac{\alpha-1}{k+1}})$, hence the total running time of the while loop is $O(r_1 \cdot m^{1+\frac{\alpha-1}{k+1}})$. \qedhere
    \end{enumerate}
\end{proof}

Now we are ready to prove the correctness and running time of $\SrtCycSpr$.

\begin{lemma} \label{L-shortCycleSparse}
   Let $2\leq \alpha \leq k < \frac{n}{2}$
    such that 
    $k +1 = q (\alpha-1)+r$, 
    where $q\geq 1$ and $0\leq r < \alpha-1$ are integers.
    Algorithm $\SrtCycSpr(G,k,\alpha)$ runs whp in $\tilde{O}((q+r) \cdot m^{1+\frac{\alpha-1}{k+1}})$
    time and either returns a $C_{\leq 2k}$, or determines that $g > 2\alpha$.
\end{lemma}

\begin{proof}

First, $\SrtCycSpr(G,k,\alpha)$ returns a cycle $C$ only if $\BFSSampleM(G,k,\alpha)$, $\HandleRemainderM(G,k,\alpha)$, or $\dHybrid(G,\alpha)$ returns a cycle $C$. 
If $\BFSSampleM$ or $\dHybrid$ returns a cycle $C$ then by Lemma~\ref{L-BFSSample-M} or by Lemma~\ref{L-dHybrid-properties}, $\wt(C)\leq 2k$.
If $\HandleRemainderM$ returns a cycle $C$ then by Lemma \ref{L-HandleRemainderM}, since $k\geq \alpha$ and hence $q\geq 1$,  $\wt(C)\leq 2k$.

Second, we show that if no cycle was found then $g>2\alpha$.
If no cycle was found then it might be that some vertices were removed from the graph.
A vertex $u$ can be removed either by  $\BFSSampleM$ or by $\HandleRemainderM$.  It follows from Lemma~\ref{L-BFSSample-M} and Lemma~\ref{L-HandleRemainderM}, that $u$ is not part of a $C_{\leq 2\alpha}$ when $u$ is removed. Since only vertices that are not on a $C_{\leq 2\alpha}$ are removed, every $C_{\leq 2\alpha}$ that was in the input graph also belongs to the updated graph. 
After the (possible) removal of vertices, we call $\dHybrid(G,\alpha)$ with the updated graph. Since we are in the case that no cycle was found, $\dHybrid$ did not return a cycle. It follows from Lemma~\ref{L-dHybrid-properties} that $g>2\alpha$ in the updated graph, and therefore $g>2\alpha$ also in the input graph.

Now we turn to analyze the running time of $\SrtCycSpr$.
At the beginning, $\SrtCycSpr$ calls $\BFSSampleM$.
By Lemma~\ref{L-BFSSample-M}, $\BFSSampleM$ runs in $\Ot(\lfloor \frac{k+1}{\alpha-1}\rfloor \cdot
    m^{1+\frac{\alpha-1}{k+1}}) = \Ot(q \cdot m^{1+\frac{\alpha-1}{k+1}})$ time.
Let $G_1$ be the graph after the call to $\BFSSampleM$.
Recall that $y=\lceil \frac{k+1}{\alpha-1}\rceil-1$. By Lemma~\ref{L-BFSSample-M},
  for every $u\in G_1$ we have $|E_u^{(k+1)-y \cdot (\alpha-1)}| < m^{1-\frac{y \cdot (\alpha-1)}{k+1}}$, whp.
 If $r>0$ then $y=q$ and $|E_u^{r}| < m^{\frac{r}{k+1}}$.  
 If $r=0$ then $y=q-1$ and $|E_u^{\alpha-1}| < m^{\frac{\alpha-1}{k+1}}$.

Next, $\SrtCycSpr$ checks if $r = (k-1) \bmod (\alpha-1) > 0$. 
We divide the rest of the proof  to the case that $r = 0$ and to the case that $r > 0$. 
If $r = 0$ then $\SrtCycSpr$ calls $\dHybrid$. 
Since $r = 0$, we have $|E_u^{\alpha-1}| < m^{\frac{\alpha-1}{k+1}}$ after $\BFSSampleM$.
By Corollary~\ref{C-dHybrid-sparse-time}, the running time of $\dHybrid$ is $O(m\cdot m^{\frac{\alpha-1}{k+1}}) = O(m^{1+\frac{\alpha-1}{k+1}})$.

We now turn to the case that $r > 0$. 
In this case it might be that $|E_u^{\alpha-1}| \geq m^{\frac{\alpha-1}{k+1}}$ for some vertices $u\in G_1$. 
Therefore, we first call $\HandleRemainderM(G,k,\alpha)$, knowing that $r>0$, $q\geq 1$ and that whp, $|E_u^{r}| < m^{\frac{r}{k+1}}$ after $\BFSSampleM$.
By Lemma~\ref{L-HandleRemainderM}, the running time is $O(r\cdot m^{1+\frac{\alpha-1}{k+1}})$.
Let $G_2$ be the graph after $\HandleRemainderM$ ends.  By Lemma~\ref{L-HandleRemainderM}, for every $u\in G_2$ we have $|E_u^{\alpha-1}| < D^{\alpha-1}$, where $D = m^{\frac{1}{k+1}}$.
Now $\SrtCycSpr$ calls $\dHybrid$, and using  Corollary~\ref{C-dHybrid-sparse-time} again, we get that the running time of $\dHybrid$ is $O(m\cdot m^{\frac{k-1}{k+1}}) = O(m^{1+\frac{k-1}{k+1}})$.

It follows from the above discussion that $\SrtCycSpr$ either returns a $C_{\leq 2k}$ or determines that $g>2\alpha$, and
the running time is whp\footnote{
    Throughout the run of $\SrtCycSpr$, some of the bounds that we get on $|E_v^d|$ for vertices $v\in V$ and distances $d$, are whp, because the sets that we sample are hitting sets whp (see the proof of Lemma~\ref{L-BFSSample-M}). Therefore, also the running times of $\BFSSampleM$, $\HandleRemainderM$ and $\dHybrid$ are whp, since they rely on these bounds.
    }
    $\Ot(( q+r) \cdot
    m^{1+\frac{\alpha-1}{k+1}})$.
\end{proof}

Since $\SrtCycSpr$ is run by $\OneAlg$ when $m\leq O(n^{1+\frac{1}{k}})$ and when 
$\alpha\leq k$
and so $1+\frac{\alpha-1}{k+1}\leq 2$, we have $m^{1+\frac{\alpha-1}{k+1}}\leq O(n^{1+\frac{\alpha}{k}})$. Thus, the running time of $\SrtCycSpr$ is whp $\Ot(( q+r) \cdot
   \min\{ m^{1+\frac{\alpha-1}{k+1}},
    n^{1+\frac{\alpha}{k}}
   \})$,
    which is at most 
   $\Ot(( \frac{k+1}{\alpha-1}+\alpha) \cdot
   \min\{ m^{1+\frac{\alpha-1}{k+1}},
    n^{1+\frac{\alpha}{k}}
   \})$.

\subsubsection{\texorpdfstring{Algorithm $\SpCase$}{Algorithm SpecialCases}} \label{subsec:special}

 We now present algorithm $\SpCase$ that handles special cases of $k$ and $\alpha$.
$\SpCase(G,k,\alpha)$ gets as an input a graph $G$ and two integers $k\geq 2$ and $\alpha \geq 2$.

If $k\geq \frac{n}{2}$, the algorithm simply runs $\BallOrCycle(G,w,n)$ from an arbitrary vertex $w\in V$ in $O(n)\leq O(\min\{m^{1+\frac{\alpha-1}{k+1}},n^{1+\frac{\alpha}{k}}\})$ time, to check whether $G$ contains a cycle.  If a cycle is found then its length is at most $n\leq 2k$, and we return a $C_{\leq 2k}$. Otherwise, $g=\infty$ so for every integer $\alpha$ we return that $g>2\alpha$.

If $k \leq \alpha -1$, we check if the graph contains a $C_d$ for $d\gets 3,\dots, 2\alpha$, using an algorithm of Alon, Yuster and Zwick \cite{alon1997finding}. Their algorithm decides whether $G$ contains a $C_{2 \ell-1}$ or a $C_{2 \ell}$, and finds such a cycle if it does, in $O(m^{2-\frac{1}{\ell}})$ time.
Applying this algorithm with increasing cycle lengths until a length of $2\alpha$ (the values of $\ell$ are in the worst case $2,3,\dots,\alpha$), we can either find a shortest cycle or determine that $g> 2\alpha$.
The running time is $O(m^{2-\frac{1}{2}}+m^{2-\frac{1}{3}}+\dots+ m^{2-\frac{1}{\alpha}})= O(\alpha \cdot m^{2-\frac{1}{\alpha}}) = O(\alpha \cdot  m^{1+\frac{\alpha-1}{\alpha}})$ time\footnote{It is possible to modify the algorithm of Alon \etal{} \cite{alon1997finding} to search in $O(m^{2-\frac{1}{\ell}})$ time a shortest cycle of length \textit{at most} $2\ell$ instead of \textit{exactly} $2\ell-1$ or $2\ell$, and then run it only with $\ell = \alpha$, to avoid the $\alpha$ factor in the running time. 
}.
Since $k\leq \alpha-1$, we have $k+1 \leq \alpha$ and therefore $O(\alpha\cdot m^{1+\frac{\alpha-1}{\alpha}}) \leq O(\alpha\cdot m^{1+\frac{\alpha-1}{k+1}})$.
In addition, since $m\leq O(n^{1+\frac{1}{k}})$ and since $1+\frac{\alpha-1}{\alpha}\leq 2$, we have $m^{1+\frac{\alpha-1}{\alpha}} \leq O(n^{(1+\frac{1}{k})\cdot (1+\frac{\alpha-1}{\alpha})})
\leq
O(n^{(1+\frac{1}{k})\cdot (1+\frac{\alpha-1}{k+1})}) 
=
O(n^{1+\frac{\alpha}{k}})
$.
Therefore, the running time is $O(\alpha\cdot m^{1+\frac{\alpha-1}{\alpha}})\leq O(\alpha\cdot \min\{m^{1+\frac{\alpha-1}{k+1}},n^{1+\frac{\alpha}{k}}\})$.

By choosing which algorithm to run according to the  relation between $k$, $\alpha$ and $\frac{n}{2}$, we get that for every two integers $\alpha\geq2$ and  $k\geq 2$, algorithm $\OneAlg(G,k,\alpha)$ runs whp in $\Ot((\frac{k+1}{\alpha-1}+\alpha) \cdot \min\{m^{1+\frac{\alpha-1}{k+1}},n^{1+\frac{\alpha}{k}}\})$ time 
and either returns a $C_{\leq \max\{2k,g\}}$, or determines that $g > 2\alpha$.
This completes the proof of Theorem~\ref{Thm-One-Alg}.

\section{Approximation of the girth}\label{sec:girth-approx}

In this section we present two new  tradeoffs for girth approximation that follow from  Corollary~\ref{C-shortCycleBigK}. 
In these tradeoffs we use $\OneAlg$ with  $2\leq \alpha \leq k$, so by Corollary~\ref{C-shortCycleBigK}, $\OneAlg$ is an $\tilde{O}((\frac{k+1}{\alpha-1}+\alpha) \cdot \min \{m^{1+\frac{\alpha-1}{k+1}},n^{1+\frac{\alpha}{k}}\})$-time, $(2k,2\alpha)$-hybrid algorithm.

\subsection{Dense graphs}
Kadria \etal{} \cite{kadria2022algorithmic} presented an  $O((\alpha-c)\cdot n^{1+\frac{\alpha}{2\alpha - c}})$-time, $(4\alpha-2c, 2\alpha)$-hybrid algorithm, where $0 < c \leq \alpha$ are two integers (see footnote  \ref{footnoteTC} in Section~\ref{sec:overview}).  This algorithm, combined with a binary search, was used by \cite{kadria2022algorithmic} 
to compute for every $\varepsilon\in (0,1]$ a cycle $C$  such that $\wt(C) \leq 4\lceil \frac{g}{2}\rceil  - 2\lfloor \varepsilon \lceil \frac{g}{2}\rceil \rfloor \leq (2-\varepsilon)g+4$, in $\tO(n^{1+1/(2-\varepsilon)})$ time, if $g\leq \log^2 n$.
We use $\OneAlg$ in a similar way  and prove:

\begin{theorem}\label{L-approx-n}
    Let $\Ynew\geq 2$ be an integer, $\varepsilon\in [0,1]$ and  $g\leq \log^2 n$.
   It is possible to compute, whp, in    $\tO(\Ynew\cdot n^{1+1/(\Ynew-\varepsilon)})$ time, a cycle $C$ such that $\wt(C) \leq 2\Ynew\lceil \frac{g}{2}\rceil  - 2\lfloor \varepsilon \lceil \frac{g}{2}\rceil \rfloor \leq
(\Ynew-\varepsilon)g+\Ynew+2$.\footnote{We note that when $\ell> \log n$, we can run the $\Ot(n^{1+1/\ell'})$-time algorithm of \cite{kadria2022algorithmic}, which computes a $C_{\leq 2\ell' \lceil g/2\rceil}$, with $\ell' = \ell -1$, to get the required approximation. Since $\ell> \log n$, its running time is $\Ot(n)$. Thus, our running time becomes $\Ot(n^{1+1/{(\ell-\varepsilon)}})$, where the 
$\ell$ factor is absorbed into the $\Ot$ notation.}
\end{theorem}

\begin{proof}
For each $\tilde{g}$ in the range $[3,\log^2 n]$ in increasing order, we call $\OneAlg(G,k(\alpha_{\tilde{g}}), \alpha_{\tilde{g}})$, where $\alpha_{\tilde{g}} = \lceil \frac{\tilde{g}}{2}\rceil$, and $k(\alpha) = \Ynew\alpha - \lfloor \varepsilon\alpha \rfloor$.
When we find the smallest value $\tilde{g}$ for which $\OneAlg$ returns a cycle, we stop and return that cycle. 
Since $\Ynew\geq 2$ and $\varepsilon\leq 1$ we have $k(\alpha_{\tilde{g}})\geq \alpha_{\tilde{g}}$, and it follows from Corollary~\ref{C-shortCycleBigK} that $\OneAlg$ either returns a $C_{\leq 2k(\alpha_{\tilde{g}})}$ or determines that $g> 2\alpha_{\tilde{g}}\geq \tilde{g}$ in  $\tilde{O}((\frac{k(\alpha)+1}{\alpha-1}+\alpha) \cdot n^{1+\frac{\alpha}{k(\alpha)}})$ time, whp.

We first prove that the algorithm returns a cycle $C$ such that $\wt(C) \leq 2\Ynew\lceil \frac{g}{2}\rceil  - 2\lfloor \varepsilon \lceil \frac{g}{2}\rceil \rfloor \leq
(\Ynew-\varepsilon)g+\Ynew+2$.
Let $g'$ be the smallest value $\tilde{g}$ for which $\OneAlg$ returned a cycle. This implies that for $g'-1$ the algorithm did not return a cycle, and hence $g>g'-1$. Since $g$ and $g'$ are integers, we have $g\geq g'$. Also for $g'=3$ we have $g\geq g'$ since the girth is at least $3$.

The call to $\OneAlg(G,k(\alpha_{g'}),\alpha_{g'})$ returns a cycle $C$
such that  $\wt(C)\leq 2 k(\alpha_{g'}) =
2\cdot (\Ynew\alpha_{g'} - \lfloor \varepsilon \alpha_{g'} \rfloor) = 
2\cdot \Ynew \lceil \frac{g'}{2} \rceil - 2\cdot \lfloor \varepsilon \lceil \frac{g'}{2} \rceil \rfloor
\leq 
2 \Ynew \lceil \frac{g}{2} \rceil - 2 \lfloor \varepsilon \lceil \frac{g}{2} \rceil \rfloor
$. 
Thus, $\wt(C)\leq 2\Ynew\lceil \frac{g}{2}\rceil  -  2 \lfloor \varepsilon \lceil \frac{g}{2} \rceil \rfloor \leq
2\Ynew\lceil \frac{g}{2}\rceil - 2\lfloor \varepsilon \frac{g}{2} \rfloor
\leq 
2\Ynew\lceil \frac{g}{2}\rceil - 2( \varepsilon \frac{g}{2} -1) 
\leq
2\Ynew \frac{g+1}{2} - \varepsilon g +2
=
\Ynew g+\Ynew - \varepsilon g +2 = (\Ynew-\varepsilon)g +\Ynew+2
$.

For the running time, there are at most $O(\log^2 n
)$ calls to $\OneAlg$, and each call costs $\tilde{O}((\frac{k(\alpha)+1}{\alpha-1}+\alpha) \cdot n^{1+\frac{\alpha}{k(\alpha)}})$ whp (with the values of $k$ and $\alpha$ that correspond to that call).
 In each call,
 $ \frac{k(\alpha)+1}{\alpha-1} 
 = 
 \frac{\Ynew\alpha-\lfloor \varepsilon \alpha \rfloor+1}{\alpha-1}
 \leq 
  \frac{\Ynew\alpha+1}{\alpha-1}
 $ which is at most $2\Ynew+1$ since $\alpha \geq 2$.
In addition, $n^{1+\frac{\alpha}{k(\alpha)}} = n^{1+\frac{\alpha}{\Ynew\alpha-\lfloor \varepsilon \alpha \rfloor}} \leq 
n^{1+\frac{\alpha}{\Ynew\alpha- \varepsilon \alpha }} = 
n^{1+\frac{1}{\Ynew- \varepsilon}}
$. Thus, the running time of each call is 
$\tilde{O}((\Ynew+\alpha) \cdot n^{1+\frac{1}{\Ynew- \varepsilon}})$ whp, which is $\tilde{O}(\Ynew\cdot n^{1+\frac{1}{\Ynew- \varepsilon}})$ since $\alpha \leq \log^2 n$ in each call. Therefore, the total running time is, whp\footnote{The running time of each call to $\OneAlg$ is whp. 
Since the number of calls to $\OneAlg$ is at most $\log^2 n < n$, using a union bound argument as in the proof of Lemma \ref{L-BFSSample-M} we get that
also the total running of all the calls is whp. \label{footnoteAWHP}}, $\tilde{O}(
\log^2 n
\cdot \Ynew\cdot n^{1+\frac{1}{\Ynew- \varepsilon}}) = \tilde{O}(\Ynew\cdot n^{1+\frac{1}{\Ynew- \varepsilon}})$.
\end{proof}

 Our algorithm can be viewed as a natural generalization of a couple of algorithms from \cite{kadria2022algorithmic}. 
By setting $\Ynew=2$ in Theorem~\ref{L-approx-n} we get the $(2-\varepsilon)g+4$ approximation of \cite{kadria2022algorithmic}.
By setting $\varepsilon=0$ we get an  $\tO(\Ynew\cdot n^{1+1/\Ynew})$ time algorithm that computes a  $C_{\leq 2\Ynew\lceil \frac{g}{2}\rceil}$, which is similar to the $\tO(n^{1+1/\ell})$ time algorithm of \cite{kadria2022algorithmic} that computes a $C_{\leq 2\ell \lceil \frac{g}{2}\rceil}$. 

\subsection{Sparse graphs} 
We use a similar approach to obtain a tradeoff for girth approximation in sparse graphs. We prove the following theorem.

\begin{theorem}\label{L-approx-m}
    Let $\Ynew\geq 3$ 
    be an integer, $\varepsilon\in [0,1)$ and  $g\leq \log^2 n$.
    It is possible to compute, whp, in  $\tO(\Ynew\cdot m^{1+1/(\Ynew-\varepsilon)})$ time, a cycle $C$ such that $\wt(C) \leq 2\Ynew(\lceil \frac{g}{2}\rceil-1)  - 2\lfloor \varepsilon (\lceil \frac{g}{2}\rceil-1) \rfloor-2 \leq
(\Ynew-\varepsilon)g-\Ynew+2\varepsilon$.
\end{theorem}

\begin{proof}
For each $\tilde{g}$ in the range $[3,\log^2 n]$ in increasing order, we call $\OneAlg(G,k(\alpha_{\tilde{g}}), \alpha_{\tilde{g}})$, where $\alpha_{\tilde{g}} = \lceil \frac{\tilde{g}}{2}\rceil$, and $k(\alpha) = \Ynew(\alpha-1) - \lfloor \varepsilon(\alpha-1) \rfloor-1$.
When we find the smallest value $\tilde{g}$ for which $\OneAlg$ returns a cycle, we stop and return that cycle. 
Since $\Ynew\geq 3$ and $\varepsilon < 1$ we have $k(\alpha_{\tilde{g}})\geq \alpha_{\tilde{g}}$, and it follows from Corollary~\ref{C-shortCycleBigK} that $\OneAlg$ either returns a $C_{\leq 2k(\alpha_{\tilde{g}})}$ or determines that $g> 2\alpha_{\tilde{g}}\geq \tilde{g}$ in  $\tilde{O}((\frac{k(\alpha)+1}{\alpha-1}+\alpha) \cdot m^{1+\frac{\alpha-1}{k(\alpha)+1}})$ time, whp.

We first prove that the algorithm returns a cycle $C$ such that $\wt(C) \leq 2\Ynew(\lceil \frac{g}{2}\rceil-1)  - 2\lfloor \varepsilon (\lceil \frac{g}{2}\rceil-1) \rfloor-2 \leq
(\Ynew-\varepsilon)g-\Ynew+2\varepsilon$.
Let $g'$ be the smallest value $\tilde{g}$ for which $\OneAlg$ returned a cycle. As before, this implies that $g\geq g'$. 
The call to $\SrtCycSpr(G,k(\alpha_{g'}),\alpha_{g'})$ returns a cycle $C$
such that  
\begin{align*}
\wt(C)\leq 2 k(\alpha_{g'}) &=
2\cdot ( \Ynew(\alpha_{g'}-1) - \lfloor \varepsilon(\alpha_{g'}-1) \rfloor-1 )\\
& = 
2\cdot ( \Ynew(\lceil \frac{g'}{2} \rceil-1) - \lfloor \varepsilon(\lceil \frac{g'}{2} \rceil-1) \rfloor-1 )\\
& \leq 
2\cdot ( \Ynew(\lceil \frac{g}{2} \rceil-1) - \lfloor \varepsilon(\lceil \frac{g}{2} \rceil-1) \rfloor-1 ) \\
&=
2 \Ynew(\lceil \frac{g}{2} \rceil-1) - 2\lfloor \varepsilon(\lceil \frac{g}{2} \rceil-1) \rfloor-2. 
\end{align*}

Now, since $\lfloor \varepsilon(\lceil \frac{g}{2} \rceil-1) \rfloor
\geq 
\lfloor \varepsilon( \frac{g}{2} -1) \rfloor
\geq 
 \varepsilon( \frac{g}{2} -1) -1 
 =
  \varepsilon\frac{g}{2} -\varepsilon -1
$,
we get that $\wt(C) \leq 
2 \Ynew(\lceil \frac{g}{2} \rceil-1) - 2(\varepsilon\frac{g}{2} -\varepsilon -1)-2
=
2 \Ynew\lceil \frac{g}{2} \rceil-2\Ynew - \varepsilon g+2\varepsilon
\leq
2 \Ynew\frac{g+1}{2}-2\Ynew - \varepsilon g+2\varepsilon
=
\Ynew g+\Ynew-2\Ynew - \varepsilon g+2\varepsilon
=
(\Ynew-\varepsilon)g -\Ynew +2\varepsilon
$.

For the running time, there are at most $O(\log^2 n)$ calls to $\OneAlg$, and each call costs $\tilde{O}((\frac{k(\alpha)+1}{\alpha-1}+\alpha) \cdot m^{1+\frac{\alpha-1}{k(\alpha)+1}})$ whp (with the values of $k$ and $\alpha$ that correspond to that call).
We have $ \frac{k(\alpha)+1}{\alpha-1}
=
\frac{\Ynew(\alpha-1) - \lfloor \varepsilon(\alpha-1) \rfloor}{\alpha-1}
\leq
\frac{\Ynew(\alpha-1)}{\alpha-1}
=
\Ynew
$.
In addition, $m^{1+\frac{\alpha-1}{k+1}} = m^{1+\frac{\alpha-1}{\Ynew(\alpha-1) - \lfloor \varepsilon(\alpha-1) \rfloor}} \leq 
m^{1+\frac{\alpha-1}{\Ynew(\alpha-1) -  \varepsilon(\alpha-1) }}
=
m^{1+\frac{1}{\Ynew - \varepsilon}}
$. Thus, the running time of each call is 
$\tilde{O}((\Ynew+\alpha) \cdot m^{1+\frac{1}{\Ynew - \varepsilon}})$ whp, which is $\tilde{O}(\Ynew\cdot m^{1+\frac{1}{\Ynew - \varepsilon}})$ since $\alpha \leq \log^2 n$ in each call. Therefore, the total running time is, whp\footnote{See footnote~\ref{footnoteAWHP}.}, $\tilde{O}(\log^2 n\cdot \Ynew\cdot m^{1+\frac{1}{\Ynew - \varepsilon}}) = \tilde{O}(\Ynew\cdot m^{1+\frac{1}{\Ynew - \varepsilon}})$.
\end{proof}

By setting $\varepsilon=0$ we get an  $\tO(\Ynew\cdot m^{1+1/\Ynew})$-time algorithm that computes a $C_{\leq 2\Ynew(\lceil \frac{g}{2}\rceil-1)}$, as opposed to the $\tO(\Ynew\cdot n^{1+1/\Ynew})$ time algorithm that computes a $C_{\leq2\Ynew\lceil \frac{g}{2}\rceil}$.

\bibliographystyle{plain}
\bibliography{bibliography.bib}

\end{document}

%% file: macros.tex
\newcommand{\etal}{et al.}
\newcommand{\tO}{\widetilde{O}}
\newcommand{\Ot}{\tilde{O}}
\newcommand{\polylog}{{\rm polylog}}

\newcommand{\gt}{\tilde{g}}

\newcommand{\Ynew}{\ell}

\newcommand{\codestyle}[1]{\texttt{#1}}
\newcommand{\Cycle}{\mbox{\codestyle{Cycle}}}
\newcommand{\BallOrCycle}{\codestyle{BallOrCycle}}
\newcommand{\DegenerateOrCycle}{\codestyle{DegenerateOrCycle}}
\newcommand{\IsDense}{\codestyle{IsDense}}

\newcommand{\codeNull}{\codestyle{null}}
\newcommand{\codeYes}{\codestyle{Yes}}
\newcommand{\codeNo}{\codestyle{No}}
\newcommand{\codeAnd}{~ \mathrm{and} ~}

\newcommand{\wt}{wt}

\newcommand{\OneAlg}{\codestyle{ShortCycle}}
\newcommand{\SrtCycSpr}{\codestyle{ShortCycleSparse}}
\newcommand{\SrtCycDns}{\codestyle{ShortCycleDense}}
\newcommand{\SpCase}{\codestyle{SpecialCases}}

\newcommand{\BFSSampleM}{\codestyle{BfsSample}}
\newcommand{\DoubleBfsCycle}{\codestyle{NbrBallOrCycle}}

\newcommand{\SparseOrCycleM}{\codestyle{SparseOrCycle}}
\newcommand{\HandleRemainderM}{\codestyle{HandleRemainder}}

\newcommand{\twoSparseOrCycle}{\codestyle{$2$-SparseOrCycle}}
\newcommand{\AdditiveGirthApprox}{\codestyle{AdtvGirthApprox}}
\newcommand{\dHybrid}{\codestyle{AllVtxBallOrCycle}}
\newcommand{\kHybrid}{\codestyle{$2k$-Hybrid}}

\newcommand{\A}{\mathcal{A}}

%% file: tikzdisjoint.tex
\begin{figure}[t]
    \centering

\begin{tikzpicture}[x=0.75pt,y=0.75pt,yscale=-1,xscale=1]

\draw [color={rgb, 255:red, 0; green, 0; blue, 0 }  ,draw opacity=0.4 ]   (290.21,1148.05) -- (393.64,1174.06) ;
\draw [color={rgb, 255:red, 0; green, 0; blue, 0 }  ,draw opacity=0.4 ]   (290.21,1148.05) -- (188.64,1174.06) ;
\draw  [color={rgb, 255:red, 0; green, 0; blue, 0 }  ,draw opacity=0.4 ] (188.48,1176.06) -- (218.48,1238.99) -- (158.48,1238.99) -- cycle ;
\draw  [color={rgb, 255:red, 0; green, 0; blue, 0 }  ,draw opacity=0.4 ][fill={rgb, 255:red, 255; green, 255; blue, 255 }  ,fill opacity=1 ] (177.43,1174.06) .. controls (177.43,1167.86) and (182.45,1162.84) .. (188.64,1162.84) .. controls (194.84,1162.84) and (199.86,1167.86) .. (199.86,1174.06) .. controls (199.86,1180.25) and (194.84,1185.27) .. (188.64,1185.27) .. controls (182.45,1185.27) and (177.43,1180.25) .. (177.43,1174.06) -- cycle ;
\draw  [color={rgb, 255:red, 0; green, 0; blue, 0 }  ,draw opacity=0.4 ] (164.37,1231.58) .. controls (164.37,1228.61) and (166.78,1226.21) .. (169.74,1226.21) .. controls (172.7,1226.21) and (175.1,1228.61) .. (175.1,1231.58) .. controls (175.1,1234.54) and (172.7,1236.94) .. (169.74,1236.94) .. controls (166.78,1236.94) and (164.37,1234.54) .. (164.37,1231.58) -- cycle ;
\draw  [color={rgb, 255:red, 0; green, 0; blue, 0 }  ,draw opacity=0.4 ] (176.7,1231.58) .. controls (176.7,1228.61) and (179.1,1226.21) .. (182.07,1226.21) .. controls (185.03,1226.21) and (187.43,1228.61) .. (187.43,1231.58) .. controls (187.43,1234.54) and (185.03,1236.94) .. (182.07,1236.94) .. controls (179.1,1236.94) and (176.7,1234.54) .. (176.7,1231.58) -- cycle ;
\draw  [color={rgb, 255:red, 0; green, 0; blue, 0 }  ,draw opacity=0.4 ] (189.03,1231.58) .. controls (189.03,1228.61) and (191.43,1226.21) .. (194.4,1226.21) .. controls (197.36,1226.21) and (199.76,1228.61) .. (199.76,1231.58) .. controls (199.76,1234.54) and (197.36,1236.94) .. (194.4,1236.94) .. controls (191.43,1236.94) and (189.03,1234.54) .. (189.03,1231.58) -- cycle ;
\draw  [color={rgb, 255:red, 0; green, 0; blue, 0 }  ,draw opacity=0.4 ] (201.36,1231.58) .. controls (201.36,1228.61) and (203.76,1226.21) .. (206.73,1226.21) .. controls (209.69,1226.21) and (212.09,1228.61) .. (212.09,1231.58) .. controls (212.09,1234.54) and (209.69,1236.94) .. (206.73,1236.94) .. controls (203.76,1236.94) and (201.36,1234.54) .. (201.36,1231.58) -- cycle ;
\draw    (290.21,1148.05) -- (325.79,1174.05) ;
\draw    (290.21,1148.05) -- (256.64,1174.15) ;
\draw   (325.48,1176.15) -- (355.48,1239.08) -- (295.48,1239.08) -- cycle ;
\draw   (256.48,1176.15) -- (286.48,1239.08) -- (226.48,1239.08) -- cycle ;
\draw  [fill={rgb, 255:red, 239; green, 187; blue, 106 }  ,fill opacity=1 ] (279,1148.05) .. controls (279,1141.86) and (284.02,1136.84) .. (290.21,1136.84) .. controls (296.41,1136.84) and (301.43,1141.86) .. (301.43,1148.05) .. controls (301.43,1154.25) and (296.41,1159.27) .. (290.21,1159.27) .. controls (284.02,1159.27) and (279,1154.25) .. (279,1148.05) -- cycle ;
\draw  [color={rgb, 255:red, 0; green, 0; blue, 0 }  ,draw opacity=1 ][fill={rgb, 255:red, 180; green, 207; blue, 242 }  ,fill opacity=1 ] (245.43,1174.15) .. controls (245.43,1167.95) and (250.45,1162.93) .. (256.64,1162.93) .. controls (262.84,1162.93) and (267.86,1167.95) .. (267.86,1174.15) .. controls (267.86,1180.34) and (262.84,1185.36) .. (256.64,1185.36) .. controls (250.45,1185.36) and (245.43,1180.34) .. (245.43,1174.15) -- cycle ;
\draw  [fill={rgb, 255:red, 180; green, 207; blue, 242 }  ,fill opacity=1 ] (314.57,1174.05) .. controls (314.57,1167.86) and (319.59,1162.84) .. (325.79,1162.84) .. controls (331.98,1162.84) and (337,1167.86) .. (337,1174.05) .. controls (337,1180.25) and (331.98,1185.27) .. (325.79,1185.27) .. controls (319.59,1185.27) and (314.57,1180.25) .. (314.57,1174.05) -- cycle ;
\draw  [color={rgb, 255:red, 0; green, 0; blue, 0 }  ,draw opacity=1 ] (232.37,1231.67) .. controls (232.37,1228.7) and (234.78,1226.3) .. (237.74,1226.3) .. controls (240.7,1226.3) and (243.1,1228.7) .. (243.1,1231.67) .. controls (243.1,1234.63) and (240.7,1237.03) .. (237.74,1237.03) .. controls (234.78,1237.03) and (232.37,1234.63) .. (232.37,1231.67) -- cycle ;
\draw  [color={rgb, 255:red, 0; green, 0; blue, 0 }  ,draw opacity=1 ] (244.7,1231.67) .. controls (244.7,1228.7) and (247.1,1226.3) .. (250.07,1226.3) .. controls (253.03,1226.3) and (255.43,1228.7) .. (255.43,1231.67) .. controls (255.43,1234.63) and (253.03,1237.03) .. (250.07,1237.03) .. controls (247.1,1237.03) and (244.7,1234.63) .. (244.7,1231.67) -- cycle ;
\draw  [color={rgb, 255:red, 0; green, 0; blue, 0 }  ,draw opacity=1 ] (257.03,1231.67) .. controls (257.03,1228.7) and (259.43,1226.3) .. (262.4,1226.3) .. controls (265.36,1226.3) and (267.76,1228.7) .. (267.76,1231.67) .. controls (267.76,1234.63) and (265.36,1237.03) .. (262.4,1237.03) .. controls (259.43,1237.03) and (257.03,1234.63) .. (257.03,1231.67) -- cycle ;
\draw  [color={rgb, 255:red, 0; green, 0; blue, 0 }  ,draw opacity=1 ] (269.36,1231.67) .. controls (269.36,1228.7) and (271.76,1226.3) .. (274.73,1226.3) .. controls (277.69,1226.3) and (280.09,1228.7) .. (280.09,1231.67) .. controls (280.09,1234.63) and (277.69,1237.03) .. (274.73,1237.03) .. controls (271.76,1237.03) and (269.36,1234.63) .. (269.36,1231.67) -- cycle ;
\draw  [color={rgb, 255:red, 0; green, 0; blue, 0 }  ,draw opacity=1 ] (301.37,1231.67) .. controls (301.37,1228.7) and (303.78,1226.3) .. (306.74,1226.3) .. controls (309.7,1226.3) and (312.1,1228.7) .. (312.1,1231.67) .. controls (312.1,1234.63) and (309.7,1237.03) .. (306.74,1237.03) .. controls (303.78,1237.03) and (301.37,1234.63) .. (301.37,1231.67) -- cycle ;
\draw  [color={rgb, 255:red, 0; green, 0; blue, 0 }  ,draw opacity=1 ] (313.7,1231.67) .. controls (313.7,1228.7) and (316.1,1226.3) .. (319.07,1226.3) .. controls (322.03,1226.3) and (324.43,1228.7) .. (324.43,1231.67) .. controls (324.43,1234.63) and (322.03,1237.03) .. (319.07,1237.03) .. controls (316.1,1237.03) and (313.7,1234.63) .. (313.7,1231.67) -- cycle ;
\draw  [color={rgb, 255:red, 0; green, 0; blue, 0 }  ,draw opacity=1 ] (326.03,1231.67) .. controls (326.03,1228.7) and (328.43,1226.3) .. (331.4,1226.3) .. controls (334.36,1226.3) and (336.76,1228.7) .. (336.76,1231.67) .. controls (336.76,1234.63) and (334.36,1237.03) .. (331.4,1237.03) .. controls (328.43,1237.03) and (326.03,1234.63) .. (326.03,1231.67) -- cycle ;
\draw  [color={rgb, 255:red, 0; green, 0; blue, 0 }  ,draw opacity=1 ] (338.36,1231.67) .. controls (338.36,1228.7) and (340.76,1226.3) .. (343.73,1226.3) .. controls (346.69,1226.3) and (349.09,1228.7) .. (349.09,1231.67) .. controls (349.09,1234.63) and (346.69,1237.03) .. (343.73,1237.03) .. controls (340.76,1237.03) and (338.36,1234.63) .. (338.36,1231.67) -- cycle ;
\draw    (446.43,1237.71) -- (446.43,1168.44) ;
\draw [shift={(446.43,1165.44)}, rotate = 90] [fill={rgb, 255:red, 0; green, 0; blue, 0 }  ][line width=0.08]  [draw opacity=0] (8.93,-4.29) -- (0,0) -- (8.93,4.29) -- cycle    ;
\draw [shift={(446.43,1240.71)}, rotate = 270] [fill={rgb, 255:red, 0; green, 0; blue, 0 }  ][line width=0.08]  [draw opacity=0] (8.93,-4.29) -- (0,0) -- (8.93,4.29) -- cycle    ;
\draw    (433.43,1237.51) -- (433.43,1140.71) ;
\draw [shift={(433.43,1137.71)}, rotate = 90] [fill={rgb, 255:red, 0; green, 0; blue, 0 }  ][line width=0.08]  [draw opacity=0] (8.93,-4.29) -- (0,0) -- (8.93,4.29) -- cycle    ;
\draw [shift={(433.43,1240.51)}, rotate = 270] [fill={rgb, 255:red, 0; green, 0; blue, 0 }  ][line width=0.08]  [draw opacity=0] (8.93,-4.29) -- (0,0) -- (8.93,4.29) -- cycle    ;
\draw  [color={rgb, 255:red, 0; green, 0; blue, 0 }  ,draw opacity=0.4 ] (393.48,1176.06) -- (423.48,1238.99) -- (363.48,1238.99) -- cycle ;
\draw  [color={rgb, 255:red, 0; green, 0; blue, 0 }  ,draw opacity=0.4 ][fill={rgb, 255:red, 255; green, 255; blue, 255 }  ,fill opacity=1 ] (382.43,1174.06) .. controls (382.43,1167.86) and (387.45,1162.84) .. (393.64,1162.84) .. controls (399.84,1162.84) and (404.86,1167.86) .. (404.86,1174.06) .. controls (404.86,1180.25) and (399.84,1185.27) .. (393.64,1185.27) .. controls (387.45,1185.27) and (382.43,1180.25) .. (382.43,1174.06) -- cycle ;
\draw  [color={rgb, 255:red, 0; green, 0; blue, 0 }  ,draw opacity=0.4 ] (369.37,1231.58) .. controls (369.37,1228.61) and (371.78,1226.21) .. (374.74,1226.21) .. controls (377.7,1226.21) and (380.1,1228.61) .. (380.1,1231.58) .. controls (380.1,1234.54) and (377.7,1236.94) .. (374.74,1236.94) .. controls (371.78,1236.94) and (369.37,1234.54) .. (369.37,1231.58) -- cycle ;
\draw  [color={rgb, 255:red, 0; green, 0; blue, 0 }  ,draw opacity=0.4 ] (381.7,1231.58) .. controls (381.7,1228.61) and (384.1,1226.21) .. (387.07,1226.21) .. controls (390.03,1226.21) and (392.43,1228.61) .. (392.43,1231.58) .. controls (392.43,1234.54) and (390.03,1236.94) .. (387.07,1236.94) .. controls (384.1,1236.94) and (381.7,1234.54) .. (381.7,1231.58) -- cycle ;
\draw  [color={rgb, 255:red, 0; green, 0; blue, 0 }  ,draw opacity=0.4 ] (394.03,1231.58) .. controls (394.03,1228.61) and (396.43,1226.21) .. (399.4,1226.21) .. controls (402.36,1226.21) and (404.76,1228.61) .. (404.76,1231.58) .. controls (404.76,1234.54) and (402.36,1236.94) .. (399.4,1236.94) .. controls (396.43,1236.94) and (394.03,1234.54) .. (394.03,1231.58) -- cycle ;
\draw  [color={rgb, 255:red, 0; green, 0; blue, 0 }  ,draw opacity=0.4 ] (406.36,1231.58) .. controls (406.36,1228.61) and (408.76,1226.21) .. (411.73,1226.21) .. controls (414.69,1226.21) and (417.09,1228.61) .. (417.09,1231.58) .. controls (417.09,1234.54) and (414.69,1236.94) .. (411.73,1236.94) .. controls (408.76,1236.94) and (406.36,1234.54) .. (406.36,1231.58) -- cycle ;

\draw (286,1145.51) node [anchor=north west][inner sep=0.75pt]  [font=\footnotesize]  {$s$};
\draw (252,1170.66) node [anchor=north west][inner sep=0.75pt]  [font=\footnotesize]  {$x$};
\draw (321,1169.66) node [anchor=north west][inner sep=0.75pt]  [font=\footnotesize]  {$y$};
\draw (451,1195.46) node [anchor=north west][inner sep=0.75pt]  [font=\footnotesize]  {$k-1$};
\draw (419,1180.46) node [anchor=north west][inner sep=0.75pt]  [font=\footnotesize]  {$k$};

\end{tikzpicture}

    \caption{Disjoint sets of vertices at distance $k-1$ from $x$ and $y$ }
    \label{fig:nbrballorcycleIllustration}
\end{figure}

%% file: tikzpictureBFSSAMPLE.tex
\begin{figure}[t]
    \centering

\begin{figure}[H]
    \centering

 
\tikzset{
pattern size/.store in=\mcSize, 
pattern size = 5pt,
pattern thickness/.store in=\mcThickness, 
pattern thickness = 0.3pt,
pattern radius/.store in=\mcRadius, 
pattern radius = 1pt}
\makeatletter
\pgfutil@ifundefined{pgf@pattern@name@_wx0veqwwd}{
\pgfdeclarepatternformonly[\mcThickness,\mcSize]{_wx0veqwwd}
{\pgfqpoint{0pt}{-\mcThickness}}
{\pgfpoint{\mcSize}{\mcSize}}
{\pgfpoint{\mcSize}{\mcSize}}
{
\pgfsetcolor{\tikz@pattern@color}
\pgfsetlinewidth{\mcThickness}
\pgfpathmoveto{\pgfqpoint{0pt}{\mcSize}}
\pgfpathlineto{\pgfpoint{\mcSize+\mcThickness}{-\mcThickness}}
\pgfusepath{stroke}
}}
\makeatother

\begin{tikzpicture}[x=0.75pt,y=0.75pt,yscale=-1,xscale=1]

\draw [line width=2.25]    (382.64,1995.88) -- (396.79,1971.79) ;
\draw    (410.64,1946.79) -- (396.79,1971.79) ;
\draw    (385.21,1946.79) -- (396.79,1971.79) ;
\draw    (396.79,1971.79) -- (423.21,1971.79) ;
\draw  [fill={rgb, 255:red, 224; green, 236; blue, 247 }  ,fill opacity=1 ] (412,1971.79) .. controls (412,1965.59) and (417.02,1960.57) .. (423.21,1960.57) .. controls (429.41,1960.57) and (434.43,1965.59) .. (434.43,1971.79) .. controls (434.43,1977.98) and (429.41,1983) .. (423.21,1983) .. controls (417.02,1983) and (412,1977.98) .. (412,1971.79) -- cycle ;
\draw  [fill={rgb, 255:red, 224; green, 236; blue, 247 }  ,fill opacity=1 ] (374,1946.79) .. controls (374,1940.59) and (379.02,1935.57) .. (385.21,1935.57) .. controls (391.41,1935.57) and (396.43,1940.59) .. (396.43,1946.79) .. controls (396.43,1952.98) and (391.41,1958) .. (385.21,1958) .. controls (379.02,1958) and (374,1952.98) .. (374,1946.79) -- cycle ;
\draw   (132,1971.79) .. controls (132,1965.59) and (137.02,1960.57) .. (143.21,1960.57) .. controls (149.41,1960.57) and (154.43,1965.59) .. (154.43,1971.79) .. controls (154.43,1977.98) and (149.41,1983) .. (143.21,1983) .. controls (137.02,1983) and (132,1977.98) .. (132,1971.79) -- cycle ;
\draw   (119.43,1946.79) .. controls (119.43,1940.59) and (124.45,1935.57) .. (130.64,1935.57) .. controls (136.84,1935.57) and (141.86,1940.59) .. (141.86,1946.79) .. controls (141.86,1952.98) and (136.84,1958) .. (130.64,1958) .. controls (124.45,1958) and (119.43,1952.98) .. (119.43,1946.79) -- cycle ;
\draw [line width=2.25]    (102.64,1995.88) -- (116.79,1971.79) ;
\draw  [fill={rgb, 255:red, 239; green, 187; blue, 106 }  ,fill opacity=1 ] (119,1897.79) .. controls (119,1891.59) and (124.02,1886.57) .. (130.21,1886.57) .. controls (136.41,1886.57) and (141.43,1891.59) .. (141.43,1897.79) .. controls (141.43,1903.98) and (136.41,1909) .. (130.21,1909) .. controls (124.02,1909) and (119,1903.98) .. (119,1897.79) -- cycle ;
\draw   (132.43,1921.79) .. controls (132.43,1915.59) and (137.45,1910.57) .. (143.64,1910.57) .. controls (149.84,1910.57) and (154.86,1915.59) .. (154.86,1921.79) .. controls (154.86,1927.98) and (149.84,1933) .. (143.64,1933) .. controls (137.45,1933) and (132.43,1927.98) .. (132.43,1921.79) -- cycle ;
\draw   (106,1921.79) .. controls (106,1915.59) and (111.02,1910.57) .. (117.21,1910.57) .. controls (123.41,1910.57) and (128.43,1915.59) .. (128.43,1921.79) .. controls (128.43,1927.98) and (123.41,1933) .. (117.21,1933) .. controls (111.02,1933) and (106,1927.98) .. (106,1921.79) -- cycle ;
\draw   (145.86,1946.79) .. controls (145.86,1940.59) and (150.88,1935.57) .. (157.07,1935.57) .. controls (163.26,1935.57) and (168.29,1940.59) .. (168.29,1946.79) .. controls (168.29,1952.98) and (163.26,1958) .. (157.07,1958) .. controls (150.88,1958) and (145.86,1952.98) .. (145.86,1946.79) -- cycle ;
\draw   (158.43,1971.79) .. controls (158.43,1965.59) and (163.45,1960.57) .. (169.64,1960.57) .. controls (175.84,1960.57) and (180.86,1965.59) .. (180.86,1971.79) .. controls (180.86,1977.98) and (175.84,1983) .. (169.64,1983) .. controls (163.45,1983) and (158.43,1977.98) .. (158.43,1971.79) -- cycle ;
\draw   (79.14,1971.79) .. controls (79.14,1965.59) and (84.16,1960.57) .. (90.36,1960.57) .. controls (96.55,1960.57) and (101.57,1965.59) .. (101.57,1971.79) .. controls (101.57,1977.98) and (96.55,1983) .. (90.36,1983) .. controls (84.16,1983) and (79.14,1977.98) .. (79.14,1971.79) -- cycle ;
\draw  [color={rgb, 255:red, 0; green, 0; blue, 0 }  ,draw opacity=0.53 ][fill={rgb, 255:red, 239; green, 187; blue, 106 }  ,fill opacity=0.6 ] (399,1897.79) .. controls (399,1891.59) and (404.02,1886.57) .. (410.21,1886.57) .. controls (416.41,1886.57) and (421.43,1891.59) .. (421.43,1897.79) .. controls (421.43,1903.98) and (416.41,1909) .. (410.21,1909) .. controls (404.02,1909) and (399,1903.98) .. (399,1897.79) -- cycle ;
\draw   (410.3,1946.79) -- (375.55,1856.06) -- (445.55,1856.25) -- cycle ;
\draw  [dash pattern={on 1.5pt off 1.5pt on 1.5pt off 1.5pt}]  (511.14,1884.08) -- (359.43,1884.01) ;
\draw  [color={rgb, 255:red, 155; green, 155; blue, 155 }  ,draw opacity=0.5 ] (412.43,1921.79) .. controls (412.43,1915.59) and (417.45,1910.57) .. (423.64,1910.57) .. controls (429.84,1910.57) and (434.86,1915.59) .. (434.86,1921.79) .. controls (434.86,1927.98) and (429.84,1933) .. (423.64,1933) .. controls (417.45,1933) and (412.43,1927.98) .. (412.43,1921.79) -- cycle ;
\draw  [color={rgb, 255:red, 155; green, 155; blue, 155 }  ,draw opacity=0.5 ] (386,1921.79) .. controls (386,1915.59) and (391.02,1910.57) .. (397.21,1910.57) .. controls (403.41,1910.57) and (408.43,1915.59) .. (408.43,1921.79) .. controls (408.43,1927.98) and (403.41,1933) .. (397.21,1933) .. controls (391.02,1933) and (386,1927.98) .. (386,1921.79) -- cycle ;
\draw  [color={rgb, 255:red, 155; green, 155; blue, 155 }  ,draw opacity=0.5 ] (425.86,1946.79) .. controls (425.86,1940.59) and (430.88,1935.57) .. (437.07,1935.57) .. controls (443.26,1935.57) and (448.29,1940.59) .. (448.29,1946.79) .. controls (448.29,1952.98) and (443.26,1958) .. (437.07,1958) .. controls (430.88,1958) and (425.86,1952.98) .. (425.86,1946.79) -- cycle ;
\draw  [color={rgb, 255:red, 155; green, 155; blue, 155 }  ,draw opacity=0.5 ] (438.43,1971.79) .. controls (438.43,1965.59) and (443.45,1960.57) .. (449.64,1960.57) .. controls (455.84,1960.57) and (460.86,1965.59) .. (460.86,1971.79) .. controls (460.86,1977.98) and (455.84,1983) .. (449.64,1983) .. controls (443.45,1983) and (438.43,1977.98) .. (438.43,1971.79) -- cycle ;
\draw  [color={rgb, 255:red, 155; green, 155; blue, 155 }  ,draw opacity=0.5 ] (359.14,1971.79) .. controls (359.14,1965.59) and (364.16,1960.57) .. (370.36,1960.57) .. controls (376.55,1960.57) and (381.57,1965.59) .. (381.57,1971.79) .. controls (381.57,1977.98) and (376.55,1983) .. (370.36,1983) .. controls (364.16,1983) and (359.14,1977.98) .. (359.14,1971.79) -- cycle ;
\draw  [fill={rgb, 255:red, 180; green, 207; blue, 242 }  ,fill opacity=1 ] (105.57,1971.79) .. controls (105.57,1965.59) and (110.59,1960.57) .. (116.79,1960.57) .. controls (122.98,1960.57) and (128,1965.59) .. (128,1971.79) .. controls (128,1977.98) and (122.98,1983) .. (116.79,1983) .. controls (110.59,1983) and (105.57,1977.98) .. (105.57,1971.79) -- cycle ;
\draw   (94,1946.79) .. controls (94,1940.59) and (99.02,1935.57) .. (105.21,1935.57) .. controls (111.41,1935.57) and (116.43,1940.59) .. (116.43,1946.79) .. controls (116.43,1952.98) and (111.41,1958) .. (105.21,1958) .. controls (99.02,1958) and (94,1952.98) .. (94,1946.79) -- cycle ;
\draw  [color={rgb, 255:red, 0; green, 0; blue, 0 }  ,draw opacity=1 ][fill={rgb, 255:red, 180; green, 207; blue, 242 }  ,fill opacity=1 ] (91.43,1995.88) .. controls (91.43,1989.69) and (96.45,1984.67) .. (102.64,1984.67) .. controls (108.84,1984.67) and (113.86,1989.69) .. (113.86,1995.88) .. controls (113.86,2002.07) and (108.84,2007.1) .. (102.64,2007.1) .. controls (96.45,2007.1) and (91.43,2002.07) .. (91.43,1995.88) -- cycle ;
\draw  [color={rgb, 255:red, 0; green, 0; blue, 0 }  ,draw opacity=1 ] (170.71,1995.88) .. controls (170.71,1989.69) and (175.74,1984.67) .. (181.93,1984.67) .. controls (188.12,1984.67) and (193.14,1989.69) .. (193.14,1995.88) .. controls (193.14,2002.07) and (188.12,2007.1) .. (181.93,2007.1) .. controls (175.74,2007.1) and (170.71,2002.07) .. (170.71,1995.88) -- cycle ;
\draw  [color={rgb, 255:red, 0; green, 0; blue, 0 }  ,draw opacity=1 ] (64,1995.88) .. controls (64,1989.69) and (69.02,1984.67) .. (75.21,1984.67) .. controls (81.41,1984.67) and (86.43,1989.69) .. (86.43,1995.88) .. controls (86.43,2002.07) and (81.41,2007.1) .. (75.21,2007.1) .. controls (69.02,2007.1) and (64,2002.07) .. (64,1995.88) -- cycle ;
\draw  [color={rgb, 255:red, 0; green, 0; blue, 0 }  ,draw opacity=1 ] (144.29,1995.88) .. controls (144.29,1989.69) and (149.31,1984.67) .. (155.5,1984.67) .. controls (161.69,1984.67) and (166.71,1989.69) .. (166.71,1995.88) .. controls (166.71,2002.07) and (161.69,2007.1) .. (155.5,2007.1) .. controls (149.31,2007.1) and (144.29,2002.07) .. (144.29,1995.88) -- cycle ;
\draw  [color={rgb, 255:red, 0; green, 0; blue, 0 }  ,draw opacity=1 ] (117.86,1995.88) .. controls (117.86,1989.69) and (122.88,1984.67) .. (129.07,1984.67) .. controls (135.26,1984.67) and (140.29,1989.69) .. (140.29,1995.88) .. controls (140.29,2002.07) and (135.26,2007.1) .. (129.07,2007.1) .. controls (122.88,2007.1) and (117.86,2002.07) .. (117.86,1995.88) -- cycle ;
\draw  [color={rgb, 255:red, 155; green, 155; blue, 155 }  ,draw opacity=0.5 ] (451.71,1995.88) .. controls (451.71,1989.69) and (456.74,1984.67) .. (462.93,1984.67) .. controls (469.12,1984.67) and (474.14,1989.69) .. (474.14,1995.88) .. controls (474.14,2002.07) and (469.12,2007.1) .. (462.93,2007.1) .. controls (456.74,2007.1) and (451.71,2002.07) .. (451.71,1995.88) -- cycle ;
\draw  [color={rgb, 255:red, 155; green, 155; blue, 155 }  ,draw opacity=0.5 ] (345,1995.88) .. controls (345,1989.69) and (350.02,1984.67) .. (356.21,1984.67) .. controls (362.41,1984.67) and (367.43,1989.69) .. (367.43,1995.88) .. controls (367.43,2002.07) and (362.41,2007.1) .. (356.21,2007.1) .. controls (350.02,2007.1) and (345,2002.07) .. (345,1995.88) -- cycle ;
\draw  [color={rgb, 255:red, 155; green, 155; blue, 155 }  ,draw opacity=0.5 ] (425.29,1995.88) .. controls (425.29,1989.69) and (430.31,1984.67) .. (436.5,1984.67) .. controls (442.69,1984.67) and (447.71,1989.69) .. (447.71,1995.88) .. controls (447.71,2002.07) and (442.69,2007.1) .. (436.5,2007.1) .. controls (430.31,2007.1) and (425.29,2002.07) .. (425.29,1995.88) -- cycle ;
\draw  [color={rgb, 255:red, 155; green, 155; blue, 155 }  ,draw opacity=0.5 ] (398.86,1995.88) .. controls (398.86,1989.69) and (403.88,1984.67) .. (410.07,1984.67) .. controls (416.26,1984.67) and (421.29,1989.69) .. (421.29,1995.88) .. controls (421.29,2002.07) and (416.26,2007.1) .. (410.07,2007.1) .. controls (403.88,2007.1) and (398.86,2002.07) .. (398.86,1995.88) -- cycle ;
\draw    (207.43,1979.07) -- (207.43,1889.16) ;
\draw [shift={(207.43,1886.16)}, rotate = 90] [fill={rgb, 255:red, 0; green, 0; blue, 0 }  ][line width=0.08]  [draw opacity=0] (8.93,-4.29) -- (0,0) -- (8.93,4.29) -- cycle    ;
\draw [shift={(207.43,1982.07)}, rotate = 270] [fill={rgb, 255:red, 0; green, 0; blue, 0 }  ][line width=0.08]  [draw opacity=0] (8.93,-4.29) -- (0,0) -- (8.93,4.29) -- cycle    ;
\draw    (220.43,2002.87) -- (220.43,1888.87) ;
\draw [shift={(220.43,1885.87)}, rotate = 90] [fill={rgb, 255:red, 0; green, 0; blue, 0 }  ][line width=0.08]  [draw opacity=0] (8.93,-4.29) -- (0,0) -- (8.93,4.29) -- cycle    ;
\draw [shift={(220.43,2005.87)}, rotate = 270] [fill={rgb, 255:red, 0; green, 0; blue, 0 }  ][line width=0.08]  [draw opacity=0] (8.93,-4.29) -- (0,0) -- (8.93,4.29) -- cycle    ;
\draw    (476.14,1858) -- (476.14,1879.08) ;
\draw [shift={(476.14,1882.08)}, rotate = 270] [fill={rgb, 255:red, 0; green, 0; blue, 0 }  ][line width=0.08]  [draw opacity=0] (8.93,-4.29) -- (0,0) -- (8.93,4.29) -- cycle    ;
\draw [shift={(476.14,1855)}, rotate = 90] [fill={rgb, 255:red, 0; green, 0; blue, 0 }  ][line width=0.08]  [draw opacity=0] (8.93,-4.29) -- (0,0) -- (8.93,4.29) -- cycle    ;
\draw    (476.14,1890.08) -- (476.14,1953.07) ;
\draw [shift={(476.14,1956.07)}, rotate = 270] [fill={rgb, 255:red, 0; green, 0; blue, 0 }  ][line width=0.08]  [draw opacity=0] (8.93,-4.29) -- (0,0) -- (8.93,4.29) -- cycle    ;
\draw [shift={(476.14,1887.08)}, rotate = 90] [fill={rgb, 255:red, 0; green, 0; blue, 0 }  ][line width=0.08]  [draw opacity=0] (8.93,-4.29) -- (0,0) -- (8.93,4.29) -- cycle    ;
\draw    (489.14,1858.08) -- (489.14,1953.07) ;
\draw [shift={(489.14,1956.07)}, rotate = 270] [fill={rgb, 255:red, 0; green, 0; blue, 0 }  ][line width=0.08]  [draw opacity=0] (8.93,-4.29) -- (0,0) -- (8.93,4.29) -- cycle    ;
\draw [shift={(489.14,1855.08)}, rotate = 90] [fill={rgb, 255:red, 0; green, 0; blue, 0 }  ][line width=0.08]  [draw opacity=0] (8.93,-4.29) -- (0,0) -- (8.93,4.29) -- cycle    ;
\draw  [fill={rgb, 255:red, 180; green, 207; blue, 242 }  ,fill opacity=1 ] (385.57,1971.79) .. controls (385.57,1965.59) and (390.59,1960.57) .. (396.79,1960.57) .. controls (402.98,1960.57) and (408,1965.59) .. (408,1971.79) .. controls (408,1977.98) and (402.98,1983) .. (396.79,1983) .. controls (390.59,1983) and (385.57,1977.98) .. (385.57,1971.79) -- cycle ;
\draw  [color={rgb, 255:red, 0; green, 0; blue, 0 }  ,draw opacity=1 ][fill={rgb, 255:red, 180; green, 207; blue, 242 }  ,fill opacity=1 ] (371.43,1995.88) .. controls (371.43,1989.69) and (376.45,1984.67) .. (382.64,1984.67) .. controls (388.84,1984.67) and (393.86,1989.69) .. (393.86,1995.88) .. controls (393.86,2002.07) and (388.84,2007.1) .. (382.64,2007.1) .. controls (376.45,2007.1) and (371.43,2002.07) .. (371.43,1995.88) -- cycle ;
\draw  [pattern=_wx0veqwwd,pattern size=11.100000000000001pt,pattern thickness=0.75pt,pattern radius=0pt, pattern color={rgb, 255:red, 0; green, 0; blue, 0}] (410.63,1946.79) -- (434.69,1884.13) -- (386.61,1884.16) -- cycle ;
\draw  [fill={rgb, 255:red, 224; green, 236; blue, 247 }  ,fill opacity=1 ] (399.43,1946.79) .. controls (399.43,1940.59) and (404.45,1935.57) .. (410.64,1935.57) .. controls (416.84,1935.57) and (421.86,1940.59) .. (421.86,1946.79) .. controls (421.86,1952.98) and (416.84,1958) .. (410.64,1958) .. controls (404.45,1958) and (399.43,1952.98) .. (399.43,1946.79) -- cycle ;

\draw (457,1864) node [anchor=north west][inner sep=0.75pt]  [font=\footnotesize]  {$\alpha $};
\draw (442,1906.02) node [anchor=north west][inner sep=0.75pt]  [font=\footnotesize]  {$k-\alpha $};
\draw (492,1896.86) node [anchor=north west][inner sep=0.75pt]  [font=\footnotesize]  {$k$};
\draw (405,1894.57) node [anchor=north west][inner sep=0.75pt]  [font=\footnotesize,color={rgb, 255:red, 0; green, 0; blue, 0 }  ,opacity=0.59 ]  {$u$};
\draw (125.4,1895.24) node [anchor=north west][inner sep=0.75pt]  [font=\footnotesize]  {$u$};
\draw (157,1918.1) node [anchor=north west][inner sep=0.75pt]  [font=\footnotesize]  {$k-\alpha +1$};
\draw (223,1932.1) node [anchor=north west][inner sep=0.75pt]  [font=\footnotesize]  {$k-\alpha +2$};
\draw (98,1992.4) node [anchor=north west][inner sep=0.75pt]  [font=\footnotesize]  {$x$};
\draw (112,1967.4) node [anchor=north west][inner sep=0.75pt]  [font=\footnotesize]  {$y$};
\draw (393,1967.4) node [anchor=north west][inner sep=0.75pt]  [font=\footnotesize]  {$y$};
\draw (378,1992.4) node [anchor=north west][inner sep=0.75pt]  [font=\footnotesize]  {$x$};
\draw (118,2016.4) node [anchor=north west][inner sep=0.75pt]   [align=left] {(a)};
\draw (398,2016.4) node [anchor=north west][inner sep=0.75pt]   [align=left] {(b)};

\end{tikzpicture}
    \caption{
    The first iteration of $\BFSSampleM$.   
    (a) The edge $(x,y)$ is a sampled edge from a dense region of $u$.
    (b) The run of $\DoubleBfsCycle$ from $x$ and $y$ leads to the removal of $u$.
    }
    \label{fig:BFSSampleMFirstIteration}
\end{figure}

\hfill
\vspace{-0.5cm}

\begin{figure}[H]
    \centering

 
\tikzset{
pattern size/.store in=\mcSize, 
pattern size = 5pt,
pattern thickness/.store in=\mcThickness, 
pattern thickness = 0.3pt,
pattern radius/.store in=\mcRadius, 
pattern radius = 1pt}
\makeatletter
\pgfutil@ifundefined{pgf@pattern@name@_slnfa9pij}{
\pgfdeclarepatternformonly[\mcThickness,\mcSize]{_slnfa9pij}
{\pgfqpoint{0pt}{-\mcThickness}}
{\pgfpoint{\mcSize}{\mcSize}}
{\pgfpoint{\mcSize}{\mcSize}}
{
\pgfsetcolor{\tikz@pattern@color}
\pgfsetlinewidth{\mcThickness}
\pgfpathmoveto{\pgfqpoint{0pt}{\mcSize}}
\pgfpathlineto{\pgfpoint{\mcSize+\mcThickness}{-\mcThickness}}
\pgfusepath{stroke}
}}
\makeatother

\begin{tikzpicture}[x=0.75pt,y=0.75pt,yscale=-1,xscale=1]

\draw [line width=2.25]    (396.64,2201.88) -- (410.79,2177.79) ;
\draw [line width=2.25]    (396.64,2201.88) -- (410.79,2177.79) ;
\draw  [color={rgb, 255:red, 0; green, 0; blue, 0 }  ,draw opacity=1 ][fill={rgb, 255:red, 180; green, 207; blue, 242 }  ,fill opacity=1 ] (385.43,2201.88) .. controls (385.43,2195.69) and (390.45,2190.67) .. (396.64,2190.67) .. controls (402.84,2190.67) and (407.86,2195.69) .. (407.86,2201.88) .. controls (407.86,2208.07) and (402.84,2213.1) .. (396.64,2213.1) .. controls (390.45,2213.1) and (385.43,2208.07) .. (385.43,2201.88) -- cycle ;
\draw [line width=2.25]    (117.64,2201.88) -- (131.79,2177.79) ;
\draw   (94.14,2177.79) .. controls (94.14,2171.59) and (99.16,2166.57) .. (105.36,2166.57) .. controls (111.55,2166.57) and (116.57,2171.59) .. (116.57,2177.79) .. controls (116.57,2183.98) and (111.55,2189) .. (105.36,2189) .. controls (99.16,2189) and (94.14,2183.98) .. (94.14,2177.79) -- cycle ;
\draw  [fill={rgb, 255:red, 180; green, 207; blue, 242 }  ,fill opacity=1 ] (120.57,2177.79) .. controls (120.57,2171.59) and (125.59,2166.57) .. (131.79,2166.57) .. controls (137.98,2166.57) and (143,2171.59) .. (143,2177.79) .. controls (143,2183.98) and (137.98,2189) .. (131.79,2189) .. controls (125.59,2189) and (120.57,2183.98) .. (120.57,2177.79) -- cycle ;
\draw  [color={rgb, 255:red, 0; green, 0; blue, 0 }  ,draw opacity=1 ][fill={rgb, 255:red, 180; green, 207; blue, 242 }  ,fill opacity=1 ] (106.43,2201.88) .. controls (106.43,2195.69) and (111.45,2190.67) .. (117.64,2190.67) .. controls (123.84,2190.67) and (128.86,2195.69) .. (128.86,2201.88) .. controls (128.86,2208.07) and (123.84,2213.1) .. (117.64,2213.1) .. controls (111.45,2213.1) and (106.43,2208.07) .. (106.43,2201.88) -- cycle ;
\draw  [color={rgb, 255:red, 0; green, 0; blue, 0 }  ,draw opacity=1 ] (132.86,2201.88) .. controls (132.86,2195.69) and (137.88,2190.67) .. (144.07,2190.67) .. controls (150.26,2190.67) and (155.29,2195.69) .. (155.29,2201.88) .. controls (155.29,2208.07) and (150.26,2213.1) .. (144.07,2213.1) .. controls (137.88,2213.1) and (132.86,2208.07) .. (132.86,2201.88) -- cycle ;
\draw  [fill={rgb, 255:red, 239; green, 187; blue, 106 }  ,fill opacity=1 ] (122,2128.79) .. controls (122,2122.59) and (127.02,2117.57) .. (133.21,2117.57) .. controls (139.41,2117.57) and (144.43,2122.59) .. (144.43,2128.79) .. controls (144.43,2134.98) and (139.41,2140) .. (133.21,2140) .. controls (127.02,2140) and (122,2134.98) .. (122,2128.79) -- cycle ;
\draw   (135.43,2152.79) .. controls (135.43,2146.59) and (140.45,2141.57) .. (146.64,2141.57) .. controls (152.84,2141.57) and (157.86,2146.59) .. (157.86,2152.79) .. controls (157.86,2158.98) and (152.84,2164) .. (146.64,2164) .. controls (140.45,2164) and (135.43,2158.98) .. (135.43,2152.79) -- cycle ;
\draw   (109,2152.79) .. controls (109,2146.59) and (114.02,2141.57) .. (120.21,2141.57) .. controls (126.41,2141.57) and (131.43,2146.59) .. (131.43,2152.79) .. controls (131.43,2158.98) and (126.41,2164) .. (120.21,2164) .. controls (114.02,2164) and (109,2158.98) .. (109,2152.79) -- cycle ;
\draw   (147,2177.8) .. controls (147,2171.6) and (152.02,2166.58) .. (158.21,2166.58) .. controls (164.41,2166.58) and (169.43,2171.6) .. (169.43,2177.8) .. controls (169.43,2183.99) and (164.41,2189.01) .. (158.21,2189.01) .. controls (152.02,2189.01) and (147,2183.99) .. (147,2177.8) -- cycle ;
\draw  [color={rgb, 255:red, 0; green, 0; blue, 0 }  ,draw opacity=0.53 ][fill={rgb, 255:red, 239; green, 187; blue, 106 }  ,fill opacity=0.6 ] (399,2129.8) .. controls (399,2123.6) and (404.02,2118.58) .. (410.21,2118.58) .. controls (416.41,2118.58) and (421.43,2123.6) .. (421.43,2129.8) .. controls (421.43,2135.99) and (416.41,2141.01) .. (410.21,2141.01) .. controls (404.02,2141.01) and (399,2135.99) .. (399,2129.8) -- cycle ;
\draw   (410.3,2178.79) -- (375.55,2088.07) -- (445.55,2088.26) -- cycle ;
\draw  [dash pattern={on 1.5pt off 1.5pt on 1.5pt off 1.5pt}]  (539.14,2116.09) -- (353.99,2116.09) ;
\draw  [color={rgb, 255:red, 155; green, 155; blue, 155 }  ,draw opacity=0.5 ] (412.43,2153.8) .. controls (412.43,2147.6) and (417.45,2142.58) .. (423.64,2142.58) .. controls (429.84,2142.58) and (434.86,2147.6) .. (434.86,2153.8) .. controls (434.86,2159.99) and (429.84,2165.01) .. (423.64,2165.01) .. controls (417.45,2165.01) and (412.43,2159.99) .. (412.43,2153.8) -- cycle ;
\draw  [color={rgb, 255:red, 155; green, 155; blue, 155 }  ,draw opacity=0.5 ] (386,2153.8) .. controls (386,2147.6) and (391.02,2142.58) .. (397.21,2142.58) .. controls (403.41,2142.58) and (408.43,2147.6) .. (408.43,2153.8) .. controls (408.43,2159.99) and (403.41,2165.01) .. (397.21,2165.01) .. controls (391.02,2165.01) and (386,2159.99) .. (386,2153.8) -- cycle ;
\draw  [color={rgb, 255:red, 155; green, 155; blue, 155 }  ,draw opacity=0.5 ] (425.86,2178.8) .. controls (425.86,2172.6) and (430.88,2167.58) .. (437.07,2167.58) .. controls (443.26,2167.58) and (448.29,2172.6) .. (448.29,2178.8) .. controls (448.29,2184.99) and (443.26,2190.01) .. (437.07,2190.01) .. controls (430.88,2190.01) and (425.86,2184.99) .. (425.86,2178.8) -- cycle ;
\draw  [color={rgb, 255:red, 155; green, 155; blue, 155 }  ,draw opacity=0.5 ] (373.14,2177.8) .. controls (373.14,2171.6) and (378.16,2166.58) .. (384.36,2166.58) .. controls (390.55,2166.58) and (395.57,2171.6) .. (395.57,2177.8) .. controls (395.57,2183.99) and (390.55,2189.01) .. (384.36,2189.01) .. controls (378.16,2189.01) and (373.14,2183.99) .. (373.14,2177.8) -- cycle ;
\draw  [color={rgb, 255:red, 0; green, 0; blue, 0 }  ,draw opacity=1 ] (79,2201.89) .. controls (79,2195.7) and (84.02,2190.68) .. (90.21,2190.68) .. controls (96.41,2190.68) and (101.43,2195.7) .. (101.43,2201.89) .. controls (101.43,2208.08) and (96.41,2213.1) .. (90.21,2213.1) .. controls (84.02,2213.1) and (79,2208.08) .. (79,2201.89) -- cycle ;
\draw  [color={rgb, 255:red, 0; green, 0; blue, 0 }  ,draw opacity=1 ] (159.29,2201.89) .. controls (159.29,2195.7) and (164.31,2190.68) .. (170.5,2190.68) .. controls (176.69,2190.68) and (181.71,2195.7) .. (181.71,2201.89) .. controls (181.71,2208.08) and (176.69,2213.1) .. (170.5,2213.1) .. controls (164.31,2213.1) and (159.29,2208.08) .. (159.29,2201.89) -- cycle ;
\draw  [color={rgb, 255:red, 155; green, 155; blue, 155 }  ,draw opacity=0.5 ] (359,2201.89) .. controls (359,2195.7) and (364.02,2190.68) .. (370.21,2190.68) .. controls (376.41,2190.68) and (381.43,2195.7) .. (381.43,2201.89) .. controls (381.43,2208.08) and (376.41,2213.1) .. (370.21,2213.1) .. controls (364.02,2213.1) and (359,2208.08) .. (359,2201.89) -- cycle ;
\draw  [color={rgb, 255:red, 155; green, 155; blue, 155 }  ,draw opacity=0.5 ] (439.29,2201.89) .. controls (439.29,2195.7) and (444.31,2190.68) .. (450.5,2190.68) .. controls (456.69,2190.68) and (461.71,2195.7) .. (461.71,2201.89) .. controls (461.71,2208.08) and (456.69,2213.1) .. (450.5,2213.1) .. controls (444.31,2213.1) and (439.29,2208.08) .. (439.29,2201.89) -- cycle ;
\draw  [color={rgb, 255:red, 155; green, 155; blue, 155 }  ,draw opacity=0.5 ] (412.86,2201.89) .. controls (412.86,2195.7) and (417.88,2190.68) .. (424.07,2190.68) .. controls (430.26,2190.68) and (435.29,2195.7) .. (435.29,2201.89) .. controls (435.29,2208.08) and (430.26,2213.1) .. (424.07,2213.1) .. controls (417.88,2213.1) and (412.86,2208.08) .. (412.86,2201.89) -- cycle ;
\draw    (204.43,2184.24) -- (204.43,2116.17) ;
\draw [shift={(204.43,2113.17)}, rotate = 90] [fill={rgb, 255:red, 0; green, 0; blue, 0 }  ][line width=0.08]  [draw opacity=0] (8.93,-4.29) -- (0,0) -- (8.93,4.29) -- cycle    ;
\draw [shift={(204.43,2187.24)}, rotate = 270] [fill={rgb, 255:red, 0; green, 0; blue, 0 }  ][line width=0.08]  [draw opacity=0] (8.93,-4.29) -- (0,0) -- (8.93,4.29) -- cycle    ;
\draw    (217.43,2206.22) -- (217.43,2115.94) ;
\draw [shift={(217.43,2112.94)}, rotate = 90] [fill={rgb, 255:red, 0; green, 0; blue, 0 }  ][line width=0.08]  [draw opacity=0] (8.93,-4.29) -- (0,0) -- (8.93,4.29) -- cycle    ;
\draw [shift={(217.43,2209.22)}, rotate = 270] [fill={rgb, 255:red, 0; green, 0; blue, 0 }  ][line width=0.08]  [draw opacity=0] (8.93,-4.29) -- (0,0) -- (8.93,4.29) -- cycle    ;
\draw    (504.14,2090) -- (504.14,2111.09) ;
\draw [shift={(504.14,2114.09)}, rotate = 270] [fill={rgb, 255:red, 0; green, 0; blue, 0 }  ][line width=0.08]  [draw opacity=0] (8.93,-4.29) -- (0,0) -- (8.93,4.29) -- cycle    ;
\draw [shift={(504.14,2087)}, rotate = 90] [fill={rgb, 255:red, 0; green, 0; blue, 0 }  ][line width=0.08]  [draw opacity=0] (8.93,-4.29) -- (0,0) -- (8.93,4.29) -- cycle    ;
\draw    (504.14,2122.09) -- (504.14,2185.08) ;
\draw [shift={(504.14,2188.08)}, rotate = 270] [fill={rgb, 255:red, 0; green, 0; blue, 0 }  ][line width=0.08]  [draw opacity=0] (8.93,-4.29) -- (0,0) -- (8.93,4.29) -- cycle    ;
\draw [shift={(504.14,2119.09)}, rotate = 90] [fill={rgb, 255:red, 0; green, 0; blue, 0 }  ][line width=0.08]  [draw opacity=0] (8.93,-4.29) -- (0,0) -- (8.93,4.29) -- cycle    ;
\draw    (517.14,2090.09) -- (517.14,2185.08) ;
\draw [shift={(517.14,2188.08)}, rotate = 270] [fill={rgb, 255:red, 0; green, 0; blue, 0 }  ][line width=0.08]  [draw opacity=0] (8.93,-4.29) -- (0,0) -- (8.93,4.29) -- cycle    ;
\draw [shift={(517.14,2087.09)}, rotate = 90] [fill={rgb, 255:red, 0; green, 0; blue, 0 }  ][line width=0.08]  [draw opacity=0] (8.93,-4.29) -- (0,0) -- (8.93,4.29) -- cycle    ;
\draw  [pattern=_slnfa9pij,pattern size=11.100000000000001pt,pattern thickness=0.75pt,pattern radius=0pt, pattern color={rgb, 255:red, 0; green, 0; blue, 0}] (410.63,2178.8) -- (434.69,2116.14) -- (386.61,2116.17) -- cycle ;
\draw  [fill={rgb, 255:red, 180; green, 207; blue, 242 }  ,fill opacity=1 ] (399.57,2177.79) .. controls (399.57,2171.59) and (404.59,2166.57) .. (410.79,2166.57) .. controls (416.98,2166.57) and (422,2171.59) .. (422,2177.79) .. controls (422,2183.98) and (416.98,2189) .. (410.79,2189) .. controls (404.59,2189) and (399.57,2183.98) .. (399.57,2177.79) -- cycle ;

\draw (485,2096) node [anchor=north west][inner sep=0.75pt]  [font=\footnotesize]  {$\alpha $};
\draw (448.64,2133.58) node [anchor=north west][inner sep=0.75pt]  [font=\footnotesize]  {$ \begin{array}{l}
k_{i} -\alpha =\\
k_{i+1} -1
\end{array}$};
\draw (521,2125.87) node [anchor=north west][inner sep=0.75pt]  [font=\footnotesize]  {$k_{i}$};
\draw (406,2126.58) node [anchor=north west][inner sep=0.75pt]  [font=\footnotesize,color={rgb, 255:red, 0; green, 0; blue, 0 }  ,opacity=0.59 ]  {$u$};
\draw (158,2128.19) node [anchor=north west][inner sep=0.75pt]  [font=\footnotesize]  {$k_{i+1} -1$};
\draw (221,2147.15) node [anchor=north west][inner sep=0.75pt]  [font=\footnotesize]  {$k_{i+1}$};
\draw (128.4,2126.24) node [anchor=north west][inner sep=0.75pt]  [font=\footnotesize]  {$u$};
\draw (113,2198.4) node [anchor=north west][inner sep=0.75pt]  [font=\footnotesize]  {$x$};
\draw (127,2173.4) node [anchor=north west][inner sep=0.75pt]  [font=\footnotesize]  {$y$};
\draw (119,2222.4) node [anchor=north west][inner sep=0.75pt]   [align=left] {(a)};
\draw (392,2198.4) node [anchor=north west][inner sep=0.75pt]  [font=\footnotesize]  {$x$};
\draw (398,2222.4) node [anchor=north west][inner sep=0.75pt]   [align=left] {(b)};
\draw (407,2173.4) node [anchor=north west][inner sep=0.75pt]  [font=\footnotesize]  {$y$};

\end{tikzpicture}
  
    \caption{
    Iteration $i>1$ of $\BFSSampleM$.
    (a) The edge $(x,y)$ is a sampled edge from a dense region of $u$.
    (b) The run of $\BallOrCycle$ from $x$ and $y$ leads to the removal of $u$.
    }
    \label{fig:BFSSampleMIteration}
\end{figure}

\end{figure}